%% file: main.tex
\documentclass[11pt,letterpaper]{article}

\usepackage[breaklinks, pagebackref=true, urlcolor=blue, colorlinks, citecolor=green!50!black, linkcolor=blue,bookmarksopen,bookmarksnumbered,linktocpage=true]{hyperref}

\usepackage[letterpaper, left=1in, right=1in, top=0.9in, bottom=0.9in]{geometry}
\usepackage[utf8]{inputenc}
\usepackage[american]{babel}
\usepackage[normalem]{ulem}
\usepackage{amsmath, amssymb, cases, amsthm}
\usepackage{thmtools}
\usepackage[shortlabels]{enumitem}
\usepackage[framemethod=TikZ]{mdframed}
\usepackage{bbm}
\usepackage{bm}
\usepackage{microtype}
\usepackage{xcolor}
\usepackage{makecell}
\usepackage{mathtools}
\usepackage{float}
\usepackage{modletters}
\usepackage{comment}
\usepackage{multirow}
\usepackage[sort]{cite}
\usepackage[capitalize,noabbrev]{cleveref}
\usepackage{graphics}
\usepackage{xspace}
\usepackage{pdfsync}
\usepackage{tocloft}
\usepackage{footnotehyper}
\makesavenoteenv{tabular}

\usepackage{pifont}

\usepackage{mleftright}

\declaretheorem[numberwithin=section,refname={Theorem,Theorems},Refname={Theorem,Theorems}]{theorem}

\declaretheorem[numberlike=theorem]{lemma}
\declaretheorem[numberlike=theorem]{proposition}
\declaretheorem[numberlike=theorem]{corollary}
\declaretheorem[numberlike=theorem,style=definition]{definition}
\declaretheorem[numberlike=theorem, refname={Claim, Claims}, Refname={Claim, Claims}]{claim}
\declaretheorem[numberlike=theorem,style=remark]{remark}
\declaretheorem[numberlike=theorem,refname={Fact,Facts},Refname={Fact,Facts},name={Fact}]{fact}

\declaretheorem[numberlike=theorem, refname={Observation,Observations},Refname={Observation,Observations},name={Observation}]{observation}

\def\final{0}  
\def\iflong{\iffalse}
\ifnum\final=0  

\newcommand{\tawei}[1]{{\color{blue} [{\bf Ta-Wei:} #1]}}
\newcommand{\parth}[1]{{\color{purple} [{\bf Parth:} #1]}}
\newcommand{\TODO}[1]{{\color{blue!50!black} [{\bf Todo:} #1]}}
\else 
\newcommand{\yonggang}[1]{}
\newcommand{\danupon}[1]{}
\newcommand{\sagnik}[1]{}
\newcommand{\todo}[1]{}
\newcommand{\yuval}[1]{}
\newcommand{\jan}[1]{}
\newcommand{\blikstad}[1]{}
\newcommand{\tawei}[1]{}
\newcommand{\parth}[1]{}
\newcommand{\TODO}[1]{}
\fi  

\DeclarePairedDelimiter{\ceil}{\lceil}{\rceil}
\newcommand{\eps}{\varepsilon}

\DeclareMathOperator*{\var}{Var}

\newcommand{\set}[2][ ]{\{#2 \ifthenelse{\equal{#1}{ }}{ }{~|~#1}\}}

\newcommand{\HPC}{\mathsf{HPC}}

\DeclareMathOperator*{\E}{\mathbb{E}}

\usepackage[most]{tcolorbox}
\tcbuselibrary{skins,breakable}
\tcbset{enhanced jigsaw}

\usepackage{xparse}
\usepackage{lipsum}

\makeatletter
\newcommand\footnoteref[1]{\protected@xdef\@thefnmark{\ref{#1}}\@footnotemark}
\makeatother

\renewcommand{\paragraph}[1]{\medskip\noindent{\bf #1}\xspace}

\usepackage{tikz}
\usetikzlibrary{arrows}
\usetikzlibrary{arrows.meta}
\usetikzlibrary{shapes}
\usetikzlibrary{backgrounds}
\usetikzlibrary{positioning}
\usetikzlibrary{decorations.markings}
\usetikzlibrary{patterns}
\usetikzlibrary{calc}
\usetikzlibrary{fit}
\usetikzlibrary{decorations}

\renewcommand{\geq}{\geqslant}
\renewcommand{\leq}{\leqslant}
\renewcommand{\le}{\leq}
\renewcommand{\ge}{\geq}

\theoremstyle{definition}
\newtheorem{algocf}{Algorithm}
\newenvironment{Algorithm}{\begin{abox}\begin{algocf}}{\end{algocf}\end{abox}}

\Crefname{algocf}{Algorithm}{Algorithms}

\newenvironment{abox}{\begin{tcolorbox}[
		enlarge top by=5pt,
		enlarge bottom by=5pt,
		breakable,
		frame hidden,
		overlay broken = {
			\draw[line width=1pt, black]
			(frame.north west) rectangle (frame.south east);},
		overlay = {
			\draw[line width=1pt, black]
			(frame.north west) rectangle (frame.south east);},
		boxsep=0pt,
		left=4pt,
		right=4pt,
		top=10pt,
		arc=0pt,
		boxrule=1pt,toprule=1pt,
		colback=white
	]}
{\end{tcolorbox}}

\newcommand{\degen}{\textnormal{\textsf{DEGEN}}}

\newcommand{\fail}{\textsc{Reject}}
\newcommand{\true}{\textsc{Accept}}

\providecommand{\given}

\DeclarePairedDelimiter{\pparen}()
\DeclarePairedDelimiter{\card}\vert\vert

\DeclarePairedDelimiterX{\brac}[1][]{
\ifblank{#1}{\,\cdot\,}{#1}}

\DeclarePairedDelimiterX{\Set}[1]\{\}{
  \renewcommand\given{\nonscript\: \delimsize\vert{} \nonscript\: \mathopen{}}
\ifblank{#1}{\,\cdot\,}{#1}}

\renewcommand{\given}{\nonscript\: \vert \nonscript\:}

\NewDocumentCommand{\Prob}{sO{}E{_}{{}}m}{%
    \Pr_{#3}
    \IfBooleanTF{#1}
    {\brac*{#4}}
    {\brac[#2]{#4}}
}

\DeclarePairedDelimiterXPP{\Ot}[1]{\widetilde{O}}(){}{#1}
\DeclarePairedDelimiterXPP{\Omgt}[1]{\widetilde{\Omega}}(){}{#1}

\DeclarePairedDelimiterXPP{\Var}[1]{\var}[]{}{#1}

\newcommand{\dsi}{\cD_\SI}
\newcommand{\pisi}{\pi_\SI}

\newcommand{\ptd}{\Lambda}
\DeclarePairedDelimiterXPP{\PTD}[2]{\Lambda}(){}{#1, #2}

\NewDocumentCommand{\Dist}{sO{}E{_}{{}}m}{%
  \mu_{#3}
  \IfBooleanTF{#1}
  {\pparen*{#4}}
  {\pparen[#2]{#4}}
}

\DeclareMathOperator*{\expect}{\mathbb{E}}
\NewDocumentCommand{\Exp}{sO{}E{_}{{}}m}{%
  \expect_{#3}
  \IfBooleanTF{#1}
  {\brac*{#4}}
  {\brac[#2]{#4}}
}

\newcommand{\re}{\rv{e}}

\newcommand{\estar}{e^{\ast}}
\newcommand{\restar}{\re^{\ast}}

\newcommand{\rPi}{\rv{\Pi}}

\newcommand{\Perm}{\Gamma}
\newcommand{\rPerm}{\rv{\Gamma}}

\newcommand{\uni}{\mathcal{U}}

\newcommand{\qpi}{q^{\Pi}}

\newcommand{\score}{\mathsf{Score}}
\newcommand{\totscore}{\mathsf{TScore}}

\DeclareMathOperator*{\inform}{\mathbb{I}}
\DeclarePairedDelimiterXPP{\Inf}[1]{\inform}(){}{#1}

\DeclareMathOperator*{\entropy}{\mathbb{H}}
\DeclarePairedDelimiterXPP{\Ent}[1]{\entropy}(){}{#1}

\DeclareMathOperator{\myIC}{\mathsf{IC}}
\DeclareMathOperator{\myCC}{\mathsf{CC}}

\newcommand{\IntIC}{\myIC^{\textnormal{int}}}

\newcommand{\rPTD}{\mathbf{\Lambda}}

\newcommand{\bhpc}{\ensuremath{\textnormal{\textsf{BHPC}}}\xspace}
\newcommand{\bmhpc}{\ensuremath{\textnormal{\textsf{BMHPC}}}\xspace}
\newcommand{\hpc}{\ensuremath{\textnormal{\textsf{HPC}}}\xspace}
\newcommand{\mhpc}{\ensuremath{\textnormal{\textsf{MHPC}}}\xspace}
\newcommand{\setint}{\ensuremath{\textnormal{\textsf{SetInt}}}\xspace}
\newcommand{\defc}{\mathtt{def}}
\newcommand{\odeg}{\text{odeg}}
\newcommand{\Vaux}{V_{\text{aux}}}
\newcommand{\Eaux}{E_{\text{aux}}}

\newcommand{\mypar}[1]{\smallskip\noindent{\sffamily\bfseries #1.}~}
\newcommand{\floor}[1]{\lfloor #1 \rfloor}

\newcommand{\Bracket}[1]{\Big[#1\Big]}
\newcommand{\bracket}[1]{\mleft[#1\mright]}
\newcommand{\paren}[1]{\ensuremath{\mleft(#1\mright)}\xspace}

\newenvironment{tbox}{\begin{tcolorbox}[
		enlarge top by=5pt,
		enlarge bottom by=5pt,
		 breakable,
		 boxsep=0pt,
                  left=4pt,
                  right=4pt,
                  top=10pt,
                  boxrule=1pt,toprule=1pt,
                  colback=white,
                  arc=-1pt,
                  ]
	}
{\end{tcolorbox}}

\newcommand{\Prot}{\ensuremath{\Pi}}
\newcommand{\prot}{\ensuremath{\pi}}

\newcommand{\bA}{\bm{A}}

\newcommand{\bR}{\bm{R}}

\newcommand{\bB}{\ensuremath{\bm{B}}}
\newcommand{\bC}{\ensuremath{\bm{C}}}
\newcommand{\bD}{\ensuremath{\bm{D}}}

\newcommand{\supp}[1]{\ensuremath{\textsc{supp}(#1)}}

\newtheorem{mdresult}{Result}
\newenvironment{result}{\begin{mdframed}[backgroundcolor=lightgray!40,topline=false,rightline=false,leftline=false,bottomline=false,innertopmargin=2pt]\begin{mdresult}}{\end{mdresult}\end{mdframed}}

\newcommand{\rA}{\rv{A}}
\newcommand{\rB}{\rv{B}}
\newcommand{\rC}{\rv{C}}
\newcommand{\rD}{\rv{D}}

\newcommand{\rX}{\rv{X}}
\newcommand{\rY}{\rv{Y}}
\newcommand{\rI}{\rv{I}}

\newcommand{\rv}[1]{\ensuremath{\mathsf{#1}}}

\newcommand{\II}{\ensuremath{\mathbb{I}}}
\newcommand{\HH}{\ensuremath{\mathbb{H}}}
\newcommand{\mi}[2]{\ensuremath{\II(#1 \,; #2)}}
\newcommand{\en}[1]{\ensuremath{\HH(#1)}}

\newcommand{\itfacts}[1]{Fact~\ref{fact:it-facts}-(\ref{part:#1})\xspace}

\newcommand{\distribution}[1]{\ensuremath{\textnormal{dist}(#1)}\xspace}

\newcommand{\protHPC}{\ensuremath{\prot_{\HPC}}}

\newcommand{\rZ}{\rv{Z}}

\newcommand{\rT}{\rv{T}}
\newcommand{\rE}{\rv{E}}

\newcommand{\rProt}{\rv{\Prot}}

\newcommand{\protSI}{\ensuremath{\prot}_{\ensuremath{\textnormal{\textsf{SI}}}}\xspace}

\newcommand{\rProtSI}{\ensuremath{\rv{\Prot}}_{\ensuremath{\textnormal{\textsf{SI}}}}\xspace}

\newcommand{\SI}{\ensuremath{\textnormal{\textsf{Set-Int}}}\xspace}

\newcommand{\rS}{\rv{S}}

\newcommand{\rR}{\rv{R}}

\newcommand{\mii}[3]{\ensuremath{\II_{#3}(#1 \,; #2)}}

\newcommand{\hd}[2]{\ensuremath{\textnormal{h}(#1,#2)}\xspace}

\newcommand{\CC}[2]{\ensuremath{\textnormal{\textsf{CC}}_{#2}(#1)}\xspace}
\newcommand{\ICost}[2]{\ensuremath{\textnormal{\textsf{IC}}_{#2}(#1)}\xspace}
\newcommand{\IC}[2]{\ICost{#1}{#2}}

\newcommand{\rTheta}{\rv{\Theta}}

\title{Polynomial Pass Semi-Streaming Lower Bounds \\ for  K-Cores and Degeneracy}
\author{}
\author{Sepehr Assadi\thanks{(\texttt{sepehr@assadi.info}) 
Cheriton School of Computer Science, University of Waterloo, Canada, and Department of Computer Science, Rutgers University, USA.}
\and Prantar Ghosh\thanks{(\texttt{prantar.ghosh@gmail.com}) Department of Computer Science, Georgetown University, USA. Part of this work was done while the author was at DIMACS, Rutgers University, USA}
\and Bruno Loff\thanks{(\texttt{bruno.loff@gmail.com}) Department of Mathematics and LASIGE, Faculdade de Ci\^encias, Universidade de Lisboa, Portugal.}
\and Parth Mittal\thanks{(\texttt{parth.mittal@uwaterloo.ca}) University of Waterloo, Canada.}
\and Sagnik Mukhopadhyay\thanks{(\texttt{schwagznikst@gmail.com}) University of Sheffield, UK.} }
\date{}

\begin{document}
	
	\begin{titlepage}
		\maketitle \pagenumbering{roman}
		
		\input{abstract}

		\setcounter{tocdepth}{3}
		\newpage
		\tableofcontents
		\newpage
	\end{titlepage}
	
	\newpage
	\pagenumbering{arabic}

\input{introduction}

\input{overview}

\input{preliminaries}

\input{lowerbounds}

\input{reduction}

\input{upper}

\addcontentsline{toc}{section}{Acknowledgements}
\section*{Acknowledgements} 
The first named author would like to thank Yu Chen and Sanjeev Khanna for their collaboration in~\cite{AssadiCK19} that was the starting point of this project
and Madhu Sudan for helpful conversations.  

Sepehr Assadi is supported in part by a  Sloan Research Fellowship, an NSERC Discovery Grant, a University of Waterloo startup grant, and a Faculty of Math Research Chair grant. 

Prantar Ghosh is supported in part by NSF under award 1918989. Part of this work was done while he was at DIMACS, Rutgers University, supported in part by a grant (820931) to DIMACS from the Simons Foundation.

Bruno Loff is funded by the European Union (ERC, HOFGA, 101041696). Views and opinions expressed are however those of the author(s) only and do not necessarily reflect those of the European Union or the European Research Council. Neither the European Union nor the granting authority can be held responsible for them. He was also supported by FCT through the LASIGE Research Unit, ref.\ UIDB/00408/2020 and ref.\ UIDP/00408/2020, and by CMAFcIO, FCT Project UIDB/04561/2020, \url{https://doi.org/10.54499/UIDB/04561/2020}.

Parth Mittal is supported by a David R. Cheriton Graduate Scholarship and  Sepehr Assadi’s Sloan Research Fellowship, an NSERC Discovery grant, and a startup grant from
University of Waterloo.

Sagnik Mukhopadhyay is partially funded by the UKRI New Investigator Award, ref.\ \href{https://gow.epsrc.ukri.org/NGBOViewGrant.aspx?GrantRef=EP/X03805X/1}{EP/X03805X/1}.

\addcontentsline{toc}{section}{Bibliography}
\bibliography{Bibliography,dump}

\clearpage

\input{appendix}

\end{document}

%% file: abstract.tex
\begin{abstract}

\noindent

The following question arises naturally in the study of graph streaming algorithms: 

\begin{quote}
\textit{Is there any graph problem which is ``not too hard'', in that it can be solved efficiently with total communication (nearly) linear in the number $n$ of vertices, and for which, nonetheless, any streaming algorithm with $\Ot{n}$ space (i.e., a semi-streaming algorithm) needs a polynomial $n^{\Omega(1)}$ number of passes?}
\end{quote}
Assadi, Chen, and Khanna [STOC 2019] were the first to prove that this is indeed the case. However, the lower bounds that they obtained are for rather non-standard graph problems. 

\medskip
Our first main contribution is to present the first polynomial-pass lower bounds for natural ``not too hard'' graph problems studied previously in the streaming model: \textbf{$k$-cores} and \textbf{degeneracy}. 
We devise a novel communication protocol for both problems with near-linear communication, thus showing that $k$-cores and degeneracy are natural examples of ``not too hard'' problems.
Indeed, previous work have developed single-pass 
semi-streaming algorithms for approximating these problems. In contrast, we prove that any semi-streaming algorithm for \emph{exactly} solving these problems requires (almost) $\Omega(n^{1/3})$ passes. 

\medskip
The lower bound follows by a reduction from a generalization of the \textbf{hidden pointer chasing (HPC)} problem of Assadi, Chen, and Khanna, which is also the basis of their earlier semi-streaming lower bounds. 

\medskip

Our second main contribution is improved \textbf{round-communication} lower bounds for the underlying communication problems at 
the basis of these reductions: 
\begin{itemize}
\item We improve the previous  lower bound of Assadi, Chen, and Khanna for HPC to achieve optimal bounds for this problem. 
\item We further observe that all current reductions from HPC can also work with a generalized version of this problem that we call \textbf{MultiHPC}, 
and prove an even stronger and optimal lower bound for this generalization. 
\end{itemize}
These two results collectively allow us to improve the resulting pass lower bounds for semi-streaming algorithms by a polynomial factor, namely, from $n^{1/5}$ to $n^{1/3}$ passes.

\end{abstract}

%% file: introduction.tex
\section{Introduction} \label{sec:intro}

Graph streaming algorithms process their inputs by making one or few passes
over the edges of an input graph using limited memory.
Algorithms that use space proportional to $n$, the number of vertices, are
called \emph{semi-streaming} algorithms.
Since their introduction by~\cite{FeigenbaumKMSZ05}, graph streaming algorithms
have become one of the main theoretical research areas on  processing massive
graphs.
We refer the interested
reader to~\cite{McGregor14} for an introductory survey of earlier results on
this topic. 

In this work, we prove a polynomial-pass lower bound for any graph streaming algorithm
that computes $k$-cores or degeneracy of a given graph.  Our result is of interest from the point of view of proving strong lower bounds in the graph streaming model in addition to their direct implications 
for these two specific problems. 

\subsection{Polynomial Pass Lower Bounds in Graph Streams}\label{sec:poly-pass}

Even though the study of multi-pass graph streaming algorithms started hand in hand with single-pass algorithms in~\cite{FeigenbaumKMSZ05}, our understanding of 
powers and limitations of multi-pass algorithms, even for most basic problems, lags considerably behind. On one hand, for a problem like minimum cut, 
we have algorithms that in just $\Ot{n}$ space and two passes can solve the problem \emph{exactly}~\cite{AssadiD21}\footnote{See~\cite{RubinsteinSW18} for an implicit algorithm with the same bounds and~\cite{MukhopadhyayN20} for the extension to 
weighted cuts in $O(\log{n})$ passes.} (see~\cite[Table 1]{AssadiCK19} for a list of several such results). On the other hand, for some other basic problems such as undirected shortest path, directed reachability, and bipartite matching, the best known semi-streaming algorithms require $O({n^{1/2}})$~\cite{ChangFHT20}, 
$n^{1/2+o(1)}$~\cite{LiuJS19,AssadiJJST22}, and $n^{3/4+o(1)}$~\cite{AssadiJJST22} passes, respectively; yet, despite significant efforts, the best lower bound for any of these problems
is still (even slightly below) $\Omega(\log{n})$ passes~\cite{GuruswamiO13,ChakrabartiGMV20,AssadiR20,ChenKPSSY21}. 

A key reason behind our  weaker understanding of multi-pass streaming algorithms can be attributed to the lack of techniques for 
proving  \emph{super-logarithmic} pass lower bounds for semi-streaming algorithms. At this point, such lower bounds
are only known for a handful of problems: clique and independent set~\cite{HalldorssonSSW12}, dominating set~\cite{Assadi17sc}, Hamiltonian path~\cite{BachrachCDELP19}, maximum cut~\cite{BachrachCDELP19,KolPSY23}, vertex cover and coloring~\cite{AbboudCKP21}, exact Boolean CSPs~\cite{KolPSY23}, triangle detection~\cite{PapadimitriouS84,Bar-YossefKS02}, and diameter computation~\cite{FrischknechtHW12}. Although, for all these problems, we can 
actually prove  close-to-$n$ pass lower bounds. Let us examine this dichotomy. 

A quick glance at the list of problems above may suggest an intuitive difference between these problems and the ones like reachability or shortest path: the above list consists of problems 
that are computationally hard in a classical sense\footnote{These are  standard NP-hard problems or admit some fine-grained hardness (for the latter two)~\cite{RodittyW13,KopelowitzPP16}.}, 
suggesting that we are dealing with a ``harder'' class of problems in their case. While this intuition should not be taken as a formal evidence---as classical computational hardness does \emph{not} imply 
streaming lower bounds (which are unconditional and information-theoretic)---\cite{AssadiCK19} showed that one can also formally explain this dichotomy.

In particular, \cite{AssadiCK19} observed that these strong streaming lower bounds happen only when the communication complexity of the problem at hand
 is $\Omega(n^2)$. Such a high lower-bound on the communication complexity \textit{immediately} gives an $\Omgt{n}$-pass lower bound for semi-streaming algorithms via standard reductions. Whereas, for almost all problems of interest in the semi-streaming model, including shortest path, reachability, and bipartite matching, 
we already know an $\Ot{n}$ communication upper bound\footnote{This perhaps can be seen as this: a problem whose (unbounded round) communication complexity is already high has almost no place in the streaming
model, which is a much weaker model algorithmically.} (the protocol for bipartite matching was only discovered in~\cite{blikstadBEMN22} after the work of~\cite{AssadiCK19}, but $\Ot{n^{3/2}}$ communication protocols were known already~\cite{IvanyosKLSW12,DobzinskiNO19}). We refer the reader to~\cite[Section 1.1]{AssadiCK19} for more context regarding these observations and prior techniques for $o(\log{n})$ pass lower bounds. 

\subsubsection*{Toward Stronger Streaming Lower Bounds} 
A natural question in light of these observations, already posed in~\cite{AssadiCK19}, is the following: 

\begin{quote}
\textbf{Motivating question.}
\textit{Is there any graph problem which is ``not too hard'', in that it can be solved efficiently with communication (nearly) linear in the number $n$ of vertices, and for which, nonetheless, any semi-streaming algorithm needs a polynomial $n^{\Omega(1)}$ number of passes?}
\end{quote}
To address this question,~\cite{AssadiCK19} introduced a new (four-player) communication problem called \textbf{Hidden Pointer Chasing (HPC)}, 
which acts as a cross between \emph{Set-Intersection} and \emph{Pointer Chasing} problems, which are the main problems for, respectively, proving $\Omega(n^2)$ communication lower bounds on graphs, and $o(\log{n})$-pass lower
bounds for semi-streaming algorithms. 

Roughly speaking, the HPC problem is defined as follows. There are four players paired into two groups. Each pair
of players inside a group shares $m$ instances of the Set-Intersection problem on $m$ elements (each of the two players holds a subset of $[m]$ and they need to identify the unique intersecting element). The
intersecting element in each instance of each group ``points'' to an instance in the other group. The
goal is to start from a fixed instance, follow these pointers for a fixed number of steps, and then return the last element reached. See~\Cref{prob:hpc} for the formal description. 

This problem admits an efficient communication protocol with no limit on its number of rounds, but~\cite{AssadiCK19} showed that any $r$-round protocol
that aims to find the $(r+1)$-th pointer in HPC requires $\Omega(m^2/r^2)$ communication. This places HPC 
squarely in the middle of previous techniques and quite suitable for performing reductions to prove streaming lower bounds
even for not-too-hard graph problems. Using this,~\cite{AssadiCK19} proved the first set of polynomial-pass graph streaming lower bounds
in this class of problems: computing \emph{lexicographically-first maximal independent set (LFMIS)} and \emph{$s$-$t$ minimum cut} on graphs with exponential edge-capacities both require $\Omgt{n^{1/5}}$ passes to be solved by semi-streaming algorithms. 

Despite the significant advances on multi-pass streaming lower bounds in the last couple of years~\cite{ChakrabartiGMV20,AssadiKSY20,AssadiR20,ChenKPSSY21,AssadiN21,ChenKPSSY21b,Assadi22,AssadiKZ22,KolPSY23,ChenKPSSY23}, 
there is still no other known (not-too-hard) problem that admits a polynomial-pass lower bounds beside those of~\cite{AssadiCK19}. In addition, it is worth mentioning that, strictly speaking, neither LFMIS nor the version of $s$-$t$ minimum cut in~\cite{AssadiCK19} 
completely fit the premise of our original question: LFMIS is not purely a graph problem as it is not invariant under labeling of the vertices, and $s$-$t$ minimum cut  studied in~\cite{AssadiCK19} involves (i) making 
the non-standard assumption of exponential capacities, and (ii) even for unit-capacity graphs, is not known to admit an $\Ot{n}$ communication protocol (see~\cite{blikstadBEMN22}). 

\medskip
\noindent
We prove polynomial-pass lower bounds for two natural graph problems, $k$-cores and degeneracy, by reduction from a harder variant of the HPC problem which we call MultiHPC. We also present novel $\tO(n)$ communication protocols for these two problems. These two results together give us the first natural instances of a positive answer to our motivating question. 

These results further demonstrate the power of reductions from the HPC problem, as a technique for proving strong lower bounds in the graph streaming model, which are beyond the reach of other techniques. With this in mind, we improve the lower bound of~\cite{AssadiCK19} for the HPC problem to an optimal bound of $\Omega(m^2/r)$ communication, i.e., we improve the known bound by a factor of $r$. This contribution alone results in a polynomial improvement in the number of passes, for all lower bounds that follow via reductions from HPC (for instance, it immediately improves the bounds of~\cite{AssadiCK19} for LFMIS and exponential-capacity $s$-$t$ minimum cut to $\Omgt{n^{1/4}}$ passes).

But, as it turns out, all the known lower-bounds that follow by reduction from HPC also follow by reduction from MultiHPC. For this variant, we can prove an $\Omega(m^2)$ lower-bound for $r$ rounds (since the input size for MultiHPC is $r\cdot m^2$), and this translates to an improved semi-streaming lower-bound of $\Omgt{n^{1/3}}$ passes for all of the above problems.

\subsection{\texorpdfstring{$k$-Cores}{k-Cores} and Degeneracy in Graph Streams}\label{sec:degen-intro}

For any undirected graph $G=(V,E)$ and integer $k \geq 1$, a \emph{$k$-core} in $G$ is a maximal set $S$ of vertices such that the induced subgraph of $G$ on $S$, denoted by $G[S]$, has a minimum degree of at least $k$. In other words, 
any vertex in $S$ has at least $k$ other neighbors in $S$. 

$k$-Cores provide a natural notion of well-connectedness in massive graphs, and as such, computing $k$-cores (and more generally $k$-core decompositions; see, e.g.,~\cite{LiuSYDS22}) 
has been widely studied in databases~\cite{BonchiGKV14, ChuZ00ZXZ20, LiZZQZL19}, social networks~\cite{DhulipalaBS17, DhulipalaBS18, KhaouidBST15}, machine learning~\cite{Alvarez-HamelinDBV05, EsfandiariLM18, GhaffariLM19}, 
among others~\cite{LiWSX18,GalimbertiBGL20, SariyuceGJWC13}. 

As a result, in recent years, there has been a rapidly growing body of work on
computing $k$-cores on massive graphs in parallel and streaming
models of computation~\cite{SariyuceGJWC13,DhulipalaBS17, DhulipalaBS18,EsfandiariLM18, GhaffariLM19,LiuSYDS22}. In particular,~\cite{EsfandiariLM18} presented a single-pass algorithm that for any $\eps > 0$, 
computes a $(1-\eps)$-approximation of every $k$-core in $G$ (i.e., obtains a $(1-\eps)$-approximate $k$-core decomposition) in $\Ot{n/\eps^2}$ space (see also~\cite{SariyuceGJWC13} for an earlier streaming algorithm and~\cite{GhaffariLM19}
for a closely related parallel algorithm). 

The \emph{degeneracy} of a graph $G=(V,E)$, denoted by $\kappa(G)$, is the largest integer $k \geq 0$ such that $G$ contains a non-empty $k$-core. 
The simple greedy algorithm that at every step peels off the smallest degree vertex results in the so-called \emph{degeneracy ordering} of $G$ and $\kappa{(G)}$ 
is equal to the largest degree of a vertex removed in this peeling process~\cite{MatulaB83}. Degeneracy is a standard measure of uniform sparsity and is closely
related to other such notions like arboricity (which is within a factor $2$ of degeneracy). Moreover, computing degeneracy is a subroutine for approximating various
other problems such as arboricity~\cite{AlonYZ97}, densest subgraph~\cite{Charikar00}, $(\kappa+1)$ coloring~\cite{erdHos1966chromatic}. 

The degeneracy problem, and closely related uniform-sparsity measures such as densest subgraph, have also been studied extensively in the graph 
streaming literature~\cite{BahmaniKV12,Farach-ColtonT14,BhattacharyaHNT15,McGregorTVV15,Farach-ColtonT16,BeraCG19,AlonA20}. In particular,~\cite{Farach-ColtonT14} provided an $O(\log{n})$-pass semi-streaming algorithm that outputs a constant factor approximation to degeneracy, and~\cite{Farach-ColtonT16} subsequently improved this to a single-pass $(1-\eps)$-approximation in $\Ot{n/\eps^2}$ space (see also~\cite{McGregorTVV15} for densest subgraph and~\cite{BeraCG19,AlonA20} for degeneracy coloring). 
 
\medskip

In terms of lower bounds, \cite{BeraCG19} prove that any \emph{single-pass} streaming algorithm that computes the exact value of degeneracy or approximates it within an additive factor of $\lambda$ requires $\Omega(n^2)$ space or $\Omega(n^2/\lambda^2)$ space respectively. Our \emph{polynomial-pass} lower bounds for $k$-cores and degeneracy, now in a very strong sense, rule out the possibility of extending any prior semi-streaming algorithms computing near-optimal solutions to these problems, to compute \emph{exactly} optimal solutions. 

\subsection{Our Results}\label{sec:res}

We give an informal presentation of our results here and postpone the formal descriptions to subsequent sections. The first main result is 
our polynomial-pass lower bound for $k$-core computation and degeneracy. 
\\

\begin{result}[Formalized in~\Cref{thm:degen-lb}]\label{res:stream-lb}
	For any integer $p \geq 1$, any $p$-pass streaming algorithm for computing the degeneracy of an input $n$-vertex graph requires $\Omgt{n^2/p^3}$ space. In particular, any semi-streaming algorithm for the problem requires $\Omgt{n^{1/3}}$ passes. 

\noindent	
	Moreover, the same lower bounds also apply 
	to the algorithms that given any integer $k \geq 1$, can check whether or not the input graph contains a non-empty $k$-core. 
\end{result}

\Cref{res:stream-lb} provides a strongly negative answer to the question of obtaining semi-streaming algorithms for exact computation of degeneracy and $k$-cores in a small number of passes. 
We obtain~\Cref{res:stream-lb} via a detailed and technical reduction, presented in~\Cref{sec:reduction}, from a variant of the Hidden Pointer Chasing (HPC) problem of~\cite{AssadiCK19}, which we call Boolean Multilayer Hidden Pointer Chasing (BMHPC). 

In a standard HPC problem, we are given $m$ instances of $m$-bit Set-Intersection ($O(m^2)$ bits in total) and interpret the intersection point of each instance as pointing to a different instance among these $m$. We then wish to know the position we end up in after following $r+1$ pointers. In a Boolean variant, we only care to know the \textit{parity} of the position we end up in. In a Multilayer variant, we are given $r$ different layers, each layer with its own $m$ instances ($r m^2$ bits in total), where we think of the intersection points at each layer as pointing to some instance in the next layer, and wish to know where we end up in the last layer by following these pointers.

In addition to the reduction from BMHPC, in~\Cref{sec:communication-upper}, we present a novel and non-trivial communication protocol that finds the degeneracy ordering (and thus degeneracy itself) and non-empty $k$-cores for any given $k$, using only $\Ot{n}$ communication. This communication upper bound thus places the $k$-core and degeneracy problems as perfect illustrations of a positive answer to our and~\cite{AssadiCK19}'s motivating question outlined earlier: namely, problems 
that prior techniques could not have proven any lower bound beyond $\log{n}$ passes. 
\Cref{res:stream-lb} thus constitutes the first set of natural graph problems with polynomial pass lower bounds for semi-streaming algorithms. 

\medskip

Our second main contribution is providing optimal lower bounds for the HPC problem and all its variants (Boolean, Multilayer, and Boolean Multilayer). A communication lower-bound of $\Omega(\frac{m^2}{r^2})$ was previously known for $(r-1)$-round protocols computing the $r$-th pointer in a (single layer) HPC problem. We prove the following.

\begin{result}[Formalized in~\Cref{thm:lb-hpc}]\label{res:hpc-lb}
	For any integer $1 \le r = O(\sqrt m)$, any $(r-1)$-round protocol for computing the $r$-th pointer in the HPC problem on a universe of size $m$ requires $\tilde \Omega({m^2}/r)$ communication. 

\medskip
\noindent
 For any integer $r \geq 1$, any $(r-1)$-round protocol for computing the $r$-th pointer in the Multilayer HPC problem on a universe of size $m$ requires $\tilde \Omega(m^2)$ communication. 
 
	\medskip
	\noindent
	Moreover, the same lower bounds hold for the Boolean versions of the above. 
\end{result}

\Cref{res:hpc-lb}, by strengthening the lower bound of~\cite{AssadiCK19}, allows us to prove polynomially stronger bounds on the number of passes of semi-streaming algorithms via reductions from HPC. As it turns out, every known reduction from HPC \cite{AssadiCK19} can be easily converted to a reduction from MHPC. Our results thus imply improved pass lower bounds (from $\Omgt{n^{1/5}}$ to $\Omgt{n^{1/3}}$)  for semi-streaming algorithms solving these problems. We capture this in the next corollary.

\begin{corollary}\label{cor:lfmis-mincut}
For any integer $p \geq 1$, any $p$-pass streaming algorithm for the following problems on $n$-vertex graphs requires $\Omgt{n^2/p^3}$ space. In particular, any semi-streaming algorithm for these problems require $\Omgt{n^{1/3}}$ passes.
\begin{itemize}
    \item Computing the minimum s-t cut value in a weighted graph (with exponential edge capacities)
    \item Computing the lexicographically-first maximal independent set (LFMIS) of an undirected graph
\end{itemize}
\end{corollary}

We obtain~\Cref{res:hpc-lb} by following the elegant analysis of pointer chasing problems due to~\cite{Yehudayoff20} via the \emph{triangular discrimination distance} between distributions, as opposed to more standard measures 
such as KL-divergence and total variation distance typically used in this context. This in turn requires extending the notion of ``almost solving'' for the Set-Intersection problem introduced by~\cite{AssadiCK19} (and further refined in~\cite{AssadiR20}),
to the triangular discrimination distance: roughly speaking, this corresponds to proving a lower bound for communication protocols that, instead of finding the intersecting element, change its distribution slightly 
from uniform distribution. The analysis in~\cite{AssadiCK19} measured this change by total variation distance, but now we need to do so by triangular discrimination distance instead. Finally, we prove a nearly-optimal
lower bound on the communication-distance tradeoff for almost solving Set-Intersection in terms of the triangular discrimination distance.


%% file: overview.tex
\section{Overview}
\label{sec:overview}

\subsection{Hidden Pointer Chasing}\label{sec:over-hpc} 

The Multilayer Hidden Pointer Chasing (MHPC) problem, the starting point
of our reductions, is defined as follows.
The problem operates on two disjoint universes $\cX = \Set{x_1, \ldots , x_m}$
and $\cY = \Set{y_1, \ldots , y_m}$.
There are four players $P_A, P_B, P_C, P_D$, out of which $P_A$ and $P_B$ each
hold $r m$ subsets of $\cY$, called $A^j_x$ and $B^j_x$ for
$x \in \cX$ and $j \in \brac{r}$, and $P_C$ and $P_D$ each hold $r m$
subsets of $\cX$, called $C^j_y$ and $D^j_y$ for $y \in \cY$ and $j \in \brac{r}$,
with the promise $\card{A^j_x \cap B^j_x} = 1$ and $\card{C^j_y \cap D^j_y} = 1$
for every $j \in \brac{r}$, $x \in \cX$ and $y \in \cY$.
This means that each pair of sets of two of the players, e.g. $A^j_x$ and
$B^j_x$ defines a pointer $\{y \} = A^j_x \cap B^j_x$, which we think of as
pointing to the pair of sets in the next layer, $C^{j+1}_y$ and $D^{j+1}_y$,
belonging to the other two players.
Following these pointers, and writing a singleton set as the element it
contains,
we define a sequence $z_0 = x_1$, $z_1 = A^1_{z_0} \cap B^1_{z_0},
z_2 = C^2_{z_1} \cap D^2_{z_1}, z_3 = A^3_{z_2} \cap B^3_{z_2}$, \textit{etc}.
In the $\mhpc_{m,r}$ problem, the players wish to learn $z_r$.
In the $\bmhpc_{m,r}$ problem, the players only need to learn one bit about
$z_r$, that is, $b(z_r) := i \bmod 2$ where $i$ is the index of $z_r$ in
$\cX$ or $\cY$.
Now, there is a very obvious way of doing this in $r$ rounds, if the correct
pair of players start: the players just follow the pointers, solving the
necessary Set-Intersection instances. This costs them $r$ rounds with $O(m)$
bits of communication per round, for a total of $r m$ bits. However, we will
show: 

\begin{theorem} \label{thm:lb-mhpc-outline} Any randomized protocol with less than $r$ rounds, or even any randomized protocol with $r$ rounds which is \textit{misaligned}, in that the ``wrong'' pair of players starts to speak, cannot solve $\bmhpc_{m,r}$ correctly with fewer than $\Omega(m^2)$ bits of communication.
 \end{theorem}

This theorem is proven by combining ideas from three different previous works: \cite{AssadiCK19}, \cite{AssadiR20}, and \cite{Yehudayoff20}. But first, let us give an overall intuition for why it should be expected to hold.

In a misaligned $r$-round protocol for $\bmhpc_{m,r}$, it is players $P_C$ and
$P_D$ who begin the protocol by talking with each other.
This means that the ``wrong'' pair of players begin to speak, in the sense that
they wish to compute the value $\{z_1\} = A^1_{1} \cap B^1_{1}$, but this
instance is with $P_A$ and $P_B$, so they have no way to do this.
So the first round cannot say anything about $z_1$: the best $P_C$ and $P_D$ can
do is send some information about all of their Set-Intersection instances,
without knowing which one is important.
This means that each bit that $P_C$ and $P_D$ communicate with $P_A$ and $P_B$
in the first round can only reveal $\frac 1 m$ bits of information about the
average instance.
But now in the next round, $P_A$ and $P_B$, although they know
$z_1$, cannot have learned much information about $C^2_{z_1}$ or
$D^2_{z_1}$.
But then, how can they say anything about $\{ z_2 \} = C^2_{z_1} \cap
D^2_{z_1}$?
The difficult situation is now reversed!
This ``always one step behind'' situation is similar to what happens for pointer
chasing \cite{NisanW93, PonzioRV99, Yehudayoff20}, except now the pointers are
``hidden'' behind set intersection instances.

A previous paper of Assadi, Chen and Khanna \cite{AssadiCK19} showed a lower-bound of $\Omega(\frac{m^2}{r^2})$ for the (single layer) Hidden Pointer Chasing problem $\HPC_{m,r}$, which is a version of $\mhpc_{m,r}$ where all the layers are identical ($A^j_i, B^j_i, C^j_i, D^j_i$ is the same for all $j \in [r]$). The lower-bound was proven via an information-theoretic argument. They first show that any low-round protocol for HPC must be ``almost solving'' a set intersection instance on one of the rounds. They then show that this is impossible via an information complexity argument, akin to the lower-bound on the information complexity for set disjointness. However, a later paper by Assadi and Raz \cite{AssadiR20} directly showed that any protocol that ``almost solves'' set intersection can be used to obtain a protocol that exactly solves set intersection (hence the term ``almost solving''). This would allow us to replace the \textit{ad hoc} information complexity argument in \cite{AssadiCK19}, and instead appeal, in a black-box fashion, to a previously known lower-bound on the information complexity of Set-Intersection \cite{JayramKS03}.

One could take these previous lower-bounds for $\HPC_{m,r}$, and prove a lower-bound of $\Omega(\frac{m^2}{r})$ for $\mhpc_{m,r}$, but not the lower-bound of $\Omega(m^2)$ which we obtain here. The insufficiency of these previous proofs comes from the notion of ``almost solving'' that is used. There, a protocol is said to ``almost solve'' Set-Intersection if the distribution of the intersection point is sufficiently changed by the knowledge gained from the protocol's execution. More precisely, if the distribution of the intersection point $\mu(A \cap B \mid \Pi)$, conditioned on knowing the transcript $\Pi$, is sufficiently far away, in total variation distance (TVD), from the distribution of the intersection point $\mu(A \cap B)$, as it is known before the protocol begins. The quadratic margin in terms of $r$ is ultimately a result of the quadratic loss between TVD and Shannon information, in the use of Pinsker's inequality.

This same issue was the cause of a decades-long open problem on the complexity
of (non-hidden) pointer chasing. Nisan and Wigderson proved in 1991
\cite{NisanW91} that any $r$-round protocol for pointer chasing, where the wrong
player starts, needs to communicate $\omega(\frac{m}{r^2})$ bits. But there is a
simple upper bound of $O(\frac{m}{r})$. In 2000, Klauck \cite{klauck2000quantum}
gave a non-constructive proof of a matching lower-bound.
That is, he showed that the randomized communication complexity is indeed
$\Omega(\frac m r)$, but without providing a hard distribution, which must exist
via Yao's Principle.
This problem remained open until 2019, when Yehudayoff
\cite{Yehudayoff20} showed that the distributional complexity of pointer chasing
is $\Omega(\frac m r)$ under the uniform distribution, whenever $r \ll \sqrt m$.
The proof used a measure of information called \textit{triangular
discrimination}, which had never before been used in the lower-bounds
literature.

Thus, being simultaneously aware of the three works of \cite{AssadiCK19}, \cite{AssadiR20}, and \cite{Yehudayoff20}, one is naturally led to ask if they can be combined in such a way as to improve the $\frac{m^2}{r^2}$ lower bound for HPC, to $\frac{m^2}{r}$? And could we then prove a lower bound of $\Omega(m^2)$ for Multilayer HPC?

This turns out to be the case. We are not only able to prove \Cref{thm:lb-mhpc-outline}, but we also improve the lower bound for (single layer) HPC:

\begin{theorem} \label{thm:lb-hpc-outline} Let $r = O(\sqrt m)$. Then, any randomized protocol with less than $r$ rounds, or even any \textit{misaligned} randomized protocol with $r$ rounds, cannot solve $\bhpc_{m,r}$ correctly with fewer than $\Omega(\frac{m^2}{r})$ bits of communication.
 \end{theorem}

The key insight in the new lower bounds is that the notion of ``almost solving'' an instance of Set-Intersection can be adapted to use triangular discrimination instead of TVD. Two issues then need to be addressed.

First, we must show that a low-round protocol for HPC or MHPC must be ``almost solving'' (in the new sense) an instance of Set-Intersection in one of the rounds. The proofs follows the general outline of \cite{AssadiCK19}, but need to be adapted to use triangular discrimination instead of TVD. To see that it works, one must first understand that triangular discrimination obeys a property analogous to TVD, saying that the expected value of $f(x)$, when $x$ is sampled by some distribution $\mu$, is not too far from the expected value of $f(x)$ when $x$ is sampled by a different distribution $\nu$, if $\mu$ and $\nu$ are close with respect to triangular discrimination. This is obvious for TVD, but not as obvious for triangular discrimination. It is also not obvious how to adapt the proof to Multilayer HPC, in a way that works for any number of rounds $r \le m$.

Second, we must show that a low-information protocol that ``almost solves'' (in
the new sense) Set-Intersection can still be used to obtain a low-information
protocol that exactly solves Set-Intersection.
The proof is similar to \cite{AssadiR20}.
In that paper, a reduction is given which solves a given Set-Intersection
instance by sampling $O(1)$ runs of a protocol that ``almost solves''
Set-Intersection in terms of TVD.
We reinterpret their reduction as using the almost-solving protocol to
assign scores to elements (predicting how likely they are to be the
intersecting element), and come up with a new scoring function which allows
a reduction from set intersection to almost-solving with respect to
positive triangular discrimination.

\subsection{Reduction to Degeneracy}

We give a high-level overview of the key idea behind the reduction from \bmhpc to the streaming problem of finding the graph degeneracy. First, suppose that we want to show a streaming lower bound for the harder problem of finding a degeneracy ordering. In the classical offline setting, we can obtain such an ordering by the peeling algorithm that recursively removes the min-degree node from the graph and appends it to the end of the ordering. But naively implementing this algorithm in the semi-streaming setting seems difficult since it is inherently sequential. We can store the degree of each node in semi-streaming space and find the min-degree node $v$ in the graph. After we remove $v$, we need to find a min-degree node $v'$ in the new graph $G\setminus\{v\}$. But at the beginning of the stream, we did not know which node $v$ is, and hence might not have stored enough of its neighbors so as to update their degrees and find $v'$. Hence, naively, we need to make a new pass for each peeled vertex, which takes $\Theta(n)$ passes in total for an $n$-node graph. One might wonder whether \emph{any} semi-streaming algorithm for degeneracy ordering would need close to these many passes. If so, how do we prove it?

Consider just the basic problem of finding a min-degree node in an $n$-node graph, which is the primitive for finding a degeneracy ordering. It can be shown via a simple reduction that a streaming algorithm for this problem can be used to solve $\setint_n$, the Set-Intersection communication problem with universe size~$n$. As noted above, finding the degeneracy ordering translates to finding a sequence of nodes that have smallest degree in the remaining graph. This means we can use it to basically solve a sequence of $\setint_{\Theta(n)}$ instances. These instances are, however, not independent. The solution to the first instance gives a min-degree node in the original graph, whose removal leads to the second instance; solving this instance reveals the third instance, and so on and so forth. This gives a flavor of a combination of $\setint$ and pointer chasing, where each pointer is revealed by solving a $\setint$ instance corresponding to the previous pointer. This is precisely the concept behind \hpc (or \mhpc for that matter)! Hence, it is plausible that the degeneracy ordering problem can be reduced from \mhpc, and we embark on the journey to find such a reduction.

Recall the definition of \mhpc from \Cref{sec:over-hpc}. Given an instance of
\mhpc, we construct the following layered graph with $r+1$ layers
$L_0,\ldots,L_r$. Each layer has $m$ nodes: the nodes in the even layers
correspond to $x_i$'s and the ones in the odd layers correspond to $y_i$'s.
The edges of the graph are always between two consecutive layers. The players
$P_A$ and $P_B$ encode the sets $A^1_{x_i}$ and $B^1_{x_i}$ by adding edges
between $L_0$ and $L_1$.
Consider the following encoding: for each $i,j\in [m]$, if $y_j\in A^1_{x_i}$,
then $P_A$ adds an edge from the $i$th node in $L_0$ to the $j$th node in $L_1$.
$P_B$ does the analogous construction for the elements in $B^1_{x_i}$ (note that
this can lead to parallel edges).
$P_C$ and $P_D$ encode the sets $C^2_{y_i}$ and $D^2_{y_i}$ by adding edges between $L_1$ and $L_2$ in the analogous way. Again, $P_A$ and $P_B$ encode $A^3_{y_i}$ and $B^3_{y_i}$ with edges between $L_2$ and $L_3$, and this proceeds alternately until the relevant players add the edges between $L_r$ and $L_{r+1}$. 

Let $v_0$ be the first node in $L_0$; recall that it corresponds to $x_1=z_0$. Assume that $v_0$ is the min-degree node in the graph with $\deg(v_0)=d-1$ and all other nodes have the same degree $d$. Again, recall that $z_1=A^1_{z_0}\cap B^1_{z_0}$. By construction and by the unique-intersection promise of the $\setint$ instances of \mhpc, $v_0$ has two parallel edges to the node representing $z_1$ in $L_1$; call this node $v_1$. To all other nodes in $L_1$, $v_0$ has at most one edge. Hence, when the peeling algorithm deletes $v_0$ from the graph, only the degree of $v_1$ drops by $2$, i.e., $\deg(v_1)$ becomes $d-2$; all other nodes have at most a drop of $1$ in degree, i.e., have degree $\geq d-1$. Thus, $v_1$ becomes the new min-degree node in the graph. Now, when $v_1$ is deleted, by similar logic, the node $v_2$ in $L_2$, corresponding to the element $z_2=C^2_{z_1}\cap D^2_{z_1}$, becomes a min-degree node in the remaining graph with $\deg(v_2)=d-2$. However, now some node in $L_0$ might also have degree $d-2$; this is the case when $z_1=A^1_{x_i}\cap B^1_{x_i}$ for some $i\neq 1$ as well. Now assume that the peeling algorithm breaks ties by choosing a node in the highest layer among all min-degree nodes (and arbitrarily within the highest layer). Then, indeed it chooses $v_2$ as the next node to peel (since it is the unique min-degree node in $L_2$, again by the \setint promise). Thus, it follows inductively that the $i$th iteration of the peeling algorithm removes the node corresponding to $z_{i-1}$ in $L_{i-1}$. Hence, the $(r+1)$th node in the degeneracy ordering can be used to identify $z_r$.  

The above high-level idea has quite a few strong assumptions. The challenge is now to get rid of them. We list these challenges and then describe how we overcome them.

\begin{itemize}
    \item[(i)] The constructed graph has parallel edges. Then the reduction would only prove a lower bound against algorithms that can handle multigraphs, which is much weaker than a lower bound against algorithms that work on simple graphs.
    \item[(ii)] We assume that the tie-breaking is done by the peeling algorithm so as to pick a vertex in the highest layer. It is not at all clear how to get rid of this assumption in a straightforward way. 
    \item[(iii)] We also assume that we can set the initial degrees in such a way that $v_0$ has degree $d-1$ and all other nodes have degree $d$. It is not clear that we can do this while preserving the relevant properties of the construction.
    \item[(iv)] Even if we can overcome the above challenges and the reduction goes through, then we prove a lower bound for finding degeneracy ordering, which is (at least formally) harder than the problem of finding the degeneracy value. Ideally we would like to show the lower bound for the simplest variant of the problem: checking whether degeneracy of the graph is smaller than a given value $k$ or not.  
\end{itemize}

To get around (i), we modify the construction to have a pair of nodes represent each element. The edge construction is done in the following way. Suppose the pair $(u_1,u_2)$ represents an element $x_i$ in layer $\ell-1$, and $(w_1,w_2)$ represents $y_j$ in layer $\ell$. If $y_j\in A_{x_i}^\ell$, then we add edges from $u_1$ to both $w_1$ and $w_2$. Similarly, if $y_j\in B_{x_i}^\ell$, then we add edges from $u_2$ to both $w_1$ and $w_2$. Note that if $y_j\in A^\ell_{x_i}\cap B^\ell_{x_i}$, then we have all $4$ cross edges between the two pairs, and otherwise we have only $2$ edges between them, one on each $w_i$. Hence, when $u_1$ and $u_2$ are removed, both $w_1$ and $w_2$ lose degree by $2$ if $y_j$ is the intersecting element, and otherwise they only lose degree by $1$. This captures the property of the reduction that we want, without constructing parallel edges.

For (ii), we do something more elaborate. On a high level, we duplicate each of the layers $L_1,\ldots,L_r$ to provide a ``padding'' between two initially-consecutive layers. This padding has additional nodes that create an asymmetry between the layer preceding it and the one succeeding it. This asymmetry ensures that the degrees of the nodes in the higher layer drop more than those in the lower layer. Then, we can proceed with the peeling algorithm as planned.

To handle (iii), we show that once we are done with the construction based on the \mhpc instance, we can consider each node, look at its degree, and add edges from it to some auxiliary vertices so as to reach its ``target degree''. We need to be careful about two things: one, we preserve the properties of the construction so that the reduction goes through, and two, we do not add too many new nodes that might make the bound obtained from the reduction weak. We succeed in achieving a construction without violating the above. 

Finally, for (iv), we observe that while we gave the above outline for a reduction from \mhpc, the ``easier'' boolean version \bmhpc has a similar lower bound. We then succeed in extending the ideas to reduce the boolean problem of ``checking whether degeneracy $\leq k$'' from the \bmhpc problem, thus obtaining the desired lower bound for this simple variant. For reduction from \bmhpc, where the goal is to output just the bit $b(z_r)$ (see definition in \Cref{sec:over-hpc}) rather than $z_r$, we need to make non-trivial modifications in the graph: we join the nodes which represent the bit-$1$ elements in the last layer, to some ``special nodes'' $S$. The other nodes in the last layer are not joined to them. The special nodes are also adjacent to all nodes in the other layers. We show that if $b(z_r)=1$, then after the peeling algorithm removes the nodes corresponding to $z_r$ in the last layer, the degrees of the special nodes drop enough such that all the remaining vertices get peeled one by one, while having degree at most some value $k$ during deletion. This implies that the graph has degeneracy $\leq k$. Otherwise, if $b(z_r)=0$, we show that after peeling off the nodes representing $z_r$ in the last layer, the minimum vertex-degree in the remaining subgraph is at least $k+1$, implying that the degeneracy of the graph must be at least $k+1$. We give the detailed reduction and proof in \Cref{sec:reduction}.

\subsection{Communication Upper Bounds for Degeneracy}

We give a short overview of the $\Ot{n}$ communication protocol for computing
the degeneracy of a graph.
In the two player communication model,
the edges of the input graph $G$ are split into two disjoint sets $E_A$ and
$E_B$ given to the players Alice and Bob respectively, and they wish to
find the degeneracy of $G$.
Note that the search problem (finding the degeneracy) reduces to its
decision counterpart (is the degeneracy $\leq k$?) by a binary
search, costing only a $\log n$ multiplicative factor in the communication
cost. Hence, we focus on the version where Alice and Bob are additionally given an
integer $k$, and wish to decide if the degeneracy of $G$ is at most $k$.

To solve this decision problem, we implement the following version of the
peeling algorithm in a communication protocol:
while there is a vertex of degree $\leq k$, remove it.
If the graph is non-empty at the end, reject, otherwise accept.
The main challenge in adapting this algorithm is that in the worst case, it
seems to update the degree of almost all vertices in $G$ after each deletion, and
there is no way to do that without a lot of communication.

However, we observe that if a vertex has degree at least $k + \sqrt{n}$, then
it cannot be deleted for the next $\sqrt{n}$ iterations (since each iteration can reduce its degree by at most one).
This observation alone gives us the following $\Ot{n \sqrt{n}}$ communication
protocol:
\begin{enumerate}
  \item Compute the degree of each vertex in $G$.
  \item Ignore all vertices of degree $\geq k + \sqrt{n}$ while performing
    $\sqrt{n}$ rounds of the trivial peeling algorithm.
  \item Go to Step 1.
\end{enumerate}

Note that the communication in Step 2 comes from Alice and Bob sending each other the low degree ($< k + \sqrt{n}$) neighbors of the vertex deleted in each iteration
of the peeling algorithm.
We observe that while a vertex has degree $\geq k + \sqrt{n}$, it is not listed in Step 2, and once its degree falls below the threshold of
$k + \sqrt{n}$, it is listed at most $\sqrt{n}$ times due to Step 2. Thus, the total communication due to Step 2 over the course of the entire protocol is bounded by $\tO(n\sqrt{n})$.
Also, we recompute the degrees of \emph{all} vertices
(which costs $O(n \log n)$ communication each time)
at most $\sqrt{n}$ times; these two facts combined give us the desired bound.

To get an improved $\Ot{n}$ communication protocol, we extend the idea above
to partition the vertices into $\log n$ sets, where the $i$-th set contains
vertices of degree between $k + 2^{i - 1}$ and $k + 2^i$.
While the global approach (of simply ignoring the high-degree vertices for
$\sqrt{n}$ steps) does not work any more, we are able to make a more local
argument as follows:
for a vertex of degree $k + \ell$ to be deleted, it must lose at least $\ell$
neighbors, which means it must lose $\ell / 2$ neighbors in either Alice's or Bob's edge
set.
But now the players can just track this ``private'' loss of degree of each
vertex, and communicate to update the degree of a vertex in the $i$-th set
only when either private degree falls by at least $2^{i - 2}$.
We are able to show that the degree of each vertex is updated $O( \log n )$
times over the entire course of this new protocol, and hence the total communication is $\Ot{n}$. We further show that this can be extended to finding a $k$-core of the graph with $\Ot{n}$ communication. Thus, we establish finding degeneracy and $k$-core as not-too-hard problems.

%% file: preliminaries.tex
\section{Preliminaries} \label{sec:prelim}

\subsection{The Set-Intersection Problem}

\begin{definition}
    The \textbf{$m$-bit Set-Intersection} problem is a communication problem
    for two players Alice and Bob.
    Alice is given a subset $X \subseteq [m]$ and Bob is given a subset
    $Y \subseteq [m]$, with the promise that $|X \cap Y| = 1$.
    Their goal is to learn the point of intersection $\{t\} = X \cap Y$. 
\end{definition}

Let us define a distribution $\cD_\SI$ over inputs for Set-Intersection. This is the positive case of Razborov's original hard distribution for disjointness \cite{Razborov92}.

\begin{tbox}
    \textbf{Distribution $\cD_\SI$}~\cite{AssadiR20} on sets
    $(X, Y)$ from the universe $\cX=[m]$: 

    \begin{itemize}
    \item Uniformly sample two disjoint sets $X', Y' \subseteq [m]$ of size
      $\frac m 4 - 1$.
    \item Sample $\estar \in [m]\setminus(X' \cup Y')$ uniformly at random
      and define $X = X' \cup \Set{\estar}$ and $Y = Y' \cup \Set{\estar}$.
\end{itemize}
\end{tbox}

\begin{proposition}\label{prop:SI-hard}
  The internal information complexity of \emph{finding} the intersecting element
  in a $\setint_m$ instance sampled from $\dsi$ is $\Omega(m)$.
\end{proposition}

\Cref{prop:SI-hard} is not directly implied by standard set-intersection lower
bounds, since they use distributions where $X$ and $Y$ never intersect.
To obtain this proposition, one can (for example) apply the result
of~\cite{JayramKS03} on the information cost of set disjointness protocols on
intersecting distributions (see~\cite{AssadiCK19} for more details).

\subsection{The (Boolean) Hidden Pointer Chasing Problem}

The Hidden Pointer Chasing ($\hpc$) problem is a composition of Pointer Chasing and Set-Intersection, defined as follows.

\begin{definition}\label{prob:hpc}
The \textbf{Hidden Pointer Chasing ($\HPC$)} problem is a communication problem
with four players $P_A, P_B, P_C$, and $P_D$. 
Let $\cX := \{x_1,...,x_m\}$ and $\cY := \{y_1,...,y_m\}$ be two disjoint universes. 
    \begin{itemize}
        \item For any $x\in \cX$, $P_A$ and $P_B$ are given an instance $(A_x,B_x)$ of \setint over the universe $\cY$ where $A_x \cap B_x = \{t_x\}$ for a single target element $t_x \in \cY$. We define $A := \{A_{x_1} ,\ldots, A_{x_m} \}$ and $B:=\{B_{x_1} ,\ldots, B_{x_m} \}$ as the whole input to $P_A$ and $P_B$, respectively.
        
        \item For any $y\in \cY$, $P_C$ and $P_D$ are given an instance $(C_y,D_y)$ of $\setint$ over the universe $\cX$ where $C_y \cap D_y = \{t_y\}$ for a single target element $t_y \in \cX$. We define $C := \{C_{y_1} ,\ldots, C_{y_m} \}$ and $D:=\{D_{y_1} ,\ldots,D_{y_m} \}$ as the whole input to $P_C$ and $P_D$, respectively.
        
        \item We define two mappings $f_{AB} : \cX \to \cY$ and $f_{CD} : \cY \to \cX$ such that:
        \[
            (a) ~ \textnormal{for any $x\in \cX$, $f_{AB}(x) = t_x$;} \qquad \qquad (b) ~ \textnormal{for any $y\in \cY$, $f_{CD}(y) = t_y$}.
         \]
        \item Let $x_1 \in \cX$ be an arbitrary fixed element known to all players. The pointers $z_0,z_1,z_2,\ldots$ are defined inductively as 
        $$z_0 := x_1, ~~z_1 := f_{AB}(z_0),~~ z_2 := f_{CD}(z_1), ~~z_3 := f_{AB}(z_2),~\ldots.$$
    \end{itemize}
For any integers $m,r \geq 1$, the \textnormal{\textbf{$m$-bit, $r$-step Hidden Pointer Chasing}} problem, denoted by $\hpc_{m,r}$ is defined as the communication problem of finding the pointer $z_r$. Meaning, we wish to output the index $i \in [m]$ such that $z_r = x_i$ (when $r$ is even) or $z_r = y_i$ (when $r$ is odd). 

We also define a Boolean variant of $\hpc$ as follows. In the \textnormal{\textbf{$m$-bit, $r$-step \underline{Boolean} Hidden Pointer Chasing}} problem, denoted by $\bhpc_{m,r}$, the goal is to output the bit $b(z_r):=i \mod 2$, where $z_r = x_i$ or $y_i$ depending on if $r$ is even or odd respectively. 
\end{definition}

\subsection{The (Boolean) Multilayer Hidden Pointer Chasing Problem}

The Multilayer Hidden Pointer Chasing problem ($\mhpc$) combines
$r$ instances of $\hpc$ as its input in the following manner.

\begin{definition}
  The \textbf{Multilayer Hidden Pointer Chasing} ($\mhpc$) problem is a
  communication problem with four players $P_A, P_B, P_C$ and $P_D$. Let $m,r$ be integers $\geq 1$. The $m$-bit, $r$-layer \mhpc problem, denoted by $\mhpc_{m,r}$ is defined as follows.
  \begin{itemize}
    \item For each $i \in \bracket{r}$, the players receive an $\hpc$ instance
      $(A^i, B^i, C^i, D^i)$, each over the disjoint universes
      $\cX$ and $\cY$ of size $m$.
      Let $t^i_x$ denote the unique element in the intersection
      $A^i_x \cap B^i_x$, and $t^i_y$ the unique element in $C^i_y \cap D^i_y$.
    \item For each $i \in \brac{r}$, we have the maps $f^i_{AB} : \cX \to \cY$
      and $f^i_{CD} : \cY \to \cX$, defined as:
      \[
        (a) ~ \textnormal{for any $x\in \cX$, $f^i_{AB}(x) = t^i_x$;}
        \qquad \qquad
        (b) ~ \textnormal{for any $y\in \cY$, $f^i_{CD}(y) = t^i_y$}.
      \]
    \item
      The pointers $z_i$ are defined inductively as:
      \[
        z_0 := x_1, ~~z_1 := f^1_{AB}(z_0),~~ z_2 := f^2_{CD}(z_1),
        ~~z_3 := f^3_{AB}(z_2),~\ldots.
      \]
  \end{itemize}

  The goal of the problem is to find $z_r$, i.e., the index $i \in [m]$ such that $z_r = x_i$ (when $r$ is even) or $z_r = y_i$ (when $r$ is odd). 

Analogous to \hpc, we define the Boolean variant of $\mhpc$ as follows. In the \textnormal{\textbf{$m$-bit, $r$-layer \underline{Boolean} \mhpc}} problem, denoted by $\bmhpc_{m,r}$, the goal is to output the bit $b(z_r):=i \mod 2$, where $z_r = x_i$ or $y_i$ depending on if $r$ is even or odd respectively. 
\end{definition}

We remark that $C^j, D^j$ for odd $j$, and $A^j, B^j$ for even $j$, are redundant (they do not affect the output of the function). They are kept for notational convenience and ease of presentation. 

\subsection{Misaligned \texorpdfstring{$r$-Round}{r-Round} Protocols}

We are interested in bounded-round protocols for $\hpc$, $\bhpc$, $\mhpc$ and $\bmhpc$. In each round, two of the players speak with each other ($P_A$ with $P_B$, or $P_C$ with $P_D$), but not with the other pair of players, with a message between pairs of players triggering the next round.

\begin{definition}[Round of $\hpc_{m,r}$]\label{def:round}
An $r$-round protocol is a protocol where the communication is organized in rounds $1, 2, \ldots, r$. A given pair of players --- $P_A$ and $P_B$ or $P_C$ and $P_D$ --- communicate with each other in the odd rounds (the first, third, etc rounds) and the other pair of players communicate with each other in the even rounds. Each round ends with the speaking pair of players sending a message to the other pair of players, and this marks the beginning of the next round.
\end{definition}

It is easy to see that one can solve any of the above HPC variants by $r$-round protocols communicating $O(r m)$ bits, if $P_A$ and $P_B$ speak in the odd rounds, and $P_C$ and $P_D$ speak in the even rounds. We will show lower-bounds for \textit{misaligned} $r$-round protocols where the \textit{wrong} pair of players speak:

\begin{definition}
We call an $r$-round protocol \textit{misaligned} if the pair of players $P_C$ and $P_D$ speak in the odd rounds, and $P_A$ and $P_B$ speak in the even rounds.
\end{definition}

The above definitions naturally suggest that, in an $r$-round protocol, we should think of $z_j$ as the pointer that the players wish to know by the end of round $j$. In a misaligned protocol, however, the players speaking in round $j$ do not hold the Set-Intersection instance that needs to be solved to find $z_j$.

\subsection{Notation and Terminology}

\mypar{Notation for Protocols} We use the notation $\prot_\SI$, $\prot_\bhpc$,
and $\prot_\bmhpc$, to denote protocols for Set-Intersection, $\bhpc$ and
$\bmhpc$, respectively. 

In general, we will use the upright $\rA$ to denote a random variable
associated with $A$.

Now consider a situation where we choose inputs $(A, B, C, D)$ from some
distribution on inputs to $\hpc_{m,r}$.
We then let $\rZ_1, \ldots , \rZ_r$ denote random variables with values
$z_1, \ldots, z_r$ of \Cref{prob:hpc}.
We use the variable $j \in [r]$ to index the rounds of a
$r$-round protocol.
For any $j \in [r]$, we define $\Prot_j$ as the set of all messages
communicated in round $j$, and more generally $\Prot$ denotes the entire
protocol transcript. 

For any $x \in \cX$ and $y \in \cY$, we define random variables
$\rT_x \in \cY$ and $\rT_y \in \cX$ giving the intersecting elements of the
Set-Intersection instance $(\rA_x, \rB_x)$ and $(\rC_y, \rD_y)$, respectively.
We define $\rE_j = (\rPi_1, \ldots, \rPi_{j-1}, \rZ_1, \ldots, \rZ_{j-1})$ to
include all messages exchanged before round $j$, together with the entire path
followed by the pointers before round $j$.
Note that during round $j$, the players wish to compute the intersection
$A_{z_j} \cap B_{z_j}$ (or $C_{z_j} \cap D_{z_j}$), so knowing $\rE_j$ should
not help them to do that. (Note that $\rE_1$ is empty.)

On an input $(A, B, C, D)$ to $\bmhpc_{m,r}$, we define the notation
$Z_i$, $j$, $\Prot_j$, $T^j_x$, $T^j_y$ analogously.

If we have random variables $\rA, \rB$, we will broadly use the notation
$\mu(\rA \mid \rB)$ to mean ``the distribution of $\rA$, conditioned on the
value of $\rB$ that occurred''.
So, for example, if $\rB$ is constant we may write $\mu(\rA)$, and this
is some fixed distribution.
But if $\rB$ is non-constant, the distribution $\mu(\rA \mid \rB)$ may vary
depending on $\rB$, and so $\mu(\rA \mid \rB)$ is itself also a random
variable, which depends on $\rB$ and can be determined by knowing $\rB$ alone
(even without knowing $\rA$).

\mypar{Other notation}
An input graph to a streaming algorithm will usually be called $G=(V,E)$ and $n$ will denote $|V|$. The degeneracy of graph $G$ will be denoted by $\kappa(G)$, or simply $\kappa$ if the graph is clear from the context. The notation ``$\log x$'' stands for $\log_2 x$. For a positive integer $a$, the notation $[a]$ denotes the set $\{1,2,\ldots,a\}$ and for integers $a<b$, $[a,b]$ denotes the set $\{a,a+1,\ldots,b\}$.

\subsection{Positive Triangular Discrimination} \label{sec:td}

\begin{definition}[Positive triangular discrimination]
    Given two distributions $\mu$ and $\nu$ on a sample space $\Omega$, the \textit{positive triangular discrimination} 
    $\Lambda(\mu,\nu)$ is:
    \[
    \Lambda(\mu, \nu) = \sum_{\substack{x \in \Omega\\\mu(x) > \nu(x)}}\frac{(\mu(x) - \nu(x))^2}{\mu(x) + \nu(x)}.  
    \]
\end{definition}

In this paper, we work exclusively with \textit{positive} triangular discrimination. Below we state some of the useful properties of positive triangular discrimination. We defer the proofs of the following propositions to \Cref{sec:td-props}.

\begin{restatable}{proposition}{lambdatvd}\label{prop:prop-lambda}
$0 \le \frac{\|\mu - \nu\|^2_1} 8\leq \Lambda(\mu, \nu)\leq \frac{\|\mu - \nu\|_1} 2 \le 1$. 
\end{restatable}

The following property is already implicit in \cite{Yehudayoff20}, but making it explicit now simplifies the presentation later. Note that an analogous equation for total variation distance, i.e., $\E_\mu[f(x)] \le \|\mu-\nu\|_1\cdot \max_{x \in x} f(x) + \cdot \E_{\nu}[f(x)]$ can be obtained almost from the definition. The following statement, however, needs a bit more work.

\begin{restatable}{proposition}{lambdaloss}\label{itm:lambda-loss}
For any function $f:\Omega\to\bbR_{\ge 0}$,
\[ \E_\mu[f(x)] \le \Lambda(\mu,\nu)\cdot \max_{x \in \Omega} f(x) + 6 \cdot \E_{\nu}[f(x)]. \]
\end{restatable}

\noindent
We will also make use of the convexity of $\Lambda$.

\begin{restatable}{proposition}{convex}\label{prop:convex}
Positive triangular discrimination $\Lambda(\mu,\nu)$ is convex as a joint function of the two distributions $\mu$ and $\nu$. 
\end{restatable}

\subsection{Degeneracy and \texorpdfstring{$k$-Cores}{k-Cores}}

The definition of a $k$-core and the degeneracy $\kappa(G)$ of a graph $G$ can be found in \Cref{sec:degen-intro}. Here, we define some related terms that come up often in the paper. We also state some basic facts about them that we later use in our proofs.

Given an ordering $\sigma$ of vertices, we say $v\prec_\sigma u$ if $u$ appears (somewhere) after $v$ in $\sigma$. 

\begin{definition}[Outdegree with respect to an ordering]\label{def:odegordering}
    For a graph $G=(V,E)$. given an ordering $\sigma$ of vertices in $V$, the outdegree of a vertex $v$ with respect to $\sigma$, denoted by $\odeg_{\sigma}(v)$, is the number of neighbors of $v$ that appear after it in $\sigma$. Formally,
$$\odeg_{\sigma}(v) = |\{u\in N(v): v \prec_\sigma u\}|$$ 
\end{definition}

We can also think of it as orienting all edges $(x,y)$ from $x$ to $y$ if $x\prec_\sigma y$. Then $\odeg_\sigma(v)$ is the outdegree of $v$ in the resultant oriented graph.

\begin{definition}[$k$-ordering]\label{def:k-ordering} 
    Given a graph $G=(V,E)$, an ordering $\sigma$ of vertices in $V$ is called a $k$-ordering if $\forall v\in V: \odeg_\sigma(v)\leq k$.   
\end{definition}

The following fact connects vertex orderings to degeneracy. 

\begin{fact}[\cite{MatulaB83}]\label{fact:min-k-ord}
    The minimum number $k$ for which a graph $G$ admits a $k$-ordering is $\kappa(G)$.
\end{fact}

The above fact implies that to prove that a graph has degeneracy at most $k$, it suffices to exhibit a $k$-ordering of the nodes. We now formally define \emph{degeneracy ordering} of a graph.  

\begin{definition}[Degeneracy ordering] \label{def:degen-ord}
    Given a graph $G=(V,E)$, a $\kappa(G)$-ordering of $V$ is called a degeneracy ordering of $G$.
\end{definition}

A degeneracy ordering of a graph can be obtained by running the following ``peeling algorithm'' on it: starting with an empty sequence $\sigma$, recursively remove the minimum degree vertex from $G$ and append it to $\sigma$ until the graph becomes empty. Again, the maximum degree that a node has when it is peeled by this algorithm is exactly $\kappa(G)$. If not, then the peeling algorithm gives a $k$-ordering for some $k<\kappa(G)$ which violates \Cref{fact:min-k-ord}. Hence, we get the following fact.

\begin{fact}\label{fact:peelingmin-degen}
Suppose the peeling algorithm on a graph $G$ removes its nodes in the order $v_1,\ldots,v_n$. Define $G_0:=G$ and $G_i:= G\setminus \{v_1,\ldots,v_i\}$ for $i\in [n]$. Then, $$\kappa(G)=\max_{i\in [n]} \deg_{G_{i-1}}(v_i)$$. 
\end{fact}

%% file: lowerbounds.tex
\newcommand{\bucket}[1]{\textsf{Bucket} #1}

\section{New Lower Bounds for HPC and Multilayer HPC} \label{sec:lower-bound-hpc}

The two main results of this section are an $\Omega(m^2)$ lower-bound against misaligned $r$-round protocols for $\bmhpc_{m,r}$, and an $\Omega(\frac{m^2}{r})$ lower-bound against misaligned $r$-round protocols for $\bhpc_{m,r}$. Notice that in a misaligned $r$-round round protocol for $\bmhpc_{m,r}$ or $\bhpc_{m,r}$, it is players $P_C$ and $P_D$ who begin the protocol by talking with each other. This means that the ``wrong'' pair of players begin to speak, in the sense that they wish to compute the value $\{z_1\} = A_{z_0} \cap B_{z_0}$, but this instance is with Alice and Bob, so they have no way to do this. So the first round of communication says nothing about $z_1$: the best $P_C$ and $P_D$ can do is send some information about all of their Set-Intersection instances, without knowing which is important. So each bit that $P_C$ and $P_D$ communicate with $P_A$ and $P_B$ in the first round, can only reveal $\frac 1 m$ bits of information about the average instance. So, in the next round, $P_A$ and $P_B$ cannot have learned what $C_{z_1}$ or $D_{z_1}$ are, either. But then, how can they learn $z_2$? The difficult situation is now reversed! This ``always one step behind'' situation is similar to what happens for pointer chasing \cite{NisanW93, PonzioRV99, Yehudayoff20}, except now the pointers are ``hidden'' behind set intersection instances.

\medskip\noindent
How much communication must the players use in order to escape this unfortunate situation? In the case of $\bmhpc_{m,r}$, it is not clear how to do this in less than $O(m^2)$ bits of communication. One can think of two strategies: either reveal an entire layer, causing the protocol to become aligned at that layer, or solve roughly $\frac{m}{r}$ set intersections instances in each layer, which will cause the protocol to become aligned (follow the pointer of a solved instance) at any given layer with probability $\frac{1}{r}$. We will prove that these two strategies are essentially optimal.

However, in the case of $\bhpc_{m,r}$, it can be seen that $O(\frac{m^2}{r})$ bits of communication suffice. Indeed, $P_C$ and $P_D$ (who speak first) may choose $N = O(\frac{m}{r})$ of their Set-Intersection instances uniformly at random, solve them all by communicating $O(\frac{m^2}{r})$ bits, and then send the solution of all these instances to $P_A$ and $P_B$. Then, in the next round, $P_A$ and $P_B$ begin the standard protocol for solving $\hpc_{m,r}$, meaning, they compute $z_1$ and send it back, then $P_C$ and $P_D$ compute $z_2$ and send it back, etc. If at any point $P_A$ and $P_B$ compute some $z_j$ which is among the $N$ revealed Set-Intersection instances, then they can skip a step, and learn $z_{j+2}$ immediately without communicating. The protocol is now aligned, and can be finished within the remaining number of rounds. The probability that the path taken crosses one of the instances revealed by $P_C$ and $P_D$ is $\Omega(r/ N)$, which can be made arbitrarily close to $1$ with our choice of $N$.

On the other hand, one should not expect to be able to prove any significant lower-bound when $r$ is a sufficiently large fraction of $\sqrt n$. Indeed, in this case, the birthday paradox says that the path being traversed will cycle with high probability. The players may then follow the path (wasting the first round entirely) until they find a cycle, and then they will immediately learn $z_k$.

\medskip\noindent
The intuitive explanation and the upper-bounds given above will lead us to conjecture the following, necessarily tight lower-bounds, which are the main results of this section.

 \begin{theorem} \label{thm:lb-mhpc}
 Any misaligned $r$-round protocol, solving $\bmhpc_{m,r}$ correctly with constant probability requires $\Omega(m^2)$ bits of communication.
 \end{theorem}
 
 \begin{theorem} \label{thm:lb-hpc}
Any misaligned $r$-round protocol, where $r = o(\sqrt m)$, solving $\bhpc_{m,r}$ correctly with constant probability requires $\Omega(\frac{m^2}{r})$ bits of communication.
 \end{theorem}

 A weaker lower-bound of $\frac{m^2}{r^2}$ for HPC was proven in~\cite{AssadiCK19}.
 Our proofs of these tight lower-bounds follow a similar strategy to the earlier
 proof, but with a somewhat modified protocol inspired by~\cite{Yehudayoff20}.
We begin by proving \Cref{thm:lb-hpc} in \Cref{sec:eps-solving-to-mhpc}, because the proof is slightly simpler, and it is enough for our streaming lower-bounds. The proof of \Cref{thm:lb-hpc} appears in \Cref{sec:eps-solving-to-bhpc}.
Both proofs follow a similar outline. We first show that a short, misaligned $r$-round protocol for $\bmhpc$ (or $\bhpc$) must be ``almost solving'' an instance of Set-Intersection on one of the rounds, while revealing little information.
The notion of ``almost solving'' will be formally defined in \Cref{sec:almost-solving}. 
Then, in \Cref{sec:exact-to-eps}, we show that any low-information protocol for ``almost-solving'' Set-Intersection can be used to actually solve Set-Intersection, also with low information. 
Putting it all together in \Cref{sec:put-together}, we will find that a misaligned $r$-round protocol for $\bmhpc_{m,r}$ communicating $o(m^2)$ bits (or for $\bhpc_{m,r}$ communicating $o(\frac{m^2}{r})$ bits), would allow us to solve Set-Intersection with $o(m)$ information cost, which is known to be impossible by \cite{JayramKS03}.

\subsection{``Almost Solving''}\label{sec:almost-solving}

The notion of ``almost solving'' is given precisely below.
Intuitively, a protocol $\eps$-solves Set-Intersection if the distribution of
the intersection point $\Set{\restar} = \rA \cap \rB$ is sufficiently altered by
the execution of the protocol.
Meaning, the distribution of $\restar$ conditioned on knowing the protocol (and
possibly one of the inputs) is sufficiently far-away from the distribution of
$\restar$ without this knowledge.
In~\cite{AssadiCK19}, the chosen notion of \textit{sufficiently far} is the
total variation distance.
Here we will choose positive triangular discrimination.

\begin{definition}[$\eps$-solving Set-Intersection, with respect to positive
  triangular discrimination]\label{def:eps-solving}
Let $\cD$ be a distribution on inputs $(A, B)$ for Set-Intersection.
A randomized protocol $\pi$ \textit{internally $\eps$-solves Set-Intersection} (with
respect to positive triangular discrimination) on $\cD$ if \textit{at least one}
of the following hold true (where $\Pi$ is the transcript of $\pi$, including
public randomness):
\[
  \Exp*_{\rPi, \rA}{\PTD{\mu(\restar \mid \rPi, \rA)}{\mu(\restar \mid \rA)}} \geq \eps
  \qquad\text{ or } \qquad
  \Exp*_{\rPi, \rB}{\PTD{\mu(\restar \mid \rPi, \rB)}{\mu(\restar \mid \rB)}} \geq \eps.
\]
A protocol $\pi$ \textit{externally $\eps$-solves Set-Intersectiom} on $\cD$ if
the following holds true:
\[
  \Exp*_{\rPi}{\PTD{\mu(\restar \mid \rPi)}{\mu(\restar )}} \geq \eps.
\]
\end{definition}

\begin{observation}\label{obs:eps-solving-remove-conditioning}
    From the joint convexity of positive triangular discrimination (\Cref{prop:convex}), it follows that for any (possibly correlated) random variables $\rA, \rB, \rC$,
\begin{align*}
    \E_{\rB,\rC}[\PTD{\mu(\rA\mid \rB,\rC)}{\mu(\rA\mid \rC)}] %
    & = \E_{\rB}\left[\sum_{c} \Pr[\rC =c] \PTD{\mu(\rA\mid \rB,\rC=c)}{\mu(\rA\mid \rC=c)}\right] \\
    & \ge \E_{\rB}\left[\Lambda\left(\sum_{c} \Pr[\rC =c]\mu(\rA\mid \rB,\rC=c), \sum_{c} \Pr[\rC =c]\mu(\rA\mid \rC=c)\right)\right]\\
    & = \E_{\rB}[\PTD{\mu(\rA\mid \rB)}{\mu(\rA)}].
\end{align*}
That is, removing a random variable from both conditions can only decrease the $\Lambda$-distance. In particular, if a protocol externally $\eps$-solves Set-Intersection on some distribution $\cD$, it also internally $\eps$-solves Set-Intersection on $\cD$.
\end{observation}

\begin{observation}\label{obs:eps-solving-private-randomness}
    It may be felt that a more natural definition of internally $\eps$-solving
    should include any eventual private randomness of the protocol $\pi$.
    However, noting that $\{\restar\} = \rA \cap \rB$ is a function of $\rA$ and
    $\rB$, and that $\rA$ and $\rB$ are independent given $\rPi$ (by
    \Cref{fact:rectangle})
    it follows that the distribution of $\restar$ is independent of the private randomness of $\pi$, when given $\rA$ and $\rPi$. So
    \[
    \Exp*_{\rPi, \rA}{\PTD{\mu(\restar \mid \rPi, \rA)}{\mu(\restar \mid \rA)}} = \Exp*_{\rPi, \rA}{\PTD{\mu(\restar \mid \rPi, \rR_A, \rA)}{\mu(\restar \mid \rA)}}
    \]
    where $\rR_A$ is Alice's randomness, and likewise on Bob's side.
    That is, including the private randomness gives the same definition.
\end{observation}

\subsection{Hard Distributions} \label{sec:distribution}

Next, we define two hard distributions for $\cD_\bmhpc$ and $\cD_\bhpc$,  over inputs to $\bmhpc$ and $\bhpc$, respectively.

\begin{tbox}
    \textbf{Distribution $\cD_\bmhpc$} on tuples $(\rA,\rB,\rC,\rD)$ from the universes $\cX$ and $\cY$: 

    \begin{enumerate}
		\item For any $i\in [r]$, and any $x \in \cX$, sample $(A^i_x,B^i_x) \sim \cD_\SI$ from the universe $\cY$ \emph{independently}. 
		\item For any $i \in [r]$, and any $y \in \cY$, sample $(C^i_y,D^i_y) \sim \cD_\SI$ from the universe $\cX$ \emph{independently}. 
    \end{enumerate}
\end{tbox}

\begin{tbox}
    \textbf{Distribution $\cD_\bhpc$} \cite{AssadiCK19} on tuples $(\rA,\rB,\rC,\rD)$ from the universes $\cX$ and $\cY$: 

    \begin{enumerate}
		\item For any $x \in \cX$, sample $(A_x,B_x) \sim \cD_\SI$ from the universe $\cY$ \emph{independently}. 
		\item For any $y \in \cY$, sample $(C_y,D_y) \sim \cD_\SI$ from the universe $\cX$ \emph{independently}. 
    \end{enumerate}
\end{tbox}

It was shown in~\cite{AssadiCK19} that misaligned randomized $r$-round protocols
for $\hpc_{m,r}$ need $\frac{m^2}{r^2}$ bits of communication under $\cD_\bhpc$.
In \Cref{sec:eps-solving-to-bhpc}, we improve this result and show that
misaligned randomized $r$-round protocols for $\bhpc_{m,r}$ (an easier problem
than $\hpc_{m,r}$) need $\frac{m^2}{r}$ bits of communication.

We restate the following simple observations from~\cite{AssadiCK19} about
$\cD_\bhpc$.
Analogous statements can be made for $\cD_\bmhpc$.

\begin{observation}[\cite{AssadiCK19}]\label{obs:distHPC}
	Distribution $\cD_\bhpc$ is \underline{not} a product distribution. However, in this distribution: 
	\begin{enumerate}[label=(\roman*)]
		\item The inputs to $P_A$ and $P_B$ are {independent} of the inputs to
      $P_C$ and $P_D$, i.e., $(\rA,\rB) \perp (\rC,\rD)$. 
		\item For any $x \in \cX$, $(\rA_x,\rB_x)$ is independent of all other
      $(\rA_{x'},\rB_{x'})$ for $x'\neq x \in \cX$.
      Similarly for all $y,y' \in \cY$ and $(\rC_y,\rD_y)$ and $(\rC_{y'},\rD_{y'})$. 
	\end{enumerate}
\end{observation}

Based on this observation, we also have the following simple property. 

\begin{proposition}[\cite{AssadiCK19}]\label{prop:hpc-rectangle}
  Let $\prot_\bhpc$ be any deterministic protocol for $\bhpc_{m,r}$ on
  $\cD_\bhpc$.
  Then, for any transcript $\Prot$ of $\prot_\bhpc$, $(\rA,\rB)
  \perp (\rC,\rD) \mid \rPi = \Prot$. 
\end{proposition}
\begin{proof}
  Follows from the rectangle property of the protocol $\prot_\bhpc$
  (Fact~\ref{fact:rectangle}).
  In particular, the same exact argument as in the two-player case implies that
  if $[(A, B),(C, D)]$ and $[(A', B'),(C', D')]$ are mapped to the same
  transcript $\Prot$, then $[(A, B),(C', D')]$ and $[(A', B'),(C, D)]$ are
  mapped to $\Prot$ as well.
  Hence, since $(\rA, \rB) \perp (\rC,\rD)$ by \Cref{obs:distHPC}, the
  inputs corresponding to the same protocol would also be independent of each
  other, namely, $(\rA, \rB) \perp (\rC, \rD) \mid \rPi = \Prot$. 
\end{proof}

\subsection{\texorpdfstring{Short Protocols for $\bmhpc$ Give Low-Information $\eps$-Solvers for $\SI$}{Short Protocols for BMHCP Give Low-Information eps-Solvers for Set Intersection}}\label{sec:eps-solving-to-mhpc}

We now show that if we have a short, $r$-round protocol for $\bmhpc_{m,r}$ on $\cD_\bmhpc$, then during one of its rounds it must be externally (and hence also internally) $\eps$-solving Set-Intersection on $\cD_\SI$, while revealing little information. This is proven by making precise the following intuitive line of thinking, the statement of which is formally captured by \Cref{lem:mhpc-dist}. 

\begin{itemize}
  \item
      By the end of round $j$, the speaking pair of players would like to
      have learned $\{\rZ_{j} \} = \rA^j_{\rZ_{j-1}} \cap \rB^j_{\rZ_{j-1}}$
      (for odd $j$, or $\{\rZ_{j} \} = \rC^j_{\rZ_{j-1}} \cap \rD^j_{\rZ_{j-1}}$,
      for even $j$).
      Now suppose that in this round they have very little information about
      $\rZ_{j-1}$.
      Let us say, for example, that the distribution of
      $\rZ_{j-1} \mid \rE_{j-1}, \rPi_{j-1}$ is close to uniform.
      This certainly holds when $j = 2$, since
      $\mu(\rZ_1 | \rE_1, \rPi_1) = \mu(\rZ_1)$ is
      completely uniform (because the wrong pair of players speak first and the
      way we construct $\cD_\bmhpc$).

  \item Now suppose that despite this, somehow, the two speaking players
    nonetheless manage to learn a lot of information about $\rZ_{j}$. Then they
    have ``almost solved'' an \textit{unknown} instance (an instance belonging
    to the other players, no less), out of the possible remaining instances of
    Set-Intersection.
    Because they have little information about $\rZ_{j-1}$,
    meaning $\rZ_{j-1}$ is close to uniform, they have thus ``almost solved'' an
    unknown, nearly-uniformly-chosen instance. 

  \item Then, intuitively, they must have ``almost solved'' this unknown instance of Set-Intersection without revealing a lot of information about it. This turns out to be impossible, as we will see in \Cref{sec:exact-to-eps}. 
\end{itemize}

One difficulty in implementing this intuitive strategy is in finding the right measure of information that allows for one to carry through the entire argument. In \cite{AssadiCK19}, \textit{total variation distance from uniform} was used as a measure of how much the speaking players know about which Set-Intersection instance they are trying to solve. But TVD incurs the quadratic loss inherent in Pinsker's inequality, leading to a lower-bound of $\frac{m^2}{r^2}$. The same quadratic loss appeared originally in Nisan and Wigderson's lower-bounds for pointer chasing \cite{NisanW93}, and this led to a three-decade-long gap between known upper and lower-bounds for pointer chasing. This gap was finally closed by Amir Yehudayoff \cite{Yehudayoff20}, using positive triangular discrimination. This suggested the possibility of using Yehudayoff's approach to improve the lower-bound of \cite{AssadiCK19}, and it is this plan that we realize here.

\bigskip\noindent
The following lemma states that no low-information protocol (in the sense of information cost) for $\bmhpc$ can reveal a lot of information (in the sense of $\Lambda$-distance-from-uniform) about a uniformly chosen $\setint$ instance at any layer $j$. The proof shows that any such protocol would give us a low-information protocol for $\eps$-solving $\setint$, which we will show is impossible in \Cref{sec:exact-to-eps}.

\begin{restatable}
{lemma}{progress-multi} \label{lem:progress}
There is a constant $c > 0$ such that, for any protocol for $\bmhpc_{m,r}$
running on input $(\rA,\rB,\rC,\rD)$ sampled from $\cD_{\bmhpc}$, and for any
$j \in [r]$, we have:
\begin{align*}
\E_{(\rE_j, \rPi_j)}\E_{x \sim \cU_\cX}\left[\Lambda(\mu(\rT^j_x \mid \rE_j,
\rPi_j), \mu(\rT^j_x))\right] 
\leq  
c \cdot \frac{\mi{\rPi}{\rA^j \mid \rA^{< j}, \rB} + \mi{\rPi}{\rB^j \mid \rB^{> j}, \rA}}{m^2}\\
\E_{(\rE_j, \rPi_j)}\E_{y \sim \cU_\cY}\left[\Lambda(\mu(\rT^j_y \mid \rE_j,
\rPi_j), \mu(\rT^j_y))\right] 
\leq  
c \cdot \frac{\mi{\rPi}{\rC^j \mid \rC^{< j}, \rD} + \mi{\rPi}{\rD^j \mid \rD^{> j}, \rC}}{m^2}\\
\end{align*}
\end{restatable}

\begin{proof}
    Take any given $j \in [r]$, and consider the following protocol for
    $\eps$-solving $\setint$ on $\cD_{\SI}$:
    
	\begin{abox}
    \textbf{Protocol $\protSI$}: The protocol for $\eps$-solving $\SI$ using a protocol $\protHPC$ for $\bmhpc_{m,r}$. 
    
     \smallskip
    
    \textbf{Input:} An instance $(X, Y) \sim \cD_\SI$ with $X, Y\subseteq\cY$.
    
    
    \begin{enumerate}
    	\item \textbf{Sampling the instance.} Alice and Bob create an instance
        $(A, B, C, D)$ of $\bmhpc_{m,r}$ as follows: 
	\begin{enumerate}
	
 \item Using \underline{public coins}, Alice and Bob sample an index
   $I \in [m]$ uniformly at random, and  
	Alice sets $A^j_{x_I} = X$ and Bob sets $B^j_{x_I} = Y$ using their given inputs in $\SI$. 
	
 \item Using \underline{public coins}, Alice and Bob sample $A^j_{x_1}, \ldots, A^j_{x_{I-1}}$ from the A-side marginal of $\cD_\SI$, and they sample $B^j_{x_{I+1}}, \ldots, B^j_{x_{m}}$ from the B-side marginal of $\cD_\SI$. 
	
 \item Using \underline{private coins}, Alice samples $A^j_{x_r}$ for $r \in \{I+1, \ldots, n\}$ so that $(A^j_{x_r},B^j_{x_r}) \sim \cD_\SI$. Similarly Bob samples $B^j_{x_\ell}$ for $\ell \in \{1, \ldots, I-1\}$. This completes the construction of $(A^j, B^j)$. 
        
 \item Using \underline{public coins}, Alice and Bob sample $A^{< j}, B^{> j}, C$ and $D$ from their respective marginals in $\cD_\bmhpc$ (this is possible by independence).
 \item Using \underline{private coins}, Alice samples $A^{> j}$ conditioned on the publicly sampled $B^{>j}$, and Bob samples $B^{< j}$ conditioned on the publicly sampled $A^{<j}$. 

        \end{enumerate}

	\item The players now run the first $j$ rounds of the protocol $\prot_\bmhpc$ on the instance $(A,B,C,D)$, by Alice playing $P_A$, Bob playing $P_B$, and both Alice and Bob simulating $P_C$ and $P_D$ with no communication (this is possible as both Alice and Bob know $(C, D)$ entirely). 
 
    Let us write $\Pi_\SI$ to denote the messages of $P_A$ and $P_B$, which were actually sent by Alice and Bob, and $\Pi_{\le j}$ to denote the entire simulated transcript of $\prot_\bmhpc$, including the messages of $P_C$ and $P_D$ (which Alice and Bob learned without communicating).
    \end{enumerate}
\end{abox}

From a simple direct-sum argument, we get an upper-bound on the information cost of $\pi_\SI$:

\begin{claim} \label{clm:direct-sum}
    \[ \ICost{\pi_\SI}{\dsi} = \frac{\mi{\rPi_{\le j}}{\rA^j \mid \rA^{< j}, \rB}
    + \mi{\rPi}{\rB^j \mid \rB^{> j}, \rA}}{m}. \]
\end{claim}

\begin{proof}
    The internal information cost of $\prot_\SI$ is
    \[
      \IC{\prot_\SI}{\cD_\SI}
      := \mi{\rX}{\rProtSI \mid \rY, \rR, \rR_B}
      + \mi{\rY}{\rProtSI \mid \rX, \rR, \rR_A},
    \]

    where $R, R_A, R_B$ are the public, Alice's, and Bob's randomness.
    Let us bound the first term, as the second term can be bounded in just the
    same way.
    We then derive:   
\begin{align*}
	\mi{\rPi_\SI}{\rX \mid \rY, \rR, \rR_B} %
    & = \mi{\rPi_{\le j}}{\rX \mid \rY, \rI,\rA^j_{<\rI},\rB^j_{>\rI},\rA^{< j},\rB^{> j},\rC,\rD, \rB^j_{< \rI}, \rB^{< j}}\\
    & = \mi{\rPi_{\le j}}{\rA^j_\rI \mid \rI,\rA^j_{<\rI},\rA^{< j},\rB,\rC,\rD}\\
    & =  \sum_{i=1}^m \frac{1}{m} \mi{\rPi_{\le j}}{\rA^j_i \mid \rI=i,\rA^j_{<i},\rA^{< j},\rB,\rC,\rD}\\
    & =  \sum_{i=1}^m \frac{1}{m} \mi{\rPi_{\le j}}{\rA^j_i \mid \rA^j_{<i},\rA^{< j},\rB,\rC,\rD}\\
    & =  \frac{1}{m} \mi{\rPi_{\le j}}{\rA^j \mid \rA^{< j},\rB,\rC,\rD}\\
    & =  \frac{1}{m} \mi{\rPi_{\le j}}{\rA^j \mid \rA^{< j},\rB}. 
\end{align*}
All equalities are by definition, except for the fourth which is because none of
the remaining distributions depend on the choice of $\rI$, and the last, which
is because $\rC, \rD$ are independent of $\rA^j$, even when conditioned on
$\rPi_{\le j}, \rA^{<j}, \rB$.
\end{proof}

On the other hand, we also get an $\eps$-solver, for $\eps$ being exactly the $\Lambda$ distance-from-uniform of a uniformly chosen coordinate in the $j$-th layer:

\begin{claim} \label{clm:average-coordinate}
$\pi_\SI$ internally $\eps$-solves Set-Intersection on $\dsi$ where:
\[
  \eps \geq \E_{(\rE_j, \rPi_j)}
  \Exp*_{x \sim \cU_\cX} {\PTD{\mu(\rT^j_x \mid \rE_j, \rPi_j)}{\mu(\rT^j_x)}}.
\]
\end{claim}

\begin{proof} 
     By \Cref{def:eps-solving} and \Cref{obs:eps-solving-private-randomness}, we have an $\eps$-solver on Bob's side, for
     \[
       \eps = \E[\Lambda(\mu(\rT \mid \rY, \rPi_\SI,
     \rI, \rR, \rR_B), \mu(\rT \mid \rY, \rI, \rR, \rR_B))], 
     \]
     where $T = X \cap Y = A^j_I \cap B^j_I$,
     $I$ and $R = A^j_{<I}, B^j_{>I}, A^{< j}, B^{> j}, C, D$ are the public
     randomness, and
     $R_B = B^j_{< I}, B^{< j}$ is Bob's private randomness. 
     This is exactly equal to
     \[
     \E[\Lambda(\mu(\rT \mid \rY, \rE_j, \rPi_{j}, \rI, \rR, \rR_B),
     \mu(\rT \mid \rY, \rI, \rR, \rR_B))], \]
     because $E_j, \Pi_j$ is a function of $\Pi_\SI, I, R, R_B$ (Bob knows both
     $A^{<j}$, $B^{<j}$ and $C$, $D$, so he also knows $Z_{< j}$, and
     $\Pi_{\le j}$).
     In turn, this equals
     \begin{align*}
       \E_{\rE_j,\rPi_j, i \sim \uni_{\brac{m}}}
       \Exp*{\PTD{\mu(\rT^j_{x_i} \mid \rE_j, \rPi_{j}, \rI = i, \rY, \rR,\rR_B)}
       {\mu(\rT^j_{x_i}\mid \rY, \rI=i,\rR, \rR_B)}}
     \end{align*}
     Observe that none of the distributions above depend on the event $\rI = i$. So this quantity is equal to
     \begin{align*}
       \E_{\rE_j,\rPi_j, i \sim \uni_{\brac{m}}}
       \Exp*{\PTD{\mu(\rT^j_{x_i} \mid \rE_j, \rPi_{j}, \rY, \rR,\rR_B)}
       {\mu(\rT^j_{x_i}\mid \rY, \rR, \rR_B)}}
     \end{align*}
     Now, by \Cref{obs:eps-solving-remove-conditioning}, this is greater or equal to
     \[
       \E_{\rE_j,\rPi_j, i \sim \uni_{\brac{m}}}
       \Exp*{\PTD{\mu(\rT^j_{x_i} \mid \rE_j, \rPi_{j} )}{\mu(\rT^j_{x_i})}}. \qedhere
     \]
\end{proof}

From \Cref{cor:si-ic-lb-external} proven in the next section, it now follows:
  \begin{align*}
    \frac{\mi{\rPi_{\le j}}{\rA^j \mid \rA^{< j}, \rB} + \mi{\rPi_{\leq j}}{\rB^j \mid \rB^{> j}, \rA}}{m} & \ge c \cdot \eps m\\
    & \geq c m \cdot \E_{(\rE_j, \rPi_j)}\E_{x \sim \cU_\cX}
    \left[\Lambda(\mu(\rT^j_x \mid \rE_j, \rPi_j), \mu(\rT^j_x)) \right].
  \end{align*}
  This concludes the proof of \Cref{lem:progress}.
\end{proof}

\bigskip\noindent
It is now possible to prove that there is no low-information protocol for solving $\bmhpc$. The argument was roughly outlined above: if we assume by induction that the $\Lambda$ distance-from-uniform is small for the answer at the previous layer, then we are essentially trying to solve a uniformly random instance at the current layer. And we just proved that this is not possible.

\begin{lemma}\label{lem:mhpc-dist}
For the same constant $c > 0$ of the \Cref{lem:progress}, and for any protocol
for $\bmhpc_{m,r}$ running on input $(A, B, C, D)$ sampled from
$\cD_{\bmhpc}$, the following statement holds for all $j \in
[r]$:%
\footnote{Intuitively, the statement says that, if $\pi_\bmhpc$
communicates little, then the distribution of $Z_j$, conditioned on the
information available at the end of round $j$, is close to uniform.}
\[
  \E_{(\rE_j, \rPi_j)}\left[\Lambda(\mu(\rZ_j \mid \rE_j, \rPi_j), \mu(\rZ_j))\right]
  \leq 6 c \cdot \frac{\mi{\rPi}{\rA^{\le j},\rC^{\le j} \mid \rB,\rD} + \mi{\rPi}{\rB^{\le j},\rD^{\le j} \mid \rB^{> j}, \rD^{> j}, \rA,\rC}}{m^2}.
\] 

\end{lemma}

This is proven by induction. The base case is easy to see: note that $\mu(\rZ_1
\mid \rE_1, \rPi_1) = \mu(\rZ_1)$ as $\rZ_1$ is independent of $\rE_1 = z_0$
which is fixed and $\rZ_1$ is independent of $\rPi_1$ as $\Pi_1$ is a function
of $(C,D)$.
We obtain case $j$ from case $j - 1$ in two steps: we first show that if the
protocol reveals a lot of information about the variable $\rZ_{j}$, then it must
essentially be $\eps$-solving a \textit{uniformly random instance}. This is
because $Z_{j} = T_{Z_{j-1}}$ and $\rZ_{j-1}$ is, under the inductive case
$j-1$, close to uniform.
We then appeal to the previous lemma.

\begin{proof}
The information---measured by $\Lambda$-distance to uniform---that the $j$-th
round of the protocol reveals about $\rZ_j = \rT_{\rZ_{j-1}}$ is the
information revealed about an instance sampled according to $\rZ_{j-1}$:
  \begin{align*}
		\E_{(\rE_j,\rProt_j)} \Bracket{\Lambda(\mu(\rZ_j \mid
    \rE_j,\rProt_j),\mu(\rZ_j))} &=\E_{(\rE_j,\rProt_j)}
    \Bracket{\Lambda(\mu(\rZ_j \mid \rE_j),\mu(\rZ_j))}\\
		&\hspace{-1.5cm}= \E_{(\rZ_{<j},\rProt_{<j})} \Bracket{\Lambda(\mu(\rZ_j
    \mid \rZ_{<j},\rProt_{<j}),\mu(\rZ_j))} \tag{by definition of
  $\rE_j := (\rZ_{<j},\rProt_{<j})$, and since $\rZ_j$ and $\rPi_j$ are
independent given $\rE_j$.} \\
		&\hspace{-1.5cm}=\E_{(\rZ_{<j},\rProt_{<j})}
    \Bracket{\Lambda(\mu(\rT^j_{\rZ_{j-1}} \mid
    \rZ_{<j-1},\rZ_{j-1},\rProt_{<j}),\mu(\rZ_j))}\tag{by definition, the pointer $Z_j = T_{Z_{j-1}}$}\\
		&\hspace{-1.5cm}= \E_{(\rZ_{<j-1},\rProt_{<j})} \; \E_{z \sim \rZ_{j-1} \mid
    \paren{\rZ_{<j-1},\rProt_{<j}}}\; \Bracket{\Lambda(\mu(\rT^j_{z} \mid
  \rZ_{<j-1},\rProt_{<j}),\mu(\rZ_j))}.
	\end{align*}
  We removed the conditioning on $\rZ_{j-1} = z$ in the first distribution
  because, conditioned on knowing $\rZ_{<j-1},\rProt_{<j}$, for any fixed value of
  $z$ the value of $\rT_{z}$ is independent of the value of $\rZ_{j-1}$:
  \begin{itemize}
	    \item If $j-1$ is odd, $\rT_{z}$ is a function of  $(\rC, \rD)$ and
        $\rZ_{j-1}$ is a function of $(\rA, \rB)$.
            \item If $j-1$ is even, $\rT_{z}$ is a function of $(\rA,\rB)$ and
              $\rZ_{j-1}$ is a function of $(\rC, \rD)$.
            \item And, furthermore, $(\rA,\rB)$ and $(\rB,\rD)$ are independent
              given $\rZ_{<j-1}, \rPi_{<j}$. 
	\end{itemize}

	\medskip\noindent
    \Cref{lem:progress} says that if the protocol were to reveal information
    about the solution $\rT_z$ of a truly uniform instance $z$, then this would
    give us a protocol for $\eps$-solving set disjointness. But, by induction,
    we know that the chosen instance $z \sim \rZ_{j-1}$ is close to uniform. It
    should then follow from \Cref{itm:lambda-loss} that if the protocol reveals
    some information about this almost uniform instance, it must also reveal
    information about a truly uniform instance. Let us now make this precise.
    
  Without loss of generality suppose that $j-1$ is odd and hence $z \in \cY$.
  Then we can upper bound the last expression above
  by:\footnote{\label{fn:tvd}In the first inequality it should be understood how
  $\Lambda$ is playing two different roles. On the one hand it is being used as
an information-theoretic quantity (the distance to uniform), and on the other it
is being used in a capacity similar to how total variation distance appears in
the obvious upper-bound $\E_\mu[f(x)] \le \E_\nu[f(x)] + \|\mu - \nu\|_1 \cdot
\max f(x)$.}

	\begin{align*}
          &\E_{(\rE_j,\rPi_j)} \Bracket{\Lambda(\mu(\rZ_j \mid \rE_j,\rPi_j),\mu(\rZ_j))}  \\
          &\qquad \qquad \leq \E_{(\rZ_{<j-1},\rPi_{<j})} \bracket{6 \cdot \E_{\paren{z \sim \cU_{\cY}}}\Bracket{\Lambda(\mu(\rT^j_{z} \mid \rZ_{<j-1},\rPi_{<j}),\mu(\rZ_j))}} \\
          &\qquad \qquad  \qquad \qquad \qquad  \qquad + \E_{(\rZ_{<j-1},\rPi_{<j})} \Bracket{\Lambda(\mu(\rZ_{j-1} \mid \rE_{j-1},\rPi_{j-1}),\cU_\cY)} %
            \tag{by \Cref{itm:lambda-loss} and linearity of expectatiom}
          \\
          &\qquad \qquad = 6 \cdot\E_{(\rZ_{<j-1},\rPi_{<j})} \bracket{ \E_{\paren{z \sim \cU_{\cY}}}\Bracket{\Lambda(\mu(\rT^j_{z} \mid \rZ_{<j-1},\rPi_{<j}),\mu(\rZ_j))}} \\
          &\qquad \qquad  \qquad \qquad \qquad  \qquad + \E_{(\rZ_{<j-1},\rPi_{<j})} \Bracket{\Lambda(\mu(\rZ_{j-1} \mid \rE_{j-1},\rPi_{j-1}),\mu(\rZ_{j-1}))}\tag{by definition $z_{j-1} \in \cY$ and $\mu(\rZ_{j-1}) = \cU_\cY$.}\\
          &\qquad \qquad\leq \bigg[6 \cdot \underbrace{c \cdot \frac{\mi{\rPi}{\rA^j \mid \rA^{< j}, \rB} + \mi{\rPi}{\rB^j \mid \rB^{> j}, \rA}}{m^2}}_{\text{\Cref{lem:progress}}} \\
          &\qquad \qquad\qquad 
            + \underbrace{6 c \cdot \frac{\mi{\rPi}{\rA^{\le j-1},\rC^{\le j-1} \mid \rB,\rD}
           + \mi{\rPi}{\rB^{\le j-1},\rD^{\le j-1} \mid \rB^{> j-1}, \rD^{> j-1}, \rA,\rC}}{m^2}}_{\text{induction hypothesis (case $j-1$)}}\bigg]\\
          & \qquad \qquad= \bigg[6c \cdot \frac{\mi{\rPi}{\rA^j \mid \rA^{< j},\rC^{< j}, \rB,\rD} + \mi{\rPi}{\rB^j \mid \rB^{> j},\rD^{>j}, \rA,\rC}}{m^2} \\
          &\qquad \qquad\qquad 
            + 6 c \cdot \frac{\mi{\rPi}{\rA^{\le j-1},\rC^{\le j-1} \mid \rB,\rD}
            + \mi{\rPi}{\rB^{\le j-1},\rD^{\le j-1} \mid \rB^{> j-1}, \rD^{> j-1}, \rA,\rC}}{m^2}\bigg]\\
          &         \qquad \qquad
          \leq
          \bigg[6 c \cdot \frac{\mi{\rPi}{\rA^j,\rC^j\mid \rA^{< j},\rC^{< j}, \rB,\rD} + \mi{\rPi}{\rB^j,\rD^j \mid \rB^{> j},\rD^{>j}, \rA,\rC}}{m^2} \\
          &\qquad \qquad\qquad 
            + 6 c \cdot \frac{\mi{\rPi}{\rA^{\le j-1},\rC^{\le j-1} \mid \rB,\rD}
            + \mi{\rPi}{\rB^{\le j-1},\rD^{\le j-1} \mid \rB^{> j-1}, \rD^{> j-1}, \rA,\rC}}{m^2}\bigg]\\
          & \qquad\qquad =
            6 c \cdot \frac{\mi{\rPi}{\rA^{\le j},\rC^{\le j} \mid \rB,\rD}
           + \mi{\rPi}{\rB^{\le j},\rD^{\le j} \mid \rB^{> j}, \rD^{> j},
         \rA,\rC}}{m^2},
	\end{align*}
where the last inequality follows from the data processing inequality
(\Cref{fact:it-facts}~\Cref{part:data-processing}).
This concludes the proof.
\end{proof}

\subsection{\texorpdfstring{$\eps$-Solvers for $\SI$ Give Exact Solvers}{eps-Solvers for Set Intersection Give Exact Solvers}}\label{sec:exact-to-eps}

We now show a lower-bound against protocols trying to $\eps$-solve
Set-Intersection, by reducing from exactly solving Set-Intersection:

\begin{theorem}\label{thm:si-ic-lb} For any constant $\gamma \in (0,1)$, there exists a constant $c \ge 1$ such that, given a protocol $\pisi$ which internally 
$\eps$-solves Set-Intersection (in $\ptd$ distance) on $\dsi$ with
$\eps \ge \frac{8}{m}$,
we may obtain a protocol $\pi$ that exactly solves Set-Intersection on $\dsi$,
such that
$\ICost{\pi}{\dsi} \leq {\frac{c}{\eps} \cdot \ICost{\pisi}{\dsi} + \gamma m}$.
\end{theorem}

From the known fact (\Cref{prop:SI-hard}) that Set-Intersection has information
complexity $\Omega(m)$ on $\dsi$ we get the following:

\begin{corollary}\label{cor:si-ic-lb-external}
  There exists a constant $c > 0$ such that
  any protocol $\pisi$ that internally $(\ptd, \eps)$-solves $\setint_m$ on
  inputs from the distribution $\dsi$, with $\eps \geq 8 / m$, has $\ICost{\pisi}{\dsi} \geq c \cdot \eps m$.
\end{corollary}

An analogue of \Cref{thm:si-ic-lb} was proven in
\cite{AssadiR20}, where $\eps$-solving was with respect to total variation
distance instead of triangular discrimination.
It was then shown that any protocol for $\eps$-solving Set-Intersection in this
TVD sense gives a protocol for Set-Intersection with a $O(1/\eps^2)$
blowup in information cost.
\Cref{thm:si-ic-lb}, instead, gives a blowup of $O(1/\eps)$.
So our assumption is stronger, but our conclusion is also stronger.

\begin{proof}[Proof of \Cref{thm:si-ic-lb}]
We begin by recalling the notation we use for the Set-Intersection problem:
Alice and Bob's input sets are
denoted by $X$ and $Y$ respectively, and their corresponding random variables
by $\rX$ and $\rY$.
The unique element in the intersection $X \cap Y$ is denoted by $\estar$,
and its random variable is $\restar$.
Throughout this proof, for a permutation $\sigma$ and a set $S$,
$\sigma(S)$ denotes the set $\Set{ \sigma(s) \mid s \in S }$.

\noindent
Without loss of generality, we may assume that we have a protocol $\pisi$ that
internally $\eps$-solves from \emph{Alice}'s side
Set-Intersection on inputs distributed according to $\dsi$, that is,
\[
  \rPTD = 
  \Exp*_{\rX, \rPi}{
     \PTD*{ \Dist{ \restar \given \rPi, \rX}}
     {\Dist{ \restar, \rX}}
  } \geq \eps,
\]
where $\eps \geq 8 / m$.
We will use $\pisi$ to design a protocol $\pi$ which solves Set-Intersection
on $\dsi$.

For now, suppose that Alice and Bob proceed as follows: they sample a uniformly
random permutation $\Perm$ of $[m]$ using public randomness, and run the
$\eps$-solving protocol $\pisi$ on $(\Perm(X), \Perm(Y))$, obtaining a
transcript $\Pi$.
Alice now knows the induced distribution of the intersection
point $\Dist{\rPerm(\estar) \given \rPi = \Pi, \rPerm(X) = \Perm(X)}$, conditioned
on the transcript she has seen, and Bob knows the analogous distribution with
his input fixed.
Without the conditioning on the transcript, we know by symmetry of $\dsi$ that
$\Dist{\rPerm(\estar) \given \rPerm(X) = \Perm(X)}$ is uniform over $\Perm(X)$.
In fact, we have the following observation about the distribution of inputs
generated by applying $\rPerm$.

\begin{observation}\label{obsv:dsi}
  For any fixed $X$ and $Y$ such that $\card{X} = \card{Y} = m / 4$, and
  $\card{X \cap Y} = 1$,
  and a uniformly random permutation $\rPerm$,
  the instance $(\rPerm(X), \rPerm(Y))$ is distributed according to $\dsi$.
\end{observation}

To see this, consider the fixed values $X = \Set{1, \ldots m / 4}$ and
$Y = \Set{m / 4, \ldots m / 2 - 1}$.
Because $\rPerm$ is a uniformly random permutation, it maps $m / 4$ to a
uniformly random element of $\brac{m}$, and the rest of $X$ and $Y$ to
uniformly random subsets of size $m / 4 - 1$.
By symmetry, this argument applies to any fixed $X$ and $Y$ of size $m / 4$
intersecting in exactly one element.

\subsubsection{A single round of \texorpdfstring{$\pi$}{the protocol}}\label{subsec:single-round}

We start by describing a single round of the protocol $\pi$.
On the input $(X, Y)$:
\begin{enumerate}
  \item Alice and Bob publicly sample a random permutation $\Perm$ of
    $\bracket{m}$, and apply it to $X$ and $Y$ to obtain $\Perm(X)$ and
    $\Perm(Y)$ respectively.
  \item They run $\pisi$ on $\Perm(X)$ and $\Perm(Y)$ to obtain a transcript
    $\Pi$.
    By doing so, Alice learns a distribution $\qpi_1, \ldots , \qpi_m$, where
    \[ \qpi_i = \Prob{ \rPerm(\estar) = i \given \rPi = \Pi, \rPerm(X) = \Perm(X)} . \]
  \item Alice uses $\qpi$ to assign ``scores'' to each element $e$ of
    $\brac{m}$, which reflect the likelihood that $\Set{e}$ is the
    intersection of $X$ and $Y$.
\end{enumerate}

This strategy follows the general idea of~\cite{AssadiR20}, but where they assigned the score $1$ to all $e$ such that $\qpi_{\Perm(e)}$ is in the largest-probability half of $\qpi_1, \ldots , \qpi_m$ (and this works for $\eps$-solving under TVD), we assign scores in a smoother fashion.
Let $i = \Perm(e)$ and $n = m / 4$ for brevity.
Then if $\qpi_{i} \leq 1/n$, we give $e$ a score of $0$.
Otherwise, its score is $(\qpi_i - 1/n) / (\qpi_i + 1/n)$.
This choice of the ``scoring function'' may seem odd at this juncture, so
we offer two comments about it:
\begin{itemize}
  \item Note that this is equal to $1 - (2 / n) / (\qpi_i + 1/n)$, which is an
    increasing function of $\qpi_i$, and hence passes a basic sanity test
    (i.e.\ if $\pisi$ says an element is more likely to be $\Perm(\estar)$,
    we believe it).
  \item 
    The scoring function was chosen to make the expected score of $\estar$
    larger than the expected score of any non-$\estar$ element by at least
    $\rPTD$, and the variance of every element's score turned out to be small.
\end{itemize}

We package the properties of this scoring function into a few claims below.
Throughout these claims, we will use the shorthand notation
$\rA \mid B$ to mean $\rA \mid \rB = B$, where $\rA$ and $\rB$ are random
variables, and $B$ is some fixed value that $\rB$ can take.

\begin{claim}\label{clm:exp-score-diff}
  For any non-$\estar$ element $e \in \brac{m}$,
  \[ \Exp_{\rPerm} {\score(\estar) - \score(e)} \geq \rPTD / 2. \]
\end{claim}

\begin{proof}
  The proof is in two parts; first we bound $\score(\estar)$ under some
  appropriate conditioning, and then $\score(e)$ for non-$\estar$ elements.

  \noindent
  Fix any $(X,Y)$ with $|X| = |Y| = n$ and $|X\cap Y| = 1$.
  Consider a fixed choice of $\Perm(X)$ and the transcript $\Pi$, and hence
  a distribution $\qpi$.
  Let $P$ denote the set of indices where $\qpi_i > 1/n$ (i.e.\ those that lead
  to positive scores).
  Note that $P \subseteq \Perm(X)$; because we condition on $\Perm(X)$
  (and more) to obtain $\qpi$.
  By definition, the probability that $\estar$ is mapped to $i$ is $q_i$, and
  hence:
  \[
    \Exp*_{\rPerm \given \Perm(X), \Pi} {\score(\estar)}
    = \sum_{i \in P} \qpi_i \cdot \frac{\qpi_i - 1/n}{\qpi_i + 1/n}.
  \]
  On the other hand, fixing both $\Perm(X)$, $\Perm(Y)$ (and hence
  also $\Perm(\estar)$) and $\Pi$, fixes the distribution $\qpi$,
  but leaves completely free the map
  $\Perm: X \setminus \Set{\estar} \to \Perm(X) \setminus \Set{\Perm(\estar)}$.
  So we find that:
 
    $\rPerm(e)$ is uniformly random over
      $\Perm(X) \setminus \Set{\Perm(\estar)}$ for any non-$\estar$ element
      $e \in X$.

  This means that for any $i \in \brac{m}$,
  \[
    \Prob_{ \rPerm \given \Perm(X), \Perm(Y), \Pi } {\rPerm(e) = i} \leq
    \frac 1{n - 1}.
  \]
  Hence we can compute the expected score of $e$ as:
  \begin{align*}
    \Exp_{\rPerm \given \Perm(X), \Perm(Y), \Pi}{\score(e)}
    &= \sum_{i \in P}  \Prob{\rPerm(e) = i} \cdot \frac{\qpi_i - 1/n}{\qpi_i + 1/n}
    \leq \sum_{i \in P}  \frac 1{n - 1} \cdot \frac{\qpi_i - 1/n}{\qpi_i + 1/n}\\
    &= \sum_{i \in P}  \pparen*{\frac 1{n} + \frac 1{n(n - 1)}} \cdot \frac{\qpi_i - 1/n}{\qpi_i + 1/n}
    \leq \frac 1n + \sum_{i \in P}  \frac 1{n} \cdot \frac{\qpi_i - 1/n}{\qpi_i + 1/n},
  \end{align*}
  where the last inequality follows because $\card{P} \leq n - 1$.
  And now taking expectation over $\Perm(X)$ and $\Pi$:
  \begin{align*}
    \expect_{\rPerm(X), \rPi}
    \Exp_{\rPerm \given \rPi, \rPerm(X)}{\score(\estar) - \score(e)}
    &=
    \Exp*_{\rPerm(X), \rPi}{
    \Exp_{\rPerm \given \rPi, \rPerm(X)}{\score(\estar)}
    - \expect_{\rPerm(Y) \given \rPi, \rPerm(X)} \Exp_{\rPerm \given \rPerm(X), \rPerm(Y)}{\score(e)}
    }\\
    &\geq
    \Exp*_{\rPerm(X), \rPi}{\sum_{i \in P} \frac{ \paren{\qpi_i - 1/n}^2 }{\qpi_i + 1/n} - 1/n} \\
    &=
    \Exp*_{\rPerm(X), \rPi}{\PTD{\qpi}{\uni_{\Perm(X)}}} - 1/n
    = \rPTD - 1 / n \geq \rPTD / 2.
    \intertext{where the last inequality is because
      $\rPTD \geq \eps \geq 8 / m = 2 / n$.
    }
  \end{align*}
\end{proof}

\begin{claim}\label{clm:variance-score-target}
  The variance of $\score(\estar)$ is at most $\rPTD$.
\end{claim}

\begin{proof}
  Recall that $\Var{X} \leq \Exp{X^2}$ for any random variable $X$.
  As before, we will bound $\Exp{\score(\estar)}$ by conditioning on
  fixed choices of $\Perm(X)$ and the transcript $\Pi$, which
  gives Alice the distribution $\qpi$.
  Also as before, let $P$ denote the set of indices $i$ such that
  $q_i \geq 1 / n$.
  \begin{align*}
    \Exp*_{\rPerm \mid \Perm(X), \Pi}{{\score(\estar)}^2}
    &= \sum_{i \in P} \qpi_i \cdot \pparen*{\frac {\qpi_i - 1 / n}{\qpi_i + 1 / n}}^2\\
    &= \sum_{i \in P} \frac {\qpi_i}{\qpi_i + 1 / n} \cdot \frac{\pparen{\qpi_i - 1/n}^2}{\qpi_i + 1 / n}\\
    &\leq \PTD{\qpi}{ \uni_{\Perm(X)}}.
  \end{align*}
  And taking expectation over $\Perm(X)$ and $\Pi$:
  \[
    \expect_{\rPerm(X), \rPi}
    \Exp*_{\rPerm \mid \rPerm(X), \rPi}{{\score(\estar)}^2}
    \leq \Exp*_{\rPerm(X), \rPi}{\PTD{\qpi}{\uni_{\rPerm(X)}}}
    = \rPTD.
  \]
\end{proof}

\begin{claim}\label{clm:variance-score-other}
  For any $e (\neq \estar) \in X$, the variance of $\score(e)$ is at most
  $2\rPTD$.
\end{claim}

\begin{proof}
  Again, because $\Var{X} \leq \Exp{X^2}$ for any random variable $X$, and
  conditioning on fixed choices of $\Perm(X), \Perm(Y)$ and $\Pi$,
  $\rPerm(e)$ for $e \in X \setminus \Set{\estar}$ is uniformly random over
  $\Perm(X \setminus \Set{\estar})$:
  \begin{align*}
    \Exp*_{\rPerm \mid \Perm(X), \Perm(Y), \Pi}{{\score(e)}^2}
    &= \sum_{i \in P} \frac 1{n - 1}
    \cdot \pparen*{\frac {\qpi_i - 1 / n}{\qpi_i + 1 / n}}^2\\
    &\leq \sum_{i \in P} \frac 1{n - 1}
    \cdot \frac{\paren{\qpi_i - 1/n}^2}{\qpi_i + 1 / n} \cdot \frac 1{1 / n}\\
    &= \PTD*{\qpi}{\uni_{\Perm(X)}} \cdot \pparen*{1 + \frac 1{n - 1}}
    \leq 2\PTD*{\qpi}{ \uni_{\Perm(X)}},
  \end{align*}
  where the first inequality holds because $\qpi_i \geq 0$.
  And taking expectation over $\Perm(X), \Perm(Y)$ and $\Pi$:
  \[
    \Exp*_{\rPerm(X), \rPi}{
      \Exp*_{\rPerm(Y) \mid \rPerm(X), \rPi} {
        \Exp*_{\rPerm \mid \rPerm(X), \rPerm(Y), \rPi}{{\score(e)}^2}
      }
    }
    \leq
    \Exp*_{\rPerm(X), \rPi}{
      2\PTD*{\qpi}{\uni_{\rPerm(X)}}
    }
    = 2\rPTD.
  \]
\end{proof}

\subsubsection{The actual protocol}\label{subsec:actual-protocol}

We are ready to describe the protocol $\pi$.
Let $\gamma \in (0,1)$ be any constant.
We will run $k = 1600 / (\eps \gamma^2)$ (independent) rounds of the
protocol in \Cref{subsec:single-round}, and for each element $e \in X$,
Alice will add up its scores in each round to obtain the total score $\totscore(e)$. \Cref{clm:exp-score-diff} tells us that we should expect the total score for $\estar$ to be higher than the total score of an element $e \neq \estar$.
We then choose the threshold
\[
  \tau = \paren{\Exp{\totscore(\estar)} + \Exp{\totscore(e)}} / 2,
\]
to be exactly in the middle of these two expected values.

Alice will then send to Bob all elements in $X$ which have $\totscore$ at least
$\tau$; let us call this set $S$.
Bob will then compute the intersection $S \cap Y$, and report the unique element
that is (hopefully) present there as the answer.

To complete the reduction, we need to show that:
\begin{itemize}
  \item The total score of the target element $\totscore(\estar)$ is at least
    $\tau$.
    This will ensure that $\estar \in S$, which means that $\pi$ is correct.
  \item The number of elements Alice sends is \emph{small}.
    This is required to control the information cost of $\pi$.
\end{itemize}
We will get both of these properties only with some (constant) high probability;
in particular, if $S$ is too large, Alice will just report failure instead
of communicating $S$.

\begin{claim}
  $\totscore(\estar) \leq \tau$ with probability at most $1 / 100$.
\end{claim}

\begin{proof}
Using Chebyshev's inequality on the total score of $\estar$, we have:

\begin{align*}
  &\mkern-36mu\Prob*{ \card*{\totscore(\estar) - \Exp{\totscore(\estar)}}
  > 1/2 \cdot \Exp{\totscore(\estar) - \totscore(e)} } \\
  &\leq \frac{ \Var*{\totscore(\estar)} }{\pparen*{1/2 \cdot \Exp{\totscore(\estar) - \totscore(e)}}^2}
  = \frac{k \cdot \Var{\score(\estar)}}{k^2 / 4 \cdot \Exp{\score(\estar) - \score(e)}^2}\\
  &\leq \frac{4\rPTD}{k\cdot\pparen{\rPTD/2}^2}
  \leq 16 / k\eps = \gamma^2/100,
\end{align*}
where the second inequality follows from \Cref{clm:exp-score-diff,clm:variance-score-target}.
Hence the probability that $\totscore(\estar)$ is less than $\tau$ is at most
$1/ 100$.
\end{proof}

\begin{claim}
  For any non-$\estar$ element $e \in X$, $\totscore(e) \geq \tau$ with
  probability at most $\gamma^2 / 50$.
\end{claim}

\begin{proof}
Using Chebyshev on the total score of $e \neq \estar$, we have:

\begin{align*}
  &\mkern-36mu\Prob*{ \card*{\totscore(e) - \Exp{\totscore(e)}}
  > 1/2 \cdot \Exp{\totscore(\estar) - \totscore(e)} } \\
  &\leq \frac{ \Var*{\totscore(e)} }{\pparen*{1/2 \cdot \Exp{\totscore(\estar) - \totscore(e)}}^2}
  = \frac{k \cdot \Var{\score(e)}}{k^2 / 4 \cdot \Exp{\score(\estar) - \score(e)}^2}\\
  &\leq \frac{4 \cdot 2\rPTD}{k \cdot \paren{\rPTD / 2}^2}
  \leq 32 / k\eps = \gamma^2/50,
\end{align*}
where the second inequality follows from \Cref{clm:exp-score-diff,clm:variance-score-other}.
\end{proof}

This means that the expected size of $S \setminus \Set{\estar}$ is at most
$\gamma^2 m / 50$, and hence by Markov's this number is $> \gamma^2 m / 10$
with probability $\leq 1/5$.
If the number of elements above the threshold exceeds $\gamma^2 m / 10 + 1$,
then the protocol $\pi$ fails.

And now we have everything we need to bound the information cost of $\pi$.
Note that this part of the argument is exactly as in~\cite{AssadiR20}.
By the chain rule of mutual information, the information cost of $\pi$
on $\dsi$ can be broken down into three components:
\begin{itemize}
  \item The mutual information
    $\Inf{\Perm_1, \ldots , \Perm_k, R_1, \ldots , R_k ; \rX \given \rY}$
    and the symmetric term with $\rX$ and $\rY$ exchanging roles,
    where $\Perm_i$ denotes the random permutations used to scramble the input
    in the $i-$th round of $\pi$, and $R_i$ denotes the randomness of the
    $i$-th run of $\pisi$.
    These are both $0$, because the permutations and the randomness are chosen
    independently of the inputs $\rX$ and $\rY$.
  \item The mutual information
    $\Inf{\rPi_1, \ldots , \rPi_k ; \rX \given \rY, \Sigma, \rR}$
    (and the symmetric term with $\rX$ and $\rY$ swapped), where
    $\rPi_1, \ldots , \rPi_k$ denote the transcripts of the $k$ runs
    of $\pisi$, $\Sigma = \Perm_1, \ldots , \Perm_k$, and
    $\rR = R_1, \ldots , R_k$.
    By the chain rule of mutual information, we have:
    \[ \Inf{\rPi_1, \ldots , \rPi_k ; \rX \given \rY, \Sigma, \rR}
      \leq \sum_{i = 1}^k \Inf{\rPi_i ; \rX \given \rY, \Sigma, \rR, \rPi_{< i}}.
    \]
    Observe that conditioning on $\rX$, $\rY$, $\Perm_i$, and $R_i$ fixes
    $\rPi_i$, and hence makes it independent of $\rPi_{< i}$, $\Perm_{-i}$,
    and $R_{-i}$.
    Applying \Cref{prop:info-decrease}, we get that the $i$-th term of the
    sum above is upper bounded by $\Inf{\rPi_i ; \rX \given \rY, \Perm_i, R_i}$.
    But now (once again) $\rPi_i$ is fixed after conditioning on $\Perm_i(\rX)$,
    $\Perm_i(\rY)$, and $R_i$, and hence independent of the remaining parts
    of $\Perm_i$ and $\rX$ and $\rY$.
    So by \Cref{prop:info-decrease}, each term in the sum is bounded above by
    $\Inf{\rPi_i ; \Perm_i(\rX) \given \Perm(\rY), R_i}$.
    Repeating this entire argument for the symmetric term tells us that:
    \[ \Inf{\rPi_1, \ldots , \rPi_k ; \rX \given \rY, \Sigma, \rR} +
      \Inf{\rPi_1, \ldots , \rPi_k ; \rY \given \rX, \Sigma, \rR}
      \leq \sum_{i = 1}^k
      \Inf{\rPi_i ; \Perm_i(\rX) \given \Perm(\rY), R_i} +
      \Inf{\rPi_i ; \Perm_i(\rY) \given \Perm(\rX), R_i},
    \]
    which is exactly $k \cdot \IntIC_{\dsi}(\pisi)$.
  \item The mutual information between the final set $\rS$ Alice sends and
    her input.
    This is trivially upper bounded by $\Ent{\rS}$, which (since $\rS$ is a
    set of size at most $\gamma^2 \cdot m / 10$) is upper bounded by
    $H_2(\gamma^2 / 10) \cdot m$.
    Since $H_2(p) \leq 2\sqrt{p \cdot (1 - p)} \leq 2\sqrt{p}$
    (by \Cref{fact:ent-ub}),
    this is at most $\gamma m$.
\end{itemize}
So letting $c = \frac{1600}{\gamma^2}$, we get that $\IntIC_{\dsi}(\pi) \le \frac{c}{\eps} \IntIC_{\dsi}(\pisi) + \gamma m$.
\end{proof}

\begin{remark}
  The protocol $\pi$ actually works in the worst case, that is, it needs no
  assumptions about $X$ and $Y$ other than $\card{X} = \card{Y} = m / 4$, and
  that $\card{X \cap Y} = 1$.
  But of course without a distribution on its inputs we cannot talk about
  its information cost.
\end{remark}

\subsection{\texorpdfstring{Short Protocols for $\bhpc$ Give Low-Information $\eps$-Solvers for $\SI$}{Short Protocols for BHCP Give Low-Information eps-Solvers for Set Intersection}}\label{sec:eps-solving-to-bhpc}

We also wish to prove a tight lower bound for $\bhpc_{m,r}$. Note that this lower-bound is not necessary to derive our lower-bounds for $k$-cores and degeneracy, so this section may be skipped for that purpose. The lower-bound ultimately follows from the following lemma, which is analogous to \Cref{lem:mhpc-dist}.

\begin{lemma}\label{lem:nonmhpc-dist}
There is a constant $c > 0$ such that the following statement holds for all $j
\in [r]$:\footnote{Intuitively, the statement says that, if $\pi_\bhpc$
communicates little, then the distribution of $\rZ_j$, conditioned on the
information available at the end of round $j$, is close to uniform.}
\[
\mathsf{Hypothesis}_j:\;\;\;\;\;\;\;\;\E_{(\rE_j, \rPi_j)}\left[\Lambda(\mu(\rZ_j \mid \rE_j, \rPi_j), \mu(\rZ_j))\right] \leq j \cdot c \cdot \left( \frac{CC(\pi_\bhpc)+r\log m}{m^2} + \frac{r} m\right).
\] 

\end{lemma}

One can already see that the statement is more complicated. This is owed to the possibility of the pointers pointing either to some of the (few) instances about which the protocol has revealed information, and the possibility of the pointers looping. These two obstacles raise technical difficulties, which are overcome in two ways: the number of jumps $r$ is assumed to be smaller than $\sqrt m$, so looping is unlikely, and there are only $O(m/r)$ instances about which the protocol could have revealed significant information, hence when hiding an instance of set-intersection in a random position, as we did before, it is unlikely we will hit such an instance.

As before, this is proven by induction. The base case $\mathsf{Hypothesis}_1$ is easy to see: note that $\mu(\rZ_1 \mid \rE_1, \rPi_1) = \mu(\rZ_1)$ as $\rZ_1$ is independent of $\rE_1 = z_0$ which is fixed and $\rZ_1$ is independent of $\rPi_1$ as $\Pi_1$ is a function if $(C,D)$. We obtain
$\mathsf{Hypothesis}_j$ from $\mathsf{Hypothesis}_{j-1}$ in two steps: we first show that if the protocol reveals a lot of information about the variable $\rZ_{j}$, then it must essentially be ``almost-solving'' a \textit{uniformly random instance}. This is because $Z_{j} = T_{Z_{j-1}}$ and $\rZ_{j-1}$ is, under $\mathsf{Hypothesis}_{j-1}$, close to uniform. We then show that this is impossible for a short protocol to do, namely: 

\begin{restatable}
{lemma}{nonmprogress} \label{lem:nonmprogress}
There is a constant $c' > 0$ such that for any $j \in [r]$, we have:
\begin{align*}
\E_{(\rE_j, \rPi_j)}\E_{x \sim \cU_\cX}\left[\Lambda(\mu(\rT_x \mid \rE_j, \rPi_j), \mu(\rT_x))\right] \leq  c' \cdot \left( \frac{CC(\pi_\bhpc)+r\log m}{m^2} + \frac{j} m\right)\\
\E_{(\rE_j, \rPi_j)}\E_{y \sim \cU_\cY}\left[\Lambda(\mu(\rT_y \mid \rE_j, \rPi_j), \mu(\rT_y))\right] \leq c' \cdot \left( \frac{CC(\pi_\bhpc)+r\log m}{m^2} + \frac{j} m\right).
\end{align*}
\end{restatable}

We state the proof of \Cref{lem:nonmprogress} shortly. Before that, we show how \Cref{lem:nonmprogress} lets us complete the induction in order to prove \Cref{lem:nonmhpc-dist}. This also follows very closely to the analogous proof of \cite{AssadiCK19} where they showed it for total variation distance.

\begin{proof}[Proof of \Cref{lem:nonmhpc-dist}]
The information---measured by $\Lambda$-distance to uniform---that the $j$-th round of the protocol reveals about $\rZ_j = \rT_{\rZ_{j-1}}$ is the information revealed about an instance sampled according to $\rZ_{j-1}$:
  \begin{align*}
		\E_{(\rE_j,\rProt_j)} \Bracket{\Lambda(\mu(\rZ_j \mid \rE_j,\rProt_j),\mu(\rZ_j))} &=\E_{(\rE_j,\rProt_j)} \Bracket{\Lambda(\mu(\rZ_j \mid \rE_j),\mu(\rZ_j))}\\
		&\hspace{-1.5cm}= \E_{(\rZ^{<j},\rProt^{<j})} \Bracket{\Lambda(\mu(\rZ_j \mid \rZ^{<j},\rProt^{<j}),\mu(\rZ_j))} \tag{by definition of $\rE_j := (\rZ^{<j},\rProt^{<j})$, and since $\rZ_j$ and $\Pi_j$ are independent given $\rE_j$.} \\
		&\hspace{-1.5cm}=\E_{(\rZ^{<j},\rProt^{<j})} \Bracket{\Lambda(\mu(\rT_{\rZ_{j-1}} \mid \rZ^{<j-1},\rZ_{j-1},\rProt^{<j}),\mu(\rZ_j))}\tag{by definition, the pointer $Z_j = T_{Z_{j-1}}$}\\
		&\hspace{-1.5cm}= \E_{(\rZ^{<j-1},\rProt^{<j})} \; \E_{z \sim \rZ_{j-1} \mid \paren{\rZ^{<j-1},\rProt^{<j}}}\; \Bracket{\Lambda(\mu(\rT_{z} \mid \rZ^{<j-1},\rProt^{<j}),\mu(\rZ_j))}.
	\end{align*}
	We removed the conditioning on $\rZ_{j-1}$ in the first distribution because, conditioned on knowing $\rZ^{<j-1},\Prot^{<j}$, for any fixed value of $z$ the value of $\rT_{z}$ is independent of the value of $\rZ_{j-1}$: \begin{itemize}
	    \item If $j-1$ is odd, $\rT_{z}$ is a function of  $(\rC,\rD)$ and $\rZ_{j-1}$ is a function of $(\rA,\rB)$.
            \item If $j-1$ is even, $\rT_{z}$ is a function of $(\rA,\rB)$ and $\rZ_{j-1}$ is a function of $(\bC,\bD)$.
            \item And, furthermore, $(\rA,\rB)$ and $(\rB,\rD)$ are independent given $\rZ^{<j-1}, \rPi^{<j}$. 
	\end{itemize}

	\medskip\noindent
    \Cref{lem:nonmprogress} says that if the protocol were to reveal information about the solution $\rT_z$ of a truly uniform instance $z$, then this would give us a protocol for $\eps$-solving set disjointness. But, by induction, we know that the chosen instance $z \sim \rZ_{j-1}$ is close to uniform. It should then follow from \Cref{itm:lambda-loss} that if the protocol reveals some information about this almost uniform instance, it must also reveal information about a truly uniform instance. Let us now make this precise.
    
	Without loss of generality suppose that $j-1$ is odd and hence $z \in \cY$. Then 
	we can upper bound the last expression above by:
	\begin{align*}
		&\E_{(\rE_j,\rProt_j)} \Bracket{\Lambda(\mu(\rZ_j \mid \rE_j,\rProt_j),\mu(\rZ_j))}  \\
		&\qquad \qquad \leq \E_{(\rZ^{<j-1},\rProt^{<j})} \bracket{6 \cdot \E_{\paren{z \sim \cU_{\cY}}}\Bracket{\Lambda(\mu(\rT_{z} \mid \rZ^{<j-1},\rProt^{<j}),\mu(\rZ_j))}} \\
		 &\qquad \qquad  \qquad \qquad \qquad  \qquad + \E_{(\rZ^{<j-1},\rProt^{<j})} \Bracket{\Lambda(\mu(\rZ_{j-1} \mid \rE_{j-1},\rProt_{j-1}),\cU_\cY)} %
   \tag{by \Cref{itm:lambda-loss} and linearity of expectatiom}
   \\
   &\qquad \qquad = 6 \cdot\E_{(\rZ^{<j-1},\rProt^{<j})} \bracket{ \E_{\paren{z \sim \cU_{\cY}}}\Bracket{\Lambda(\mu(\rT_{z} \mid \rZ^{<j-1},\rProt^{<j}),\mu(\rZ_j))}} \tag{$\ast$}\label{eq:final-inequality}\\
		 &\qquad \qquad  \qquad \qquad \qquad  \qquad + \E_{(\rZ^{<j-1},\rProt^{<j})} \Bracket{\Lambda(\mu(\rZ_{j-1} \mid \rE_{j-1},\rProt_{j-1}),\mu(\rZ_{j-1}))}\tag{by definition $z_{j-1} \in \cY$ and $\mu(\rZ_{j-1}) = \cU_\cY$.}\\
   &\qquad \qquad\leq 6 \cdot \underbrace{c' \cdot \left( \frac{CC(\pi_\bhpc)+r\log m}{m^2} + \frac{j} m\right)}_{\text{\Cref{lem:nonmprogress}}} + \underbrace{(j-1) \cdot c \cdot \left( \frac{CC(\pi_\bhpc)+r\log m}{m^2} + \frac{r} m\right)}_{\mathsf{Hypothesis}_{j-1}}\\
  & \qquad \qquad\leq j \cdot c \cdot \left( \frac{CC(\pi_\bhpc)+r\log m}{m^2} + \frac{r} m\right) \\
  &\tag{where we replaced $j \leq r$ by $r$ in the first term and set $c = 6c'$}
	\end{align*}

This concludes the proof.
\end{proof}

\bigskip\noindent
We now prove \Cref{lem:nonmprogress}. The proof shows that any protocol that reveals information about the solution $T_z$ to random instance $z$ can be converted to a low-information protocol for externally $\eps$-solving Set-Intersection. This will later be shown to be impossible. Let us state the lemma again:

\nonmprogress*
Suppose towards a contradiction that this equation does not hold. We use $\protHPC$ to design a protocol $\protSI$
	that can externally $\eps$-solve the Set-Intersection problem $(A_x,B_x)$ for a uniformly at random chosen $x \in \cX$ and appropriately chosen $\eps \in (0,1)$ to be determined later (see \Cref{def:eps-solving} for the notion of $\eps$-solving). 
	
	\begin{tbox}
    \textbf{Protocol $\protSI$}: The protocol for $\eps$-solving $\SI$ using a protocol $\protHPC$ for $\bhpc_{m,r}$. 
    
     \smallskip
    
    \textbf{Input:} An instance $(X, Y) \sim \dsi$ over the universe $\cY$ and
    a number $j \in [r]$.
    
    
    \begin{enumerate}
    	\item \textbf{Sampling the instance.} Alice and Bob create an instance
        $(A, B, C, D)$ of $\bhpc_{m,r}$ as follows: 
	\begin{enumerate}
	\item Using \underline{public coins}, Alice and Bob sample an index $i \in [m]$ uniformly at random, and  
	Alice sets $A_{x_i} = X$ and Bob sets $B_{x_i} = Y$ using their given inputs in $\SI$. 
	\item Using \underline{public coins}, Alice and Bob sample $A_{x_1}, \ldots, A_{x_{i-1}}$ from the A-side marginal of $\cD_\SI$, and they sample $B_{x_{i+1}}, \ldots, B_{x_{m}}$ from the B-side marginal of $\cD_\SI$. 
	\item Using \underline{private coins}, Alice samples $A_{x_r}$ for $r \in \{i+1, \ldots, n\}$ so that $(A_{x_r},B_{x_r}) \sim \cD_\SI$. Similarly Bob samples $B_{x_\ell}$ for $\ell \in \{1, \ldots, i-1\}$. This completes construction of $(\bA,\bB)$. 
  \item Using \underline{public coins}, Alice and Bob sample $(C, D)$ completely
    from $\cD_\bhpc$ (this is possible by Observation~\ref{obs:distHPC} as
    $(\rA, \rB) \perp (\rC, \rD)$). 
        \end{enumerate}

           \item\label{line:AB} \textbf{Computing the answer.} 
          For $q = 1$ to $j$, Alice and Bob do the following: \begin{enumerate}
              \item \textbf{Running protocol:} In round $q$, Alice and Bob run
                the round $q$ of the protocol $\prot_\bhpc$ on the instance $(A,
                B, C, D)$ by Alice playing $P_A$, Bob playing $P_B$, and both
                Alice and Bob simulating $P_C$ and $P_D$ with no communication
                (this is possible as both Alice and Bob know $(C, D)$ entirely).
                Denote this transcript by $\Pi_q$.

           \item \textbf{Fixing pointer:} Alice and Bob compute the pointer $z_{q}$ using the fact
           that for any underlying instance $(A_x,B_x) \in (A, B) \setminus (A_{x_i},B_{x_i})$ either Alice or Bob knows the entire instance. We use $\Prot^*$ to denote the cumulative transcript generated by this step in all rounds.
           
           \item They stop if $x_i = z_{q}$.
          \end{enumerate}

	\item The players return $\Pi_\SI := (\Pi_1,\ldots,\Pi_j,\Pi^*)$ from which $\mu(\estar)$ is calculated.
    \end{enumerate}
\end{tbox}

First, we show that the internal information cost of this protocol is small if the communication cost of $\Pi_\bhpc$ is small. This uses standard direct sum argument of information complexity. More concretely, we have the following:

\begin{claim} \label{clm:nonmdirect-sum}
    $\ICost{\pi_\SI}{\cD} = O\left(\frac{CC(\pi_\bhpc)}{m} + \frac{j \log m} m\right)$.
\end{claim}

The proof of \Cref{clm:nonmdirect-sum} is similar to that of Claim 5.6 of~\cite{AssadiCK19}. We defer it to the \Cref{sec:proof-direct-sum}. At this point, we also need a bound on how well $\Pi_\SI$ can ``almost'' solve Set-Intersection. To this end, we have the following guarantee.

\begin{claim} \label{clm:nonmaverage-coordinate}
$\pi_\SI$ externally $\eps$-solves Set-Intersection on $\cD$ where:
\[
  \eps \geq \E_{(\rE_j, \rPi_j)}\E_{x \sim \cU_\cX}
  \left[\Lambda(\mu(\rT_x \mid \rE_j, \rPi_j), \mu(\rT_x)) \right] - \frac{j} m.
\]
\end{claim}

\begin{proof} 
    This proof follows the proof of Claim 5.5 of \cite{AssadiCK19} almost verbatim. The only change is that we are working with positive triangular discrimination ($\Lambda$) instead of the total variation distance ($\Delta_{\sf TVD}$). However, as $\Lambda$ is bounded by $\Delta_{\sf TVD}$ (\Cref{prop:prop-lambda}), a similar calculation yields the claim.
    By \Cref{def:eps-solving}, $\eps = \E_{\rPi_\SI}[\Lambda(\mu(\rT \mid
    \rPi_\SI), \mu(\rT))]$. So we have:

    \begin{align*}
        \E_{\rPi_\SI}[\Lambda(\mu(\rT \mid \rPi_\SI), \mu(\rT))]
        &= \E_{\rE_j, \rPi_j, \rPi^\ast, \rI}
        [\Lambda(\mu(\rT_{x_{\rI}} \mid \rPi_\SI), \mu(\rT_{x_{\rI}}))]
        \tag{as $\rT = \rT_{x_{\rI}}$}\\
        &\hspace{-1.9cm}=
        \E_{\rE_j, \rPi_j} \E_{\rI} \E_{\rPi^\ast\mid \rE_j, \rPi_j, \rI}
        [\Lambda(\mu(\rT_{x_{\rI}} \mid \rPi_\SI), \mu(\rT_{x_i}))]
        \tag{as $\rI$ is independent of $(\rE_j, \rPi_j)$}\\
        &\hspace{-1.9cm}=
        \E_{\rE_j, \rPi_j}
        \left[\sum_i \frac 1 m \E_{\rPi^\ast\mid \rE_j, \rPi_j, \rI = i}
        [\Lambda(\mu(\rT_{x_i} \mid \rPi_\SI), \mu(\rT_{x_i}))]\right]
        \tag{as $\rI \sim \cU_{[m]}$}\\
        &\hspace{-1.9cm}=
        \E_{\rE_j, \rPi_j} \left[\sum_{x_i \in \rZ^{<j}}
          \frac 1 m \Lambda(\mu(\rT_{x_i} \mid \rPi_\SI), \mu(\rT_{x_i}))
        + \sum_{x_i \notin \rZ^{<j}} \frac 1 m \Lambda(\mu(\rT_{x_i} \mid \rE_j,
      \rPi_j), \mu(\rT_{x_i}))\right]\\
        \tag{If $x_i \notin \rZ^{< j}$ then $\rPi_\SI = (\rE_j, \rPi_j)$}\\
         &\hspace{-1.9cm}\geq
         \E_{\rE_j, \rPi_j} \left[\sum_{x_i \notin \rZ^{<j}} \frac 1 m
         \Lambda(\mu(\rT_{x_i} \mid \rE_j, \rPi_j), \mu(\rT_{x_i}))\right]
         \tag{as the first sum term is at least 0}\\
         &\hspace{-1.9cm}\geq \E_{\rE_j, \rPi_j} \E_i\left[
         \Lambda(\mu(\rT_{x_i} \mid \rE_j, \rPi_j), \mu(\rT_{x_i}))\right]
         -\frac{j}{m}. \tag{as $\Lambda$ is upper bounded by 1}
    \end{align*}
\end{proof}

\noindent
Using \Cref{clm:nonmdirect-sum} and the lower bound on $\ICost{\pi_\SI}{\cD}$ (\Cref{cor:si-ic-lb-external}), we have:
\[
\eps = O\left(\frac{CC(\pi_\bhpc)}{m^2} + \frac{j \log m} {m^2}\right)
\]
Now, using \Cref{clm:nonmaverage-coordinate}, we get:
\[
\E_{(\rE_j, \rPi_j)}\E_{x \sim \cU_\cX} \left[\Lambda(\mu(\rT_x \mid \rE_j,
\rPi_j), \mu(\rT_x)) \right] - \frac j n = O\left(\frac{CC(\pi_\bhpc)}{m^2} + \frac{j \log m} {m^2}\right),
\]
which concludes the proof of \Cref{lem:nonmprogress}.

\subsection{Putting it All Together}\label{sec:put-together}

 Given \Cref{lem:mhpc-dist}, we prove \Cref{thm:lb-mhpc} in the following way:
 Consider a protocol $\pi$ that is solving $\bmhpc_{m,r}$ with communication
 cost $o(m^2)$.
 By \Cref{lem:mhpc-dist}, we have:

 \begin{align*}
     \E_{(\rE_r, \rPi_r)}\left[\Lambda(\mu(\rZ_r \mid \rE_r, \rPi_r), \mu(\rZ_r))\right]
     &\leq 6 c \cdot \frac{\mi{\rPi}{\rA,\rC \mid \rB,\rD} + \mi{\rPi}{\rB,\rD
\mid \rA,\rC}}{m^2}\\ &\underset{(\ast)}{\leq} 12c \cdot \frac{\myCC(\pi_\bhpc)}{m^2} = o(1),
 \end{align*}

where inequality $(\ast)$ follows from \Cref{part:uniform,part:info-atmost-rv}
in \Cref{fact:it-facts}.
However, next, we show that, for a correct protocol for $\bmhpc_{m,r}$, this
quantity should be $\Omega(1)$, which contradicts the assumption that
$\myCC(\pi_\bhpc) = o(m^2)$.
An entirely similar argument derives \Cref{thm:lb-hpc} from \Cref{lem:nonmhpc-dist}.
The only difference between these two arguments is the bound on $r$: We have to
assume $r = o(\sqrt{m})$ because of the additive $r^2/m$ factor in
\Cref{lem:nonmhpc-dist} (when applied with on the $r$-th round of the protocol).

\begin{lemma} \label{lem:hpc-correctness}
    Any protocol that solves $\bmhpc_{m,r}$ on distribution $\cD_\bmhpc$, or $\bhpc_{m,r}$, on distribution $\cD_\bhpc$, with constant probability must have:
    \[
 \E_{(\rE_r, \rPi_r)}\left[\Lambda(\mu(\rZ_r \mid \rE_r, \rPi_r), \mu(\rZ_r))\right] = \Omega(1).
\]
\end{lemma}

\begin{proof}
    
 Let $b$ denote the modulus function such that $b(z_k) = i \mod 2$ iff $z_r = x_i$ or $y_i$ depending on where $r$ is even or odd. If the protocol $\pi$ computes $\bmhpc_{m,r}$ or $\bhpc_{m,r}$ with probability $2/3$, we have the following inequality:

 \begin{align*}
     2/3 &= \E_{\rE_r,\rPi_r} \Pr_{\rZ_r \mid \rE_r,\rPi_r} [b(\rZ_r) =
     \cO_\pi(\rE_r,\rPi_r)]\\
     &\underset{(1)}\leq \E_{\rE_r,\rPi_r}\left[\Pr_{\rZ_r \sim
     \cU_{\cX}}[b(\rZ_r) = \cO_\pi(\rE_r,\rPi_r)] + \|\mu(b(\rZ_r) \mid \rE_r,\rPi_r) - \cU_{0,1}\|_1 \right]\\
     &\underset{(2)}\leq \E_{\rE_r,\rPi_r}\left[\Pr_{\rZ_r \sim
     \cU_{\cX}}[b(\rZ_r) = \cO_\pi(\rE_r,\rPi_r)] + \|\mu(\rZ_r \mid \rE_r,\rPi_r) - \cU_\cX\|_1 \right]\\
     &= \frac 1 2 + \E_{\rE_r,\rPi_r}\left[\|\mu(\rZ_r \mid \rE_r,\Pi_r) - \cU_\cX\|_1 \right]
 \end{align*}
 where (1) holds because of the simple fact that $\Pr_p[\cE] \leq \Pr_q[\cE] + \|p-q\|_1$ where $p$ and $q$ are two distributions over a sample space of which $\cE$ is an event (see Footnote \ref{fn:tvd}), and (2) holds because of data processing inequality. This implies:
 \begin{align*}
     \E_{\rE_r,\rPi_r}\left[\|\mu(\rZ_r \mid \rE_r,\rPi_r) - \cU_{0,1}\|_1
     \right] \geq 1/6 &\Rightarrow \left(\E_{\rE_r,\rPi_r}\left[\|\mu(\rZ_r \mid
     \rE_r,\rPi_r) - \cU_{0,1}\|_1 \right] \right)^2\geq 1/36\\
     &\Rightarrow \E_{\rE_r,\Pi_r}\left[\|\mu(\rZ_r \mid \rE_r,\rPi_r) -
     \cU_{0,1}\|_1^2 \right] \geq 1/36\\
     &\Rightarrow \E_{\rE_r,\Pi_r}
     \left[\Lambda(\mu(\rZ_r \mid \rE_r,\Pi_r), \cU_{0,1}) \right]
     \geq 1/288 = \Omega(1),
 \end{align*}
 which completes the proof.
\end{proof}

%% file: reduction.tex
\section{Reduction From \bmhpc to Degeneracy}\label{sec:reduction}

In this section we prove streaming lower bounds for degeneracy related problems. Precisely, we prove the following theorem. 
\begin{restatable}{theorem}{mainthmdegen}
\label{thm:degen-lb}
Given an integer $k$ and an $n$-vertex graph $G$, any $p$-pass streaming algorithm that can check whether $\kappa(G)\leq k$ or not requires $\Omega(n^2/p^3- (n\log n)/p)$ space.   \end{restatable}

This immediately implies the following corollary for semi-streaming algorithms.

\begin{corollary}
A semi-streaming algorithm for any of the following problems on an $n$-vertex graph $G$ requires $\tilde{\Omega}(n^{1/3})$ passes.
\begin{itemize}
    \item Check whether the degeneracy $\kappa(G)\leq k$ for a given integer $k$.
    \item Exactly compute the value of  $\kappa(G)$.
    \item Find a degeneracy ordering of $G$.
    \item Find a non-empty $k$-core of $G$ for a given value of $k$ (or report NONE if none exists).
\end{itemize}
 \end{corollary}

We prove \Cref{thm:degen-lb} by reduction from \bmhpc. Suppose we are given an instance of $\bmhpc_{m,r}$. We first describe a gadget graph $G$ based on this instance, and prove that this graph's degeneracy $\kappa(G)$ being $\leq k$ or $>k$ (for some particular value $k$) can be used to determine whether the solution to the \bmhpc instance is $0$ or $1$. Then, along standard lines, we show how a streaming algorithm that solves this boolean version of the  degeneracy problem can be used to design a communication protocol for $\bmhpc_{m,r}$, implying that the space used by the streaming algorithm must be at least the communication complexity of $\bmhpc_{m,r}$. 

\subsection{Construction of the Gadget Graph}\label{sec:redn-graph}
 In \Cref{fig:graph-const}, we show the graph constructed for a particular instance of \bmhpc. We then formally describe the construction and refer to the figure for examples. 
 
\begin{figure}[H]
\centering
\includegraphics[scale=0.6]{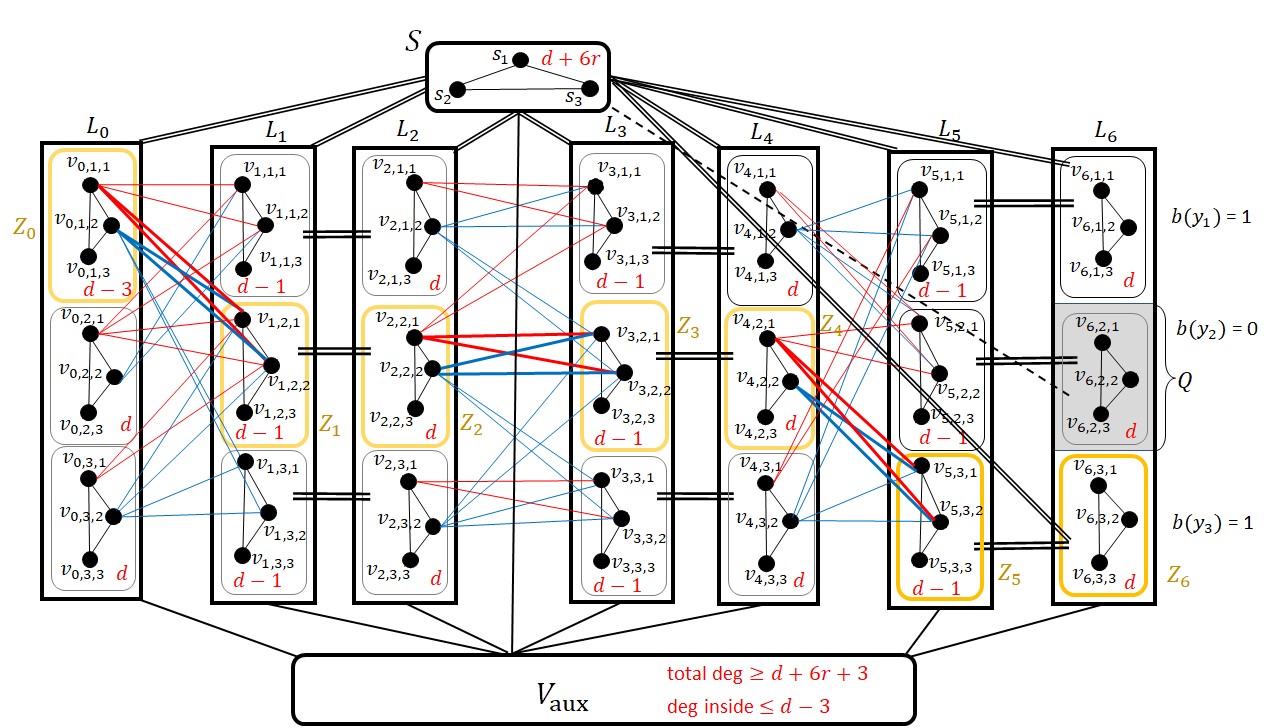}
\caption{\small The gadget graph $G$ for the following instance of $\bmhpc_{3,3}$. $\bfA^1 = \{A^1_{x_1} = \{y_1,y_2\}, A^1_{x_2} = \{y_1,y_2\}, A^1_{x_3} = \{y_2\}\}$. $\bfB^1 =\{ B^1_{x_1} = \{y_2, y_3\}, B^1_{x_2} = \{y_1\}, B^1_{x_3} = \{y_2, y_3\}\}$. $\bfC^2= \{C^2_{y_1} = \{x_1\}, C^2_{y_2} = \{x_1,x_2\}, C^2_{y_3} = \{x_3\}\}$. $\bfD^2 = \{D^2_{y_1} = \{x_1,x_2\}, D^2_{y_2} = \{x_2,x_3\}, D^2_{y_3} = \{x_2,x_3\}\}$. $\bfA^3 = \{A^3_{x_1} = \{y_2\}, A^3_{x_2} = \{y_2,y_3\}, A^3_{x_3} = \{y_1\}\}$. $\bfB^3 =\{ B^3_{x_1} = \{y_1, y_2\}, B^3_{x_2} = \{y_3\}, B^3_{x_3} = \{y_1, y_3\}\}$ (we skip $\bfC^1,\bfD^1,\bfA^2,\bfB^2,\bfC^3,\bfD^3$ because they are redundant). Hence, $z_0=x_1$; $z_1=t^1_{x_1}=y_2$; $z_2=t^2_{y_2}=x_2$; and $z_3=t^3_{x_2}=y_3$. Therefore, the answer to this instance is $b(y_3)=1$.\\
Here, a double line between two sets of nodes (e.g., the one between sets $S$ and $L_0$) signifies that all possible cross edges are present between those two sets. A single line between two sets (e.g., the one between $\Vaux$ and $L_0$) denotes that some edges maybe present between them. A dashed line (e.g., the one between sets $S$ and $Q$) denotes that absolutely no edge is present between them. For each triple, the number (in red) inside the box containing it denotes the degree of each node in the triple. The red and blue edges encode the inputs. The thick edges are between triples representing $z_i$ and $z_{i+1}$.}
\label{fig:graph-const}
\end{figure}

 The graph $G$ has $2r+1$ {\em layers} $L_0, L_1, \ldots L_{2r}$. Each layer $L_{\ell}$ has $3m$ nodes. Edges between two of these layers are always between successive layers $L_{\ell}$ and $L_{\ell+1}$ for some $\ell$.

 \mypar{Layer nodes} For each $\ell \in [0,2r]$ and for each $i\in [m]$, layer $L_\ell$ has $3$ nodes (a \emph{triple}), denoted by $v_{\ell,i,1}, v_{\ell,i,2}$, and $v_{\ell,i,3}$. For $\ell=0$, this triple represents $x_i$. For $\ell\geq 0$, in layers $L_{4\ell+1}$ and $L_{4\ell+2}$ (resp. $L_{4\ell+3}$ and $L_{4\ell+4}$), it represents $y_i$ (resp. $x_i$). For instance, in \Cref{fig:graph-const}, the triple $(v_{3,2,1}, v_{3,2,2},v_{3,2,3})$ represents $x_2$ and $(v_{6,1,1}, v_{6,1,2},v_{6,1,3})$ represents $y_1$. Formally, we add the following set of vertices. $$V_L = \left\{v_{\ell,i,c}:\ell\in [0,2r], i\in [m], c\in \{1,2,3\}\right\}$$

 \mypar{Special nodes} Outside the layers, we have a set $S$ of three special nodes $s_1, s_2, s_3$. 

\smallskip
We now define the edges on the node set $V_L\cup S$.

 \mypar{Edges within each triple} Each triple forms a triangle. Formally, we add the edge set $E_1$ to $G$, where 
$$E_1=\left\{(v_{\ell, i,1}, v_{\ell, i,2}), (v_{\ell, i,2}, v_{\ell, i,3}), (v_{\ell, i,1}, v_{\ell, i,3}): \ell\in [0,2r], i\in [m] \right\}$$

\mypar{Edges between layers $L_{2\ell-1}$ and $L_{2\ell}$} As mentioned above, for each $\ell\in [r]$, either both of the $i$th triples in layers $L_{2\ell-1}$ and $L_{2\ell}$ represent $x_i$ or both represent $y_i$. We add all $9$ cross edges between these two triples representing the same element. E.g., in \Cref{fig:graph-const}, we have all cross edges (denoted by a double line) between the triples $(v_{1,1,1}, v_{1,1,2}, v_{1,1,3})$ and $(v_{2,1,1}, v_{2,1,2}, v_{2,1,3})$, both of which represent $y_1$. Formally, we add the edge set $E_2$ to $G$, where 
$$E_2=\left\{\left\{v_{2\ell-1, i, c}: c\in \{1,2,3\}\right\}\times \left\{v_{2\ell, i, c}: c\in \{1,2,3\}\right\}: \ell\in [r], i\in [m] \right\}$$

\mypar{Edges between layers $L_{4\ell}$ and $L_{4\ell+1}$} The edges connecting layers of the form $L_{4\ell}$ and $L_{4\ell+1}$ encode the inputs $\bfA^{2\ell+1}$ and $\bfB^{2\ell+1}$. In particular, for each $i\in [m]$, the edges leaving $v_{4\ell, i, 1}$ and $v_{4\ell, i, 2}$ encode the sets $A_{x_i}^{2\ell+1}$ and $B_{x_i}^{2\ell+1}$ respectively. Precisely, if the set $A_{x_i}^{2\ell+1}$ contains the element $y_j$, then we join $v_{4\ell, i, 1}$ to both the nodes $v_{4\ell+1, j, 1}$and $v_{4\ell+1, j, 2}$. Similarly, if $B_{x_i}^{2\ell+1}$ contains $y_j$, then we join $v_{4\ell, i, 2}$ to both $v_{4\ell+1, j, 1}$ and  $v_{4\ell+1, j, 2}$ (for example, see the families ($\bfA^1$,$\bfB^1$) and ($\bfA^3$,$\bfB^3$) in \Cref{fig:graph-const}, and the corresponding edges between $L_0$ and $L_1$ and between $L_4$ and $L_5$ respectively). Formally, we add the edge sets $E_A$ and $E_B$ to $G$, where

\begin{itemize}
\item $E_A= \left\{\left\{v_{4\ell, i, 1}\right\}\times \left\{v_{4\ell+1,j,1}, v_{4\ell+1, j, 2}\right\}: y_j\in A_{x_i}^{2\ell+1}, \ell\in [0,\floor{(r-1)/2}], j\in [m]\right\}$

\item $E_B= \left\{\left\{v_{4\ell, i, 2}\right\}\times \left\{v_{4\ell+1,j,1}, v_{4\ell+1, j, 2}\right\}: y_j\in B_{x_i}^{2\ell+1}, \ell\in [0,\floor{(r-1)/2}], j\in [m]\right\}$
\end{itemize}

\mypar{Edges between layers $L_{4\ell+2}$ and $L_{4\ell+3}$} The edges between layers of the form $L_{4\ell+2}$ and $L_{4\ell+3}$ similarly encode the inputs $\bfC^{2\ell+2}$ and $\bfD^{2\ell+2}$. The construction is same as above, except with $A_{x_i}^{2\ell+1}$ and $B_{x_i}^{2\ell+1}$ replaced by $C_{y_i}^{2\ell+2}$ and $D_{y_i}^{2\ell+2}$ respectively, and $y_j$ replaced by $x_j$ (again, see the families $\bfC^2$ and $\bfD^2$ in the example in \Cref{fig:graph-const} and the corresponding edges between $L_2$ and $L_3$). Formally, we add the edge sets $E_C$ and $E_D$ to $G$, where

\begin{itemize}
\item $E_C= \left\{\left\{v_{4\ell+2, i, 1}\right\}\times \left\{v_{4\ell+3,j,1}, v_{4\ell+3, j, 2}\right\}: x_j\in C_{y_i}^{2\ell+2}, \ell\in [0,\floor{r/2}-1], j\in [m]\right\}$

\item $E_D= \left\{\left\{v_{4\ell+2, i, 2}\right\}\times \left\{v_{4\ell+3,j,1}, v_{4\ell+3, j, 2}\right\}: x_j\in D_{y_i}^{2\ell+2}, \ell\in [0,\floor{r/2}-1], j\in [m]\right\}$
\end{itemize}

\mypar{Edges on special nodes} Recall that each element in $\cX$ and $\cY$ has a bit associated with it. The special nodes $s_1,s_2,s_3$ are connected to each other and also to all the layer nodes except those representing bit-$0$-elements in the last layer $L_{2r}$ (in  \Cref{fig:graph-const}, this is the triple representing $y_2$ in $L_6$). Formally, if $r$ is even, let $Q := \left\{v_{2r, i, c}: b(x_i) = 0,~c\in \{1,2,3\}\right\}$, and otherwise, $Q := \left\{v_{2r, i, c}: b(y_i) = 0,~c\in \{1,2,3\}\right\}$. Then, we add the edge set $E_S$ to $G$, where 
$$E_S=\{(s_1,s_2), (s_2,s_3), (s_1,s_3)\}\cup \left(\{s_1, s_2, s_3\}\times (V_L \setminus Q)\right)$$

Let $V' := V_L \cup S$ and $E' := E_1\cup E_2 \cup E_S\cup E_A\cup E_B \cup E_C\cup E_D$. The graph constructed so far is $G'=(V',E')$.

\mypar{Auxiliary edges} We add some auxiliary vertices $\Vaux$ and edges $\Eaux \subseteq \Vaux\times(V'\cup \Vaux)$ to obtain the final graph 
$$G=\left(V' \cup V_{\text{aux}}, E' \cup E_{\text{aux}} \right)$$

\noindent
where we ensure the following degrees for $d:=6mr+3m$ (in \Cref{fig:graph-const}, these degrees are written in the boxes containing the triples). 

\begin{itemize}
    \item $\deg(v_{0,1,1})= \deg(v_{0,1,2}) = \deg(v_{0,1,3}) = d-3$ 
   \item For each $i\in [2,m]$, $c\in\{1,2,3\}$, $\deg(v_{0,i,c})=d$.
   \item For each $\ell\in [r], i\in [m], c\in \{1,2,3\}$, $\deg(v_{2\ell-1,i,c})=d-1$.
    \item For each $\ell\in [r], i\in [m], c\in \{1,2,3\}$, $\deg(v_{2\ell,i,c})=d$.
    \item $\deg(s_1)=\deg(s_2)=\deg(s_3)=d+6r$
    \item For each $u\in \Vaux$, $\deg(u)\geq d+6r+3$. 
\end{itemize}

\Cref{lem:redn-aux-constr} states that such a construction of $\Vaux$ and $\Eaux$ is possible with just the knowledge of the vertex degrees in $G'$ (rather than the actual graph). We shall use this fact in our reduction. The proof with the explicit construction is deferred to \Cref{app:proof-auxnode-redn}.

\begin{lemma}\label{lem:redn-aux-constr}
 Given the degree of each node in $G'=(V',E')$, we can add vertices $V_{\text{aux}}$ and edges $E_{\text{aux}}\subseteq \Vaux \times (V'\cup \Vaux)$ so that the graph $G=\left(V' \cup V_{\text{aux}}, E' \cup E_{\text{aux}}\right)$ has the above vertex-degrees and $\Theta(mr)$ nodes in total. Furthermore, the degree of any vertex in the induced subgraph $G[\Vaux]$ is $\leq d-3$.
\end{lemma}

This completes the construction of the graph $G$. 

\subsection{The Analysis}

In this section, we prove the following lemma that says that we can use the degeneracy value of the gadget graph constructed in the last section to determine the solution to the corresponding instance of \bmhpc. 

\begin{lemma}\label{lem:redn-degen-yesno}
     Given an instance of $\bmhpc_{m,r}$, let $G$ denote the gadget graph constructed in \Cref{sec:redn-graph}. If $b(z_r)=0$, then $\kappa(G)\geq d-2$. Otherwise, $\kappa(G)\leq d-3$.
\end{lemma}

We set up some notations to prove the above lemma.

\mypar{Triples corresponding to pointers} We are particularly interested in the triples of nodes that represent the pointers $z_i$ in the relevant layers. In layer $L_0$, we are interested in the triple representing the first pointer $z_0$; call it $Z_0$. Then in \emph{each} of the layers $L_1$ and $L_2$, we are interested in the triple representing the second pointer $z_1$; call these triples $Z_1$ and $Z_2$ respectively. In general, in layers $L_{2\ell-1}$ and $L_{2\ell}$, the triples that represent the pointer $z_\ell$ are called $Z_{2\ell-1}$ and $Z_{2\ell}$ respectively. In other words, the triple $Z_{\ell}$ corresponds to the pointer $z_{\ceil{\ell/2}}$ (see \Cref{fig:graph-const} for example. $Z_0$ is the triple representing $z_0=x_1$; both $Z_1$ and $Z_2$ represent $z_1=y_2$; both $Z_3$ and $Z_4$ represent $z_2=x_2$; and finally both $Z_5$ and $Z_6$ represent $z_3=y_3$.)

Let us now formalize this. Define the function $i^*:[0,2r]\to [m]$ as follows:

For $\ell\in [0,2r]$, let $\ell' := \ceil{\ell/2}$. Then $i^*(\ell)$ is the unique index $i$ such that 
$$z_{\ell'}= 
\begin{cases}
x_{i} & \text{if $\ell'$ is even}\\
y_{i} & \text{if $\ell'$ is odd}
\end{cases}
$$

Then $Z_\ell:= (v_{\ell,i^*(\ell),1}, v_{\ell,i^*(\ell),2},v_{\ell,i^*(\ell),3})$, i.e., the triple representing $z_{\ceil{\ell/2}}$ in layer $L_\ell$.

\mypar{High-level proof overview} To prove \Cref{lem:redn-degen-yesno}, we perform the peeling algorithm on the graph $G$. We shall show that the first $6r+3$ iterations of this algorithm basically chase the hidden pointers over the successive layers: to be precise, the first three iterations remove the triple $Z_0$, the next three remove $Z_1$, and so on and so forth until iterations $6r+1,6r+2$, and $6r+3$ remove $Z_{2r}$ (see \Cref{fig:first-iter,fig:odd-layer-rem,fig:even-layer-rem} for a sketch of how this proceeds). Then we have two cases as follows. 

If $b(z_r)=0$, then removal of $Z_{2r}$ does not reduce the degrees of the nodes in $S$. We show that in this case, immediately after the removal of $Z_{2r}$, the peeling algorithm must remove a node with degree $\geq d-2$ (\Cref{fig:nocase}). Thus, by \Cref{fact:peelingmin-degen}, we must have $\kappa(G)\geq d-2$.

Otherwise, if $b(z_r)=1$, then $Z_{2r}$ is indeed connected to the nodes $s_1,s_2,s_3$. We show that removal of $Z_{2r}$ reduce their degrees enough so that the peeling algorithm must remove $S$ in the next three iterations. This leads to the peeling of all the remaining nodes in the graph such that during deletion, their degree is at most $d-3$ (see \Cref{fig:yescase}). Then, by \Cref{fact:peelingmin-degen}, we must have $\kappa(G)\leq d-3$. This will complete the proof of the lemma. 

\mypar{Proof details} To formally prove \Cref{lem:redn-degen-yesno}, we first show that the peeling algorithm indeed peels off the $Z_\ell$'s one by one, and that some invariants on the vertex-degrees hold after these deletions. We start with a technical claim that we use multiple times in the proof. 

\begin{claim}\label{clm:zl-del}
Suppose that in a graph $H$, the min-degree nodes are exactly a set of three nodes $U=\{u_1,u_2,u_3\}$, which are all adjacent to each other, and have degree $\delta$ each (all other nodes in $H$ have degree $\geq \delta+1$). Then, the peeling algorithm on $H$ removes the nodes in $U$ (in some order) in the first three iterations. Further, during deletion, each $u_i$ has degree $\leq \delta$. 
\end{claim}

\begin{proof}
Since $U$ is the exact set of min-degree nodes in $H$, iteration $1$ of the peeling algorithm must delete one of the $u_i$'s. WLOG, say it deletes $u_1$. Then, $\deg(u_2)$ and $\deg(u_3)$ drop to $\delta-1$ (since they were neighbors of $u_1$); at the same time, since the degree of any other node was at least $\delta+1$, it can drop to at least $\delta$. Therefore, $u_2$ and $u_3$ remain the min-degree nodes and one of them must be removed in iteration $2$. WLOG, say it's $u_2$. As a result, the degree of $u_3$ drops to $\delta-2$ (since $u_2$ was its neighbor) while the degree of any other node remains at least $\delta-1$. Hence, iteration $3$ deletes $u_3$. Clearly, when each $u_i$ was deleted, its degree was $\leq \delta$.
\end{proof}

We now list the invariants and prove that they are maintained.

\begin{lemma}\label{lem:redn-inv}
For each $\ell\in [0,2r]$, the following invariants hold:
    \begin{enumerate}
    \item[Inv.~(1).] Just before iteration $3\ell+1$, any $s\in S$ has $\deg(s)=d+6r-3\ell$ and any $u\in \Vaux$ has $\deg(u)\geq d+6r+3-3\ell$.
     \item[Inv.~(2).] If $\ell$ is even, then just before iteration $3\ell+1$, 
    \begin{enumerate}
        \item each node in the triple $Z_{\ell}$ has degree $d-3$ 
        \item for any vertex $v\in L_{j}$ for $j\leq \ell-1$, we have $\deg(v)\geq d-2$
        \item for any vertex $v\in L_{j}$ for an odd $j \geq \ell+1$, we have $\deg(v)= d-1$ 
        \item for any vertex $v\not\in Z_{\ell}$ in a layer $L_j$ for an even $j\geq \ell$, we have $\deg(v) = d$.
    \end{enumerate}
    \item[Inv.~(3).] If $\ell$ is odd, then just before iteration $3\ell+1$,
     \begin{enumerate}
         \item the triple $Z_{\ell}$ has degree $(d-3,d-3,d-1)$
          \item for any vertex $v\not\in Z_{\ell}$ in layer $L_{\ell}$ or in $L_j$ for any $j\leq \ell-2$, we have $\deg(v)\geq d-2$
        \item for any vertex $v\in L_j$ for an odd $j\geq \ell+2$, we have $\deg(v)= d-1$ 
        \item for any vertex $v\in L_j$ for an even $j\geq \ell-1$, we have $\deg(v)= d$.
     \end{enumerate}
      \item[Inv.~(4).] Iterations $(3\ell+1,3\ell+2,3\ell+3)$ delete the vertices in the triple $Z_\ell$ in some order. Furthermore, during deletion, each vertex has degree $\leq d-3$.
\end{enumerate}  
\end{lemma}

\begin{proof}
  We prove this by induction on $\ell$. 
  
  \mypar{Base case} Consider the base case of $\ell=0$. Since $\ell$ is even, we need to ensure Invs.~(1), (2), and (4) hold. Indeed, Inv.~(3) is vacuously true. Observe that by definition, $Z_0=(v_{0,1,1}, v_{0,1,2},v_{0,1,3})$. Before iteration $1$, i.e., initially (see \Cref{fig:graph-const}), the construction ensures that each node in $Z_0$ has degree $d-3$. Hence, Inv.~2(a) holds. Inv.~2(b) holds vacuously. Our construction sets the degree of any vertex $v\in L_j$ for odd $j$ to be $d-1$, and that of any vertex outside $Z_0$ in layer $L_j$ for even $j$ to be $d$. Thus, Inv.~2(c) and Inv.~2(d) also hold. The initial degree of any node in $S$ is $d+6r$ and that of any node in $\Vaux$ is at least $d+6r+3$. This means Inv.~(2) holds as well.

  To see that Inv.~(4) holds, note that initially the min-degree vertices in the current graph are precisely the nodes in $Z_{0}$, each with degree $d-3$. All remaining nodes have degree $\geq d-1$. Hence, the premise of \Cref{clm:zl-del} holds with $H=G, U=Z_0$, and $\delta=d-3$. Therefore, by \Cref{clm:zl-del}, iterations $(1,2,3)$ delete $Z_0$ in some order, and during deletion, the degrees of these nodes are $\leq d-3$. Hence, Inv.~(4) also holds, and we are done with the base case.

\mypar{Inductive Step} Assume by induction hypothesis (henceforth ``IH'') that for some $\ell\in [0,2r-1]$, all the invariants hold for all values $\leq \ell$. Then we show that they hold for $\ell+1$ as well. 

First consider Inv.~(1). By IH Inv.~(1), any node $s\in S$ had degree $d+6r-3\ell$ before iteration $3\ell+1$. By IH Inv.~(4), the next three iterations delete nodes in $Z_\ell$, each of which shared an edge with each $s\in S$ (since $\ell\leq 2r-1$). Hence, $\deg(s)$ drops to $d+6r-3(\ell+1)$ before iteration $3(\ell+1)+1$. By the same invariant, any $u\in \Vaux$ had degree $\geq d+6r+3-3\ell$ before iteration $3\ell+1$. The deletion of $Z_\ell$ can reduce its degree by at most $3$, implying that $\deg(u)\geq d+6r+3-3(\ell+1)$ before iteration $3(\ell+1)+1$. Therefore, Inv.~(1) holds for $\ell+1$. 

For the remaining invariants, we analyze separate cases for whether $\ell+1$ is even or odd.  

\mypar{Case 1: $\ell+1$ is odd} We need to show that Invs.~(3) and (4) hold for $\ell+1$. Inv.~(2) is vacuously true. 

We know that $\ell$ is even, and by IH, Invs.~(2), and (4) hold for $\ell$. By Inv.~(2), we know the situation before iteration $3\ell+1$, and by Inv.~(4), we know that the iterations $(3\ell+1,3\ell+2,3\ell+3)$ remove $Z_{\ell}$, leading to the current graph just before iteration $3(\ell+1)+1$ (see \Cref{fig:first-iter} for an example with $\ell=0$ and \Cref{fig:even-layer-rem} for an example with $\ell=2$ for the graph given in \Cref{fig:graph-const}).

\begin{figure}[H]
\centering
\includegraphics[width=0.65\textwidth, height=6cm]{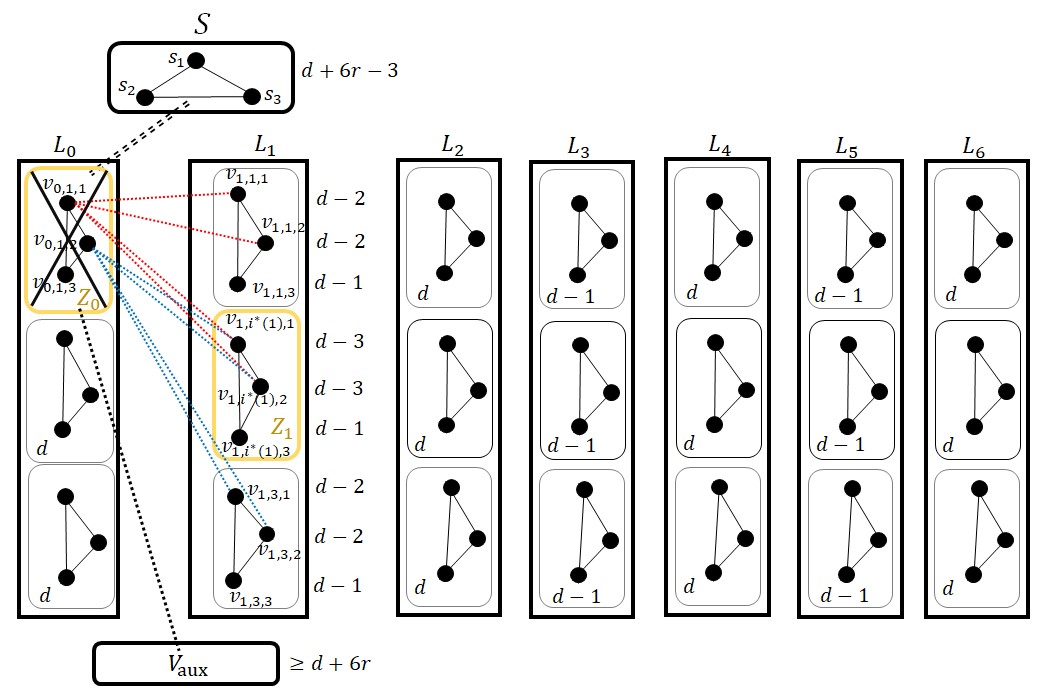}
\caption{\small The status of the graph just before iteration $4$. $Z_0$ was deleted in the last three iterations, and hence crossed out. The dotted lines denote edges that were deleted as a result. A double dotted-line between two sets of nodes signifies that originally all cross edges were present between them. Next to a node, its current degree is given. If all nodes in a set have the same degree, the value is given next to the set.}
\label{fig:first-iter}
\end{figure}

\begin{figure}[H]
\centering
\includegraphics[width=0.65\textwidth, height=6cm]{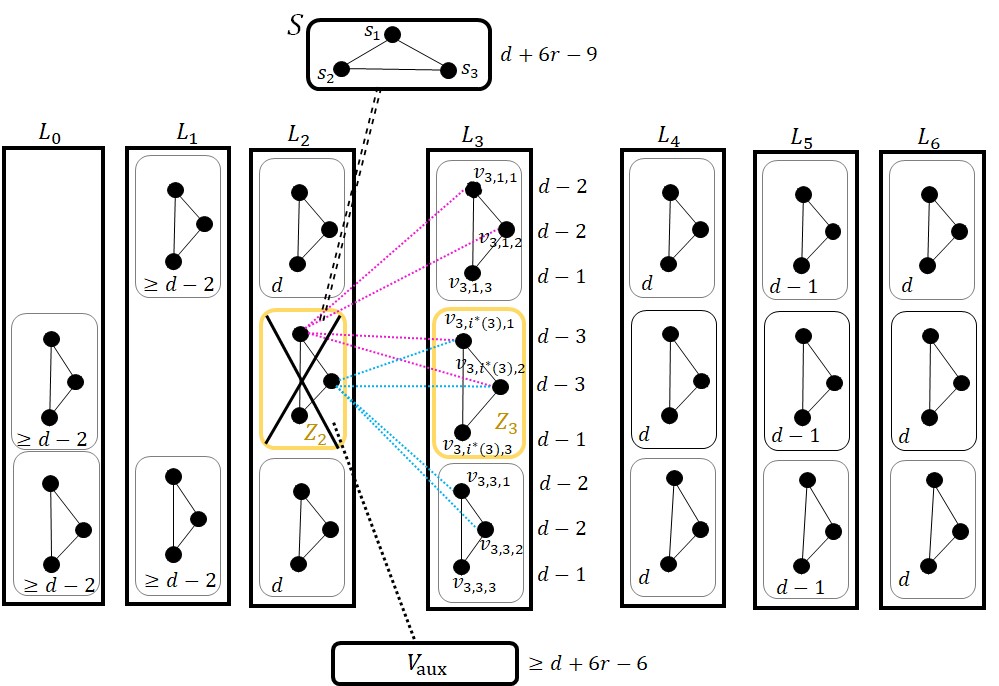}
\caption{\small The status of the graph just before iteration $10$. $Z_2$ was deleted in the last three iterations, and hence crossed out. The dotted lines denote edges that were deleted as a result. A double dotted-line between two sets of nodes signifies that originally all cross edges were present between them. Next to a node, its current degree is given. If all nodes in a set have the same degree, the value is given next to the set.}
\label{fig:even-layer-rem}
\end{figure}

 We check which nodes are affected by the removal of $Z_{\ell}$. In the original graph, the only neighbors of $Z_{\ell}$ in layer $L_{\ell-1}$ (if exists)  were the nodes in $Z_{\ell-1}$, which were removed even before $Z_\ell$ (by IH Inv.~(4)). Thus, its only neighbors in the current graph are in layer $L_{\ell+1}$, $S$, or $\Vaux$. 
 
 Suppose $\ell$ is of the form $4a+2$ (see \Cref{fig:even-layer-rem}). Then the nodes in layer $\ell+1$ represent elements in $\cX$. Recall that $Z_{\ell}$ corresponds to the element $z_{\ell'}\in \cY$ where $\ell'=\ell/2$. By construction, the edges from $Z_{\ell}$ connect to nodes that represent elements in $C_{z_{\ell'}}^{\ell'+1}$ and $D_{z_{\ell'}}^{\ell'+1}$. Consider any $x_i\in \cX$. We have three cases: (i) $x_i$ is in exactly one of the sets $C_{z_{\ell'}}^{\ell'+1}$ and $D_{z_{\ell'}}^{\ell'+1}$, (ii) $x_i$ is in neither of them, and (iii) $x_i$ is in both, i.e., $x_i=z_{\ell'+1}$. If $\ell$ is of the form $4a$ (see \Cref{fig:first-iter}), we make the analogous arguments replacing $\cX$ by $\cY$, $C$ and $D$ by $A$ and $B$, and $x_i$ by $y_i$. We now analyze each of these cases. 
 
 By IH Inv.~2(c), before iteration $3\ell+1$, each node in the triple $(v_{\ell+1,i,1}, v_{\ell+1,i,2}, v_{\ell+1,i,3})$ had degree $d-1$. Therefore, after the removal of $Z_\ell$, i.e., just before iteration $3(\ell+1)+1$, it can have the following degrees. In case (i), the triple has degrees $(d-2, d-2, d-1)$. This is because our construction ensures that exactly one of the vertices $v_{\ell,i^*(\ell),1}$ and $v_{\ell,i^*(\ell),2}$ had edges to both $v_{\ell+1,i,1}$ and $v_{\ell+1,i,2}$, and the other had none. In case (ii), it has degrees $(d-1,d-1,d-1)$. This is because none of the nodes in $Z_{\ell}$ had any edge to the triple representing $x_i$. In case (iii), the triple representing $x_i$ is precisely the triple $Z_{\ell+1}$ and has degrees $(d-3,d-3,d-1)$. This follows since both nodes $v_{\ell,i^*(\ell),1}$ and $v_{\ell,i^*(\ell),2}$ had edges to both $v_{\ell+1,i,1}$ and $v_{\ell+1,i,2}$. 
 
 The above analysis immediately implies that Inv.~3(a) holds for $\ell+1$. Now, for possible values of $r$, consider nodes in $L_r$. For $r=\ell+1$, we see that any node $v\not\in Z_{\ell+1}$ in layer $L_r$ has degree $\geq d-2$. For $r\leq \ell-1=(\ell+1)-2$, by IH Inv.~2(b), a vertex in $L_r$ had degree $\geq d-2$ before iteration $3\ell+1$, which still holds after iteration $3\ell+3$: this is because, as noted above, removal of $Z_{\ell}$ does not affect any of them. Therefore, Inv.~3(b) also holds for $\ell+1$. 
 
 For an odd $j\geq \ell+3=(\ell+1)+2$, by IH Inv.~2(c), any vertex in $L_j$ had degree $d-1$ before iteration $3\ell+1$ and it is still true after iteration $3\ell+3$ since removal of $Z_{\ell}$ cannot affect them. Thus, Inv.~3(c) holds for $\ell+1$. Finally, for an even $j\geq \ell=(\ell+1)-1$, by IH Inv.~2(d), any vertex $v\not\in Z_\ell$ had degree $d$ before iteration $3\ell+1$; since removal of $Z_{\ell}$ doesn't affect the other vertices in these layers, their degrees still have that value and Inv.~3(d) holds for $\ell+1$. Therefore, we have shown that Inv.~(3) holds for $\ell+1$.

  It remains to prove Inv.~(4). Since Inv.~(3) holds for $\ell+1$, we know that just before iteration $3(\ell+1)+1$, the triple $Z_{\ell+1}$ has degrees $(d-3,d-3,d-1)$, and any other node has degree at least $d-2$. This means the min-degree nodes in the current graph are in $Z_{\ell}$; precisely, they are $v_{\ell,i^*(\ell),1}$ and $v_{\ell,i^*(\ell),2}$, each having degree $d-3$. WLOG, say iteration $3\ell+1$ deletes $v_{\ell,i^*(\ell),1}$. Then, $\deg(v_{\ell,i^*(\ell),2})$ drops to $d-4$, and $\deg(v_{\ell,i^*,3})$ drops to $d-2$. Notice that the only other nodes that can be affected by the deletion are in either layer $L_\ell$ or layer $L_{\ell+2}$. By Inv.~3(d) for $\ell+1$, since any node in these layers had degree $d$ before iteration $3(\ell+1)+1$, their degree can drop to at least $d-1$ after the iteration. Also by Inv.~3, the vertices in other layers had degrees $\geq d-2$ before the iteration, and since they are unaffected, their degrees remain so. Hence, the current graph has the unique min-degree node $v_{\ell,i^*,2}$, which is removed by iteration $3\ell+2$. Now, $\deg(v_{\ell,i^*,3})$ drops to $d-3$. But degrees of vertices in layers $L_{\ell}$ or $L_{\ell+2}$ can now drop to at least $d-2$, while any other node maintains degree $\geq d-2$ as it's unaffected by the deletion. Hence, $v_{\ell,i^*,3}$ is the unique min-degree node in the current graph and is removed by iteration $3(\ell+1)+2$. Also observe that when each vertex was peeled, its degree was $\leq d-3$. Thus, we have proven that Inv.~(4) holds for $\ell+1$.

\mypar{Case 2: $\ell+1$ is even} We need to show that Invs.~(2) and (4) hold for $\ell+1$. Inv.~(3) holds vacuously.

 We know that $\ell$ is odd, and by IH, Invs.~(3) and (4) hold for $\ell$. Analgous to the last case, by Inv.~(3), we know the situation before iteration $3\ell+1$, and by Inv.~(4), we know that the iterations $(3\ell+1,3\ell+2,3\ell+3)$ remove $Z_{\ell}$, leading to the current graph just before iteration $3(\ell+1)+1$ (see \Cref{fig:odd-layer-rem} for an example with $\ell=1$ for the graph in \Cref{fig:graph-const}). 

\begin{figure}[H]
\centering
\includegraphics[width=0.65\textwidth, height=6cm]{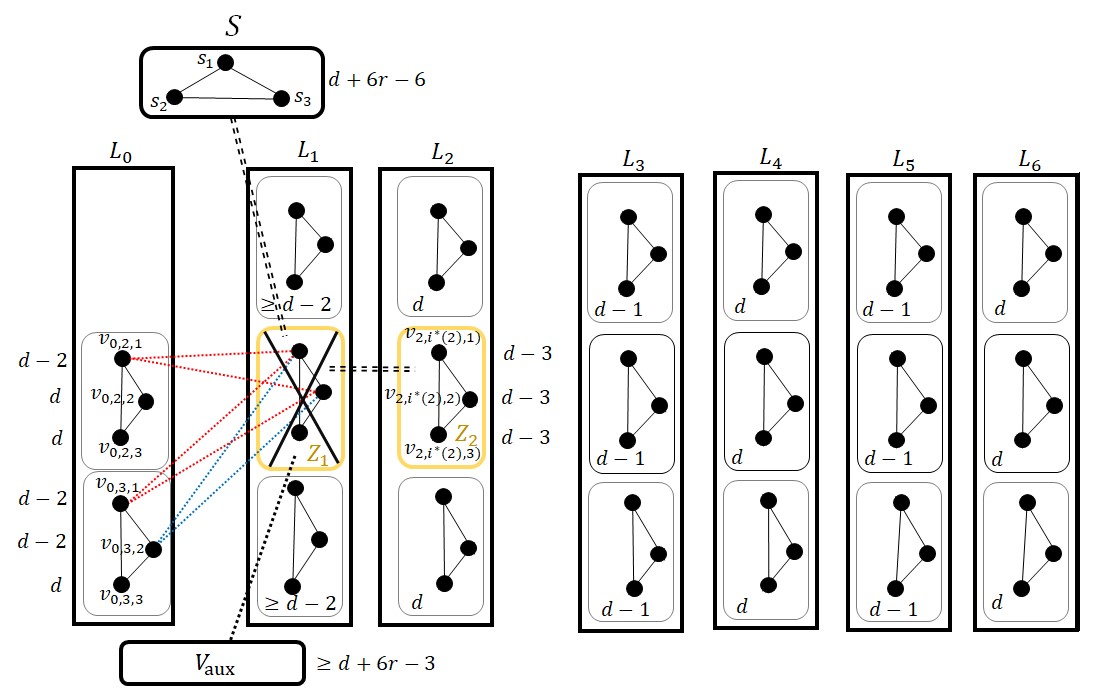}
\caption{\small The status of the graph just before iteration $7$. $Z_1$ was deleted in the last three iterations, and hence crossed out. The dotted lines denote edges that were deleted as a result. A double dotted-line between two sets of nodes signifies that originally all cross edges were present between them. Next to a node, its current degree is given. If all nodes in a set have the same degree, the value is given next to the set.}
\label{fig:odd-layer-rem}
\end{figure}

We check which nodes might be affected by the deletion of $Z_{\ell}$ in the last three iterations. By construction, $Z_{\ell}$ has edges to $Z_{\ell+1}$, to some vertices in $L_{\ell-1}$, to $S$, or to $\Vaux$. 

First consider vertices in $Z_{\ell+1}$. By IH Inv.~3(d), these nodes had degree $d$ before iteration $3\ell+1$. Further, all of them were adjacent to each node in $Z_\ell$. Hence, after $Z_{\ell}$ gets deleted, i.e., just before iteration $3(\ell+1)+1$, each of them has degree $d-3$. This shows that Inv.~2(a) holds for $\ell+1$. 

 Now consider the nodes in layer $L_{\ell-1}$. Suppose $\ell$ is of the form $4a+1$. Then the nodes in $L_{\ell-1}$ represent elements in $\cX$ and those in $L_\ell$ represent elements in $\cY$. Recall that $Z_{\ell}$ corresponds to the element $z_{\ell'}\in \cY$ where $\ell'=\ceil{\ell/2}$. Consider any $x_i\in \cX$. We have three cases: (i) $z_{\ell'}$ is in exactly one of the sets $A_{x_i}^{\ell'-1}$ and $B_{x_i}^{\ell'-1}$, (ii) $z_{\ell'}$ is in neither of them, and (iii) $z_{\ell'}$ is in both, i.e., $t_{x_i}^{\ell'-1}=z_{\ell'}$. If $\ell$ is of the form $4a+3$, we make the analogous arguments exchanging $\cX$ and $\cY$, replacing $A$ and $B$ by $C$ and $D$, and $x_i$ by $y_i$. We analyze each case.
 
 By IH Inv.~3(d), the triple representing $x_i$ in layer $\ell-1$ had degree $(d,d,d)$ before iteration $3\ell+1$. After $Z_{\ell}$ is removed, the degrees of such a triple can be the following. In case (i), it becomes either $(d-2,d,d)$ or $(d,d-2,d)$ (in \Cref{fig:odd-layer-rem}, the triple $(v_{0,2,1}, v_{0,2,2}, v_{0,2,3})$ is of this type). This is because exactly one of the vertices $v_{\ell-1,i,1}$ and $v_{\ell-1,i,2}$ had edges to both $v_{\ell,i^*(\ell),1}$ and $v_{\ell,i^*(\ell),2}$, and the other had none. In case (ii), it remains $(d,d,d)$ since neither $v_{\ell-1,i,1}$ nor $v_{\ell-1,i,2}$ had any edge to $Z_{\ell}$. In case (iii), it becomes $(d-2,d-2,d)$ (in \Cref{fig:odd-layer-rem}, the triple $(v_{0,3,1}, v_{0,3,2}, v_{0,3,3})$ is of this type). This follows since both nodes $v_{\ell-1,i,1}$ and $v_{\ell-1,i,2}$ had edges to both $v_{\ell,i^*(\ell),1}$ and $v_{\ell,i^*(\ell),2}$. 
 
 For possible values of $j$, consider the nodes in layer $L_j$ before iteration $3(\ell+1)+1$. From the above analysis, for $j=\ell-1$, any node in $L_j$ has degree $\geq d-2$. Further, by IH Inv.~3(b), for $j=\ell$ or $r\leq \ell-2$, any vertex $v\not\in Z_{\ell}$ had $\deg(v)\geq d-2$ before iteration $3\ell+1$. Since these nodes cannot be affected by the removal of $Z_{\ell}$, their degrees remain so. Therefore, we have shown that any vertex in layer $L_j$ for $j\leq \ell=(\ell+1)-1$ has degree $\geq d-2$, which means Inv.~2(b) holds for $\ell+1$. 

 By IH Inv.~3(c), for an odd $r \geq \ell+2$, any vertex in $L_j$ had degree $d-1$ before iteration $3\ell+1$. Removal of $Z_{\ell}$ doesn't affect them and hence, their degrees remain so. Thus, Inv.~2(c) holds for $\ell+1$. Finally, for an even $r\geq \ell+1$, IH Inv.~3(d) says that a vertex $v$ in $L_r$ had degree $d$ before iteration $3\ell+1$. The deletion of $Z_{\ell}$ can only affect $Z_{\ell+1}$ among these nodes. So a vertex $v\not\in Z_{\ell+1}$ in $L_j$ still has degree $d$. Therefore, Inv.~2(d), and in fact, Inv.~(2) holds for $\ell+1$.

It only remains to show that Inv.~(4) golds for $\ell+1$. Since Inv.~(2) holds for $\ell+1$, we know that just before iteration $3(\ell+1)+1$, each node in $Z_{\ell+1}$ has degree $d-3$, and any other node has degree at least $d-2$. Hence, the premise of \Cref{clm:zl-del} holds with $H$ as the current graph, $U=Z_{\ell+1}$, and $\delta=d-3$. Therefore, by \Cref{clm:zl-del}, the first three iterations on $H$, i.e., iterations $(3(\ell+1)+1,3(\ell+1)+2,3(\ell+1)+3)$ delete $Z_{\ell+1}$ in some order, and during deletion, the degree of these nodes are $\leq d-3$. Hence, Inv.~(4) holds for $\ell+1$.

Thus, we have proven that all the invariants hold for $\ell+1$ irrespective of its parity. This completes the proof by induction.
\end{proof}

 We are now ready to prove \Cref{lem:redn-degen-yesno}.

\begin{proof}[Proof of \Cref{lem:redn-degen-yesno}]
Indeed, we have two cases.

\mypar{Case 1: \boldmath{$b(z_r)=0$}} We show that the min-degree peeling algorithm must peel a vertex of degree $\geq d-2$ at some iteration. Then, \Cref{fact:peelingmin-degen} implies that $\kappa(G)$ must be at least $d-2$. 

\begin{figure}[H]
    \centering
    \includegraphics[scale=0.5]{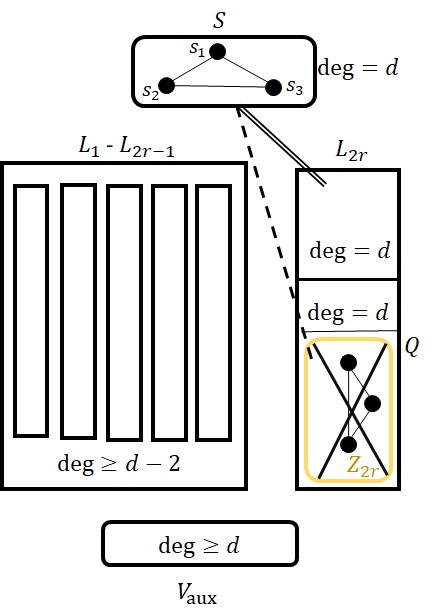}
    \caption{\small The graph after iteration $6r+3$ for the case $b(z_r)=0$. A double line between two sets of nodes signifies that all possible cross edges are present between them. A dashed line denotes that no edge is present between them.}
    \label{fig:nocase}
\end{figure}

    By \Cref{lem:redn-inv} Inv.~(4), the first $6r+3$ iterations of the peeling algorithm remove all the nodes in $Z_0, Z_1,\ldots, Z_{2r}$. Consider the status of the graph after $Z_{2r}$ is removed (see \Cref{fig:nocase}).  By the same lemma, all nodes in layers $L_0,\ldots, L_{2r}$ had degree $\geq d-2$ before $Z_{2r}$ was removed (i.e., just before iteration $6r+1$). Since its removal cannot affect any of those nodes
    (in these layers it only had edges to $Z_{2r-1}$, which was previously removed), their degrees remain so. Now consider nodes in $S$. By \Cref{lem:redn-inv} Inv.~1, just before iteration $6r+1$, each of them
    had degree $d$. After the next $3$ iterations remove $Z_{2r}$, their degrees remain $d$; this is because they were not connected to $Z_{2r}$ as $b(z_r)=0$. Again, by the same invariant, just before iteration $6r+1$, the nodes in $\Vaux$ had degree $\geq d+3$. Hence, after iteration $6r+3$, they must have degree $\geq d$. Therefore, the min-degree node in the current graph has degree $\geq d-2$, implying that iteration $6r+4$ of the peeling algorithm must peel a node of degree $\geq d-2$.

\mypar{Case 2: \boldmath{$b(z_r)=1$}} Define the ordering $\sigma$ of nodes in $G$ as follows. It agrees with the order in which the min-degree peeling algorithm peels the nodes until iteration $6r+6$. Next, $\sigma$ contains the chunk of the remaining nodes in layers $L_0,\ldots,L_{2r-1}$. They are followed by the remaining nodes in layer $L_{2r}$. Finally we place the auxiliary nodes $\Vaux$. The vertices in each of these last three chunks can be placed in any arbitrary order within the chunk. We show that $\odeg_{\sigma}(v)\leq d-3$ (see \Cref{def:odegordering}) for each $v$.

\begin{figure}[H]
    \centering
    \includegraphics[scale=0.5]{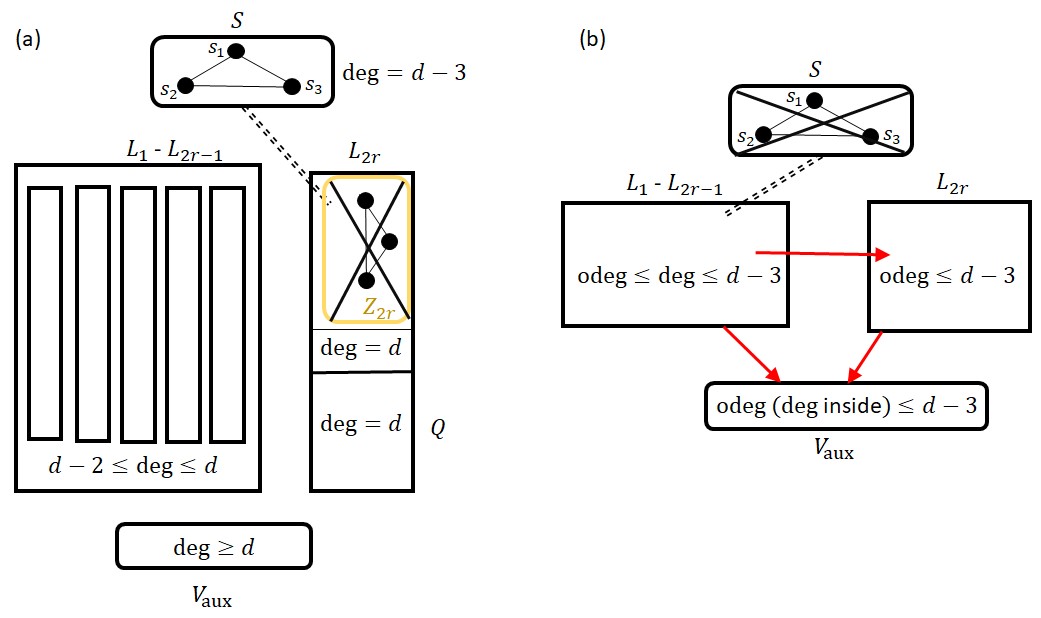}
    \caption{\small For the case $b(z_r)=1$: (a) the graph after iteration $6r+3$. $Z_{2r}$ was deleted in the last three iterations and hence crossed out. (b) The oriented red edges show a $(d-3)$-ordering of the nodes remaining after iteration $6r+6$. A dotted double line between two sets signifies that originally all cross edges were present between them.}
    \label{fig:yescase}
\end{figure}

As in the last case, by Inv.~(4) of \Cref{lem:redn-inv}, the first $6r+3$ iterations of the peeling algorithm remove all nodes in $Z_0,Z_1,\ldots,Z_{2r}$. Further, each time such a node is removed, its degree in the remaining graph is at most $d-3$. Again, by Inv.~(1), just before iteration $6r+1$, the nodes in $S$ had degree $d$. When the next $3$ iterations remove $Z_{2r}$ (see \Cref{fig:yescase}(a)), their degrees drop to $d-3$ this is because all nodes in $S$ were adjacent to all nodes in $Z_{2r}$ as $b(z_r)=1$. By the invariants in \Cref{lem:redn-inv}, all other remaining nodes have degree $\geq d-2$. Hence, the premise of \Cref{clm:zl-del} applies with $H$ as the current graph, $U=S$, and $\delta=d-3$. By \Cref{clm:zl-del}, the first three iterations on $H$, i.e., iterations $(6r+4,6r+5,6r+6)$ remove the nodes in $S$, each having degree $\leq d-3$ when deleted. Since $\sigma$ agrees with the peeling algorithm until this iteration, and we showed that every vertex $v$ until now has been removed while having degree $\leq d-3$, it follows that $\odeg_{\sigma}(v)\leq d-3$ for each such $v$. 

After $S$ is removed (see \Cref{fig:yescase}(b)), observe that all remaining nodes in layers $L_0,\ldots,L_{2r-1}$ have degree $\leq d-3$ since they had degree $\leq d$ before the removal of $S$, and they were all adjacent to all three nodes in $S$. Hence, regardless of their order in $\sigma$, $\odeg_{\sigma}(v)\leq d-3$ for any such node $v$. If we remove these nodes from the graph, then each node in layer $L_{2r}$ will have degree $\leq d-3$ in the remaining graph: it initially had degree $d$, and now has lost at least its $3$ neighbors in layer $L_{2r-1}$. Thus, $\odeg_{\sigma}(v)\leq d-3$ for each such node $v$. 

Finally, the auxilary nodes $\Vaux$ remain, and by \Cref{lem:redn-aux-constr}, the induced subgraph $G[\Vaux]$ has degree $\leq d-3$. Therefore, $\odeg_{\sigma}(v)\leq d-3$ for any $v\in \Vaux$ and we conclude that $\odeg_{\sigma}(v)\leq d-3$ for all $v\in V$. Hence, $\sigma$ is a $(d-3)$-ordering of vertices in $G$. Then by \Cref{fact:min-k-ord}, we get that $\kappa(G)\leq d-3$.
\end{proof}

\subsection{Completing the Reduction: Streaming to Communication Protocol}
We now put things together and complete the reduction to prove \Cref{thm:degen-lb}. We restate the theorem here. 

\mainthmdegen*

\begin{proof}
Suppose there is a $p$-pass $S$-space streaming algorithm $\cA$ that, given an integer $k$, outputs YES if $\kappa(G)\leq k$ and NO if not. Then, given an instance of $\bmhpc_{m,2p-1}$, the players $P_A, P_B, P_C,$ and $P_D$ run the following $(2p-1)$-phase protocol. They collectively feed the graph $G=(V,E)$ as in \Cref{lem:redn-aux-constr} to $\cA$ as follows. 

First note that the vertex set $V$ is known to all the players. Further, the edge sets $E_1,E_2,E_S$ are also input-independent and known to all players. In phase $1$ of the protocol, $P_C$ starts pass $1$ of the algorithm $\cA$ and inserts the edges in these sets. She then adds the edges $E_C$ which are completely determined by her input. Next, she sends the current state of $\cA$ along with all the vertex-degrees in the current graph to $P_D$. The player $P_D$ continues the stream by adding the edges in $E_D$, which is determined by his input. Next, he sends the current state of $\cA$ along with the updated vertex-degrees of the current graph (which he can calculate from the edges he added and the information sent to him by $P_C$) to the player $P_A$. This ends phase $1$.

In phase $2$, $P_A$ further extends the stream by adding the edges in $E_A$, which is given by her input. She sends the algorithm's state along with the updated node-degrees to $P_B$. The player $P_B$ adds the edge set $E_B$ and updates the degrees. Observe that the current graph has the edge set $E'$ as defined above. From the degree information, $P_B$ can now determine the edges $\Eaux$ to be added to set the target node-degrees (by \Cref{lem:redn-aux-constr}). Hence, $P_B$ adds those edges, which completes the first pass of $\cA$ on $G$. To end phase $2$, $P_B$ sends the current state of the algorithm to $P_C$. In phase $3$, $P_C$ starts pass $2$ of the algorithm by inserting the same edges that she added in phase $1$. They repeat the protocol as before for pass $2$, except that they no longer communicate degree information since $P_B$ already knows $\Eaux$. Thus, the players can simulate one pass of $\cA$ using two phases of the communication protocol, except the last pass that requires just one phase. After the final pass of $\cA$, the player $P_B$ looks at the output given by the algorithm for $k = d-3 = 6m(2p-1)+3m-3$ and determines the answer to $\bmhpc_{m,2p-1}$ (using \Cref{lem:redn-degen-yesno}). 

Thus, the players can solve $\bmhpc_{m,2p-1}$ using $2p-1$ phases and $O(pS+n\log n)$ bits of communication. By \Cref{thm:lb-mhpc}, it follows that $pS\geq \Omega(m^2- n\log n)$. Since $n=\Theta(mp)$ by \Cref{lem:redn-aux-constr}, we get that $pS\geq \Omega(n^2/p^2 - n\log n)$, which implies that $S\geq \Omega(n^2/p^3 - (n\log n)/p)$. This completes the proof of the theorem.  \end{proof}

We observe that the reductions for lexicographically-first maximal independent set (LFMIS) and s-t min-cut from \hpc in \cite{AssadiCK19} have a similar layered graph, where the \emph{same} \hpc instance is repeated between alternate pairs of consecutive layers. Instead, we can reduce from \mhpc as follows: similar to our reduction for degeneracy, we can make the successive pairs of layers encode the \emph{distinct} ``instances'' of \hpc that are given in a single instance of \mhpc. Then, following very similar arguments as above, the reductions from \mhpc to these problems go through and their streaming lower bounds are improved. Thus, we can prove \Cref{cor:lfmis-mincut}. We skip the formal details.

%% file: upper.tex

\section{Communication Upper Bounds}\label{sec:communication-upper}

In this section, we will see a communication upper bound for computing
the degeneracy of a graph exactly.
The setting is as follows:
Alice and Bob are respectively given \emph{disjoint} edge sets $E_A$ and $E_B$ on the set of
vertices $V$ (known to both Alice and Bob), and wish to compute the degeneracy of the graph
$G = (V, E_A \cup E_B)$.
The main theorem of this section is:

\begin{theorem}\label{thm:communication-upper}
  There exists a deterministic communication protocol that, given an $n$-vertex
  graph $G$, computes the degeneracy $\kappa(G)$ using $O(n \log^3 n)$ bits of
  communication.
  Moreover, the protocol outputs a $\kappa(G)$-ordering of the vertices
  of $G$, and a $\kappa(G)$-core in $G$.
\end{theorem}

We will focus on a decision version of the degeneracy problem, which we define
in the lemma below:

\begin{lemma}\label{lem:degen-decision}
  There exists a deterministic communication protocol that given an $n$-vertex
  graph $G$ and an integer $k$, decides whether the degeneracy
  $\kappa(G) \leq k$ using $O(n \log^2 n)$ bits of communication.
  Moreover, if the protocol accepts, it outputs a $k$-ordering of the vertices
  of $G$, and if it rejects, it outputs a $(k + 1)$-core in $G$.
\end{lemma}

For brevity, we will call the decision version of degeneracy defined in the
lemma above $\degen_{n, k}$.
Note that once we have a protocol for the decision problem, a protocol for the
search variant follows via binary search.
\begin{proof}[Proof of \Cref{thm:communication-upper}]
  Alice and Bob can binary search for the minimum $k$ such that $G$ has
  degeneracy at most $k$ (which is exactly $\kappa(G)$) using the protocol
  given by \Cref{lem:degen-decision}.
  Note that this incurs a multiplicative $\log n$ overhead in communication
  cost.
  Further, the protocol in the lemma also gives a $\kappa(G)$-ordering, since
  it accepts the input $(G, \kappa(G))$.
  Finally, to obtain a $\kappa(G)$-core, Alice and Bob can run the protocol
  in the lemma on the input $(G, \kappa(G) - 1)$.
\end{proof}

In the rest of this section, we will prove \Cref{lem:degen-decision}.
First, we show an easier (but more expensive) protocol that captures the main
idea.

\subsection{Warm-Up: An \texorpdfstring{$\Ot{n\sqrt{n}}$}{Õ(n sqrt n)} Communication Protocol}\label{sec:nrootn}
Recall the classical peeling algorithm, where one successively peels off the
vertex of minimum degree to find a degeneracy ordering of $G$.
We state a decision version of it below, which removes an arbitrary vertex of
degree $\leq k$ in each step.

\begin{Algorithm}\label{alg:peeling}
  The Peeling Algorithm. \medskip

  \textbf{Input:} A graph $G$, and an integer $k$.\\
  \textbf{Output:} \true{} if $G$ has degeneracy $\leq k$, \fail{} otherwise.

  \begin{enumerate}
    \item While $G$ is non-empty:
      \begin{enumerate}
        \item If there is a vertex $v \in G$ of degree $\leq k$,
          remove $v$ from $G$.
        \item Else output \fail.
      \end{enumerate}
    \item Output \true.
  \end{enumerate}
\end{Algorithm}

Note that in the \true{} case the $k$-ordering is simply the order the vertices
were peeled in, and in the \fail{} case, the $(k + 1)$-core is the set of vertices
remaining in $G$ when there are no vertices of degree $\leq k$.

We want to somehow port this algorithm to the communication setting.
The immediate problem is that Alice or Bob cannot compute a low degree
vertex without any communication after a vertex is removed.
The naive way to get around this is to send the updated degree of each
vertex after each deletion, but this can cost $\Omega(n^2)$ communication.
Instead, we observe that we only need to update the degrees of a remaining
vertex if it is not too large:

\begin{observation}\label{obs:sqrt-local}
  If $\deg(u) > k + \sqrt{n}$ at some step of the peeling algorithm,
  then $u$ cannot be deleted in the next $\sqrt{n}$ steps, since its degree
  falls by at most $1$ per step.
\end{observation}

Hence we have the following algorithm:
\begin{enumerate}
  \item Compute the set $S = \Set{u \in V \given \deg(u) > k + \sqrt{n}}$ of
    \emph{safe} vertices, and set it aside.
  \item Perform $\sqrt{n}$ iterations of the peeling algorithm, while updating
    the degree of a neighbor $u$ of a deleted vertex only if $u \notin S$ and $\deg(u) > k$.
  \item If the graph is non-empty, go to step 1.
\end{enumerate}

Let $N_A(u)$ and $\deg_A(u)$ denote the neighborhood and degree of $u$ within
$(V, E_A)$ (and similarly $N_B(u)$ and $\deg_B(u)$ for $E_B$).
We provide an implementation of the algorithm above in a communication protocol;
the sets $D$ and $L$ hold vertices that are ready to \emph{delete}, and
\emph{low} degree respectively.

\begin{Algorithm}\label{alg:nrootn}
  The $\Ot{n\sqrt{n}}$ communication protocol for $\degen_{n, k}$. \medskip

  \textbf{Input:} Alice and Bob get disjoint edge sets $E_A$ and $E_B$
    respectively, and both players get an integer $k$.\\
  \textbf{Output:} \true{} if $G = E_A \cup E_B$ has degeneracy $\leq k$, \fail{} otherwise.
  \begin{enumerate}
    \item Alice and Bob send $\deg_A(u)$ and $\deg_B(u)$ respectively for every
      vertex $u \in V(G)$, so they can then compute:
      \begin{itemize}
        \item $D \gets \Set*{ u \in G \given \deg(u) \leq k }$.
        \item $L \gets \Set*{ u \in G \given k + 1 \leq \deg(u) \leq k + \sqrt{n} }$.
        \item $S \gets V(G) \setminus (D \cup L)$.
      \end{itemize}
    \item For $\sqrt{n}$ iterations:
      \begin{enumerate}
        \item If $D$ is empty, output \fail.
        \item Let $v$ be any vertex in $D$. To remove $v$ from $G$:
        \item\label{step:Lnbrs} Alice and Bob send $N_A(v) \cap L$ and $N_B(v) \cap L$ respectively.
        \item\label{step:Lupdt} Hence they can update $\deg(u)$ for all $u \in L$.
        \item\label{step:Dmove} If $\deg(u) < k + 1$ for any $u \in L$,
          Alice and Bob move it to $D$.
      \end{enumerate}

    \item If $G$ is non-empty, go to Step 1.
    \item Output \true.
  \end{enumerate}
\end{Algorithm}
Note that the proofs (a $k$-ordering of the vertices in the \true{} case,
and a $(k + 1)$-core in the \fail{} case) fall out of the protocol exactly as before.
The correctness of this protocol follows from \Cref{obs:sqrt-local}; we omit
the proof here since it is subsumed by the next section.
To see that it only uses $\Ot{n \sqrt{n}}$ bits of communication, we observe
that:
\begin{itemize}
  \item Sending the degree of all vertices costs $O( n \log n )$ bits, and
    we do this at most $\sqrt{n}$ times.
  \item Once a vertex is in $D$, no communication is spent for updating its own degree.
  \item Each vertex in $L$ features at most $\sqrt{n}$ times in the neighborhood
    of a deleted vertex before it is moved to $D$, costing $O( \sqrt{n} \log n )$
    bits per vertex.
\end{itemize}
Hence the total communication cost of the protocol is $O( n\sqrt{n} \log n)$.

\subsection{An \texorpdfstring{$\Ot{n}$}{Õ(n)} Communication Protocol}

There are two changes from the previous protocol:
\begin{enumerate}
  \item 
    Instead of a single safe set of vertices, we have $O( \log n )$ safe sets
    corresponding to exponentially increasing gaps between the degree and $k$.
    More precisely, we define the set $S_i$ for $i \in \brac{\log n}$ such that
    whenever Alice and Bob compute $\deg(u)$ for some vertex $u$, they place it
    in $S_i$ if $k + 2^{i - 1} \leq \deg(u) < k + 2^i$.
  \item
    Instead of counting a global number of iterations until the safe vertices
    are updated, we track each vertex individually.
    Concretely, suppose $\deg(u) \geq k + \ell$.
    Then $u$ cannot be deleted until $\deg_A(u)$ or $\deg_B(u)$ falls by
    at least $\ell / 2$.
\end{enumerate}
The upshot of the two changes above is that for $u \in S_i$ Alice and Bob do not
have to update $\deg(u)$ until $\deg_A(u)$ or $\deg_B(u)$ falls by at least
$2^{i - 2}$.
We can now state the algorithm.
As before, the set $D$ holds all vertices which are ready for deletion.

\begin{Algorithm}\label{alg:n}
  The $\Ot{n}$ communication protocol for $\degen_{n, k}$. \medskip

  \textbf{Input:} Alice and Bob get disjoint edge sets $E_A$ and $E_B$
  respectively, and both players get an integer $k$.\\
  \textbf{Output:} \true{} if $G = E_A \cup E_B$ has degeneracy $\leq k$, \fail{} otherwise.
  \begin{enumerate}
    \item Alice and Bob send $\deg_A(u)$ and $\deg_B(u)$ respectively for every
      vertex $u \in V(G)$, so they can then compute:
      \begin{itemize}
        \item $D \gets \Set*{ u \in G \given \deg(u) \leq k }$.
        \item $S_i \gets \Set*{ u \in G \given k + 2^{i - 1} \leq \deg(u) < k + 2^i }$,
          for each $i \in \brac{\log n}$.
      \end{itemize}
    \item While {$G$ is non-empty:}
      \begin{enumerate}

        \item If $D$ is empty, output \fail.
        \item Let $v$ be any vertex in $D$. To remove $v$ from $G$:
        \item For each $i \in \brac{ \log n }$, and each $u \in S_i$,
          if $\deg_A(u)$ or $\deg_B(u)$ has fallen by $\geq 2^{i - 2}$ since it
          was last communicated, Alice sends $\deg_A(u)$ and Bob sends $\deg_B(u)$,
          hence allowing them to update $\deg(u)$.
        \item For all vertices updated in the previous step, Alice and Bob place
          them in the correct $S_i$ (or $D$).
      \end{enumerate}
    \item Output \true.
  \end{enumerate}
\end{Algorithm}

To wrap up, we prove that the algorithm above is correct, and uses $\Ot{n}$
space.

\begin{proof}[Proof of \Cref{lem:degen-decision}]
To prove that this protocol is correct, we argue that:
\begin{itemize}
  \item All vertices have degree $\leq k$ when they are removed, which
    implies that if the protocol outputs \true{}, $\kappa(G) \leq k$.
    This is immediate because we only move vertices to $D$ when we know their
    current degree exactly, and it is $\leq k$.
  \item Any vertex $u$ with degree $\leq k$ is placed in $D$, which implies
    that if $\kappa(G) \leq k$, then the protocol outputs \true{}.
    Suppose towards a contradiction that $u$ was placed in the set $S_i$ when
    its degree was updated for the final time.
    This means that $\deg(u)$ was at least $k + 2^{i - 1}$.
    But now for $\deg(u)$ to fall below $k + 1$, at least one of $\deg_A(u)$
    or $\deg_B(u)$ must fall by $\geq 2^{i - 2}$ (otherwise $\deg(u)$ falls by
    less than $2^{i - 1}$).
    Hence $\deg(u)$ was updated again following the final update, and we have
    a contradiction.
\end{itemize}

Next, we argue that this protocol uses only $\Ot{n}$ communication.
\begin{itemize}
  \item As before, the initial setup uses only $O(\log n)$ bits per vertex.
  \item While in $S_i$, a vertex can have its degree updated at most twice
    before it falls below $k + 2^{i - 1}$, and is hence dropped from
    $S_i$.
    Since there are only $\log n$ sets $S_i$, and a vertex can never gain
    degree, we only update its degree $O(\log n)$ times before it falls into $D$,
    spending $O( \log^2 n )$ bits overall.
\end{itemize}
Hence, in total, the protocol uses only $O(n \log^2 n)$ bits of communication.
\end{proof}

%% file: appendix.tex
\appendix

\section{Sampling Bounds}\label{sec:sampling}

Recall the Chernoff bound:

\begin{proposition}\label{lem:chernoff}
    Let $X_1, \cdots, X_n$ be independent random variables with values in $[0,1]$ and $X = \sum_i X_i$. Then: \begin{itemize}
        \item \textbf{Additive form: } For any $b \ge 0$, $\Pr[|X - \E[X]| \ge b] \le 2\cdot \exp\left(-\frac{b^2}{2n}\right)$, \label{itm:chernoff-add}

        \item \textbf{Multiplicative form: } For $0 \le \delta \le 1$, $\Pr[|X - \E[X]| \ge (1+\delta) E[X]] \le \exp\left(-\frac{\delta^2 E[X]}{3}\right)$.\label{itm:chernoff-mult}
    \end{itemize}
\end{proposition}

\bigskip\noindent
Now suppose we are allowed to repeatedly sample from an unknown Bernoulli distribution $\mu$. We know that $\mu$ either has acceptance probability $p$ or acceptance probability $q$. We wish to determine which is the case, via some method that may err with probability $\gamma$. Suppose that, for some $1 \ge \alpha,\beta > 0$, $p \le \alpha$, whereas $q \ge (1 + \beta) \alpha$. We may then draw 
\[
k = k(\alpha, \beta, \gamma) = \left\lceil\frac{48}{\alpha \beta^2} \ln\left(\frac{1}{\gamma}\right)\right\rceil
\]
samples $X_1, \ldots, X_k$ from $\mu$, and choose the threshold
\[
\tau = \tau(\alpha, \beta, \gamma) = \alpha \cdot \left(1 + \frac{\beta}{4}\right) \cdot k 
\]
It now follows, from the multiplicative form of the Chernoff bound (\Cref{lem:chernoff}), that
\[
\Pr[\sum_i X_i \ge \tau] = \Pr[\sum_i X_i \ge \alpha k \;\left(1 + \frac\beta4\right)] \le \exp\left(- \frac{\beta^2}{48} \alpha k \right) \le \gamma.
\]
Now, whenever $\beta \in [0,1]$, we have $1 + \frac\beta4 \le (1 + \beta) ( 1 - \frac{\beta}{4})$, and so, if the acceptance probability is $q \ge \alpha(1 + \beta)$, then
\begin{align*}
\Pr[\sum_i X_i \le \tau] %
& = \Pr[\sum_i X_i \le \alpha k\left(1 + \frac\beta4\right)] \le \Pr[\sum_i X_i \le \alpha (1 + \beta) k\left(1 - \frac\beta4\right)] \\
& \le \exp\left(- \frac{\beta^2}{48} \alpha (1+\beta) k \right) \le \gamma.
\end{align*}

\section{Proof of Propositions Regarding Triangular Discrimination} \label{sec:td-props}

\lambdatvd*

\begin{proof}
To see the lower bound on $\Lambda(\mu, \nu)$:
\[
\frac{\|\mu - \nu\|_1} 2 = \sum_{x: \mu(x) > \nu(x)} (\mu(x) - \nu(x)) =  \sum_{x: \mu(x) > \nu(x)} \frac{(\mu(x) - \nu(x))^2}{\sqrt{\mu(x) + \nu(x)}} \sqrt{\mu(x) + \nu(x)} \underset{(*)}{\leq} \sqrt{2\Lambda(\mu, \nu)}.
\]
Here $(*)$ follows from Cauchy-Schwarz. To see the upper bound:
\[
\Lambda(\mu, \nu) = \sum_{x: \mu(x) > \nu(x)}\frac{(\mu(x) - \nu(x))^2}{\mu(x) + \nu(x)} = \sum_{x: \mu(x) > \nu(x)}\frac{(\mu(x) - \nu(x))^2}{\mu(x) - \nu(x)} \cdot \frac{\mu(x) - \nu(x)}{\mu(x) + \nu(x)} \underset{(**)}{\leq} \frac{\|\mu - \nu\|_1} 2.
\]
Here $(**)$ follows from $\frac{\mu(x) - \nu(x)}{\mu(x) + \nu(x)}\leq 1$.

\end{proof}

We will make use of the following:

\begin{proposition}[\cite{Yehudayoff20}]
\label{fact:trick1}
If $|\eta| \leq \sqrt{a(b+ \eta)}$ with $a,b\geq 0$, then $\eta \leq a + 2b$.
\end{proposition}

\begin{proof}
    If $a = 0$, then $\eta = 0$ and the upper-bound follows. Otherwise, if $\eta^2 - a \eta - a b \le 0$, then
    \[
    \eta \le \frac{a + \sqrt{a^2 + 4 a b}}{2} = \frac{a}{2} (1 + \sqrt{1 + 4 b/a}) \le \frac{a}{2}(1 + 1 + 4 b/a). \qedhere
    \]
\end{proof}

\lambdaloss*

\begin{proof}
We have:
\begin{align*}
    |\eta|=\left|\E_\mu(f(x)) - \E_\nu(f(x))\right| &= \left|\sum_x(\mu(x) - \nu(x))f(x)\right|\\
    &\leq \sum_x\left|(\mu(x) - \nu(x))\right|f(x)\tag{By triangle inequality}\\
    &= \sum_x\frac{|\mu(x) - \nu(x)|\sqrt{f(x)}}{\sqrt{\mu(x) + \nu(x)}} \sqrt{(\mu(x) + \nu(x)) f(x)}\\
    &\leq \sqrt{\sum_x\frac{(\mu(x) - \nu(x))^2f(x)}{\mu(x) + \nu(x)}} \sqrt{\sum_x(\mu(x) + \nu(x)) f(x)}\tag{By Cauchy-Schwarz}
\end{align*}

\noindent
Let us look at the first term without the square root:
\begin{align*}
    \sum_x\frac{(\mu(x) - \nu(x))^2f(x)}{\mu(x) + \nu(x)} &= \sum_{x: \mu(x) \geq \nu(x)}\frac{(\mu(x) - \nu(x))^2f(x)}{\mu(x) + \nu(x)} + \sum_{x: \mu(x) < \nu(x)}\frac{(\mu(x) - \nu(x))^2f(x)}{\mu(x) + \nu(x)}\\
    &\leq \Lambda(\mu,\nu)\cdot \max_x f(x) + \sum_{x: \mu(x) < \nu(x)}\frac{(\mu(x) - \nu(x))^2f(x)}{\mu(x) + \nu(x)}\\
    &\leq \Lambda(\mu,\nu)\cdot \max_x f(x) + \sum_{x}\frac{(\nu(x))^2}{\nu(x)}f(x) = \Lambda(\mu,\nu)\cdot \max_x f(x) + \E_\nu(f(x)).
\end{align*}
which implies:
\begin{align*}
    |\eta|=\left|\E_\mu(f(x)) - \E_\nu(f(x))\right| \leq \sqrt{\underbrace{(\Lambda(\mu,\nu)\cdot \max_x f(x) + \E_\nu(f(x)))}_{a} (\eta+ \underbrace{2\E_\nu(f(x)))}_b)}
\end{align*}
At this point, we can use \Cref{fact:trick1} to get the following:
\begin{align*}
    \E_\mu(f(x)) - \E_\nu(f(x)) \leq \Lambda(\mu,\nu)\cdot \max_x f(x) + 5\Pr_\nu(\cE),\\
    \text{i.e., }\E_\mu(f(x))\leq \Lambda(\mu,\nu)\cdot \max_x f(x) + 6\E_\nu(f(x)),
\end{align*}
which completes the proof.
\end{proof}

\convex*

\begin{proof}
To see this, it suffices to compute the Hessian of each term.
We have
\[
\Lambda(\mu,\nu) = \sum_{x} \begin{cases}
\frac{(\mu(x)-\nu(x))^2}{\mu(x)+\nu(x)} & \text{if } \mu(x) \ge \nu(x)\\
0 & \text{otherwise}.
\end{cases}
\]
So let us prove that each term is jointly convex in the values of $(\mu(x), \nu(x))$. We restrict our attention to the unit square. The function is zero, hence convex, inside the region where $\mu(x) \le \nu(x)$. Furthermore, it is (obviously) non-negative inside the region where $\mu(x) > \nu(x)$. It is also convex in this region, since the Hessian is positive semidefinite there:
 \[
 \nabla^2 \frac{(\mu(x) - \nu(x))^2}{\mu(x)+\nu(x)}  =
\begin{pmatrix}
 \frac{8 \nu(x)^2}{(\mu(x) +\nu(x) )^3} & -\frac{8 \mu(x)  \nu(x) }{(\mu(x) +\nu(x) )^3} \\
 -\frac{8 \mu(x)  \nu(x) }{(\mu(x) +\nu(x) )^3} & \frac{8 \mu(x)^2}{(\mu(x) +\nu(x) )^3}
\end{pmatrix}.
\]
(The matrix is PSD since the off-diagonal entry, in absolute value, does not exceed the geometric mean of the diagonal entries.) It follows that each term is convex, and hence $\Lambda(\mu,\nu)$ is convex, in $(\mu,\nu)$.
\end{proof}

\section{Proof of \texorpdfstring{\Cref{clm:nonmdirect-sum}}{Claim}} \label{sec:proof-direct-sum}

\input{directsum}

\section{Proof of \texorpdfstring{\Cref{lem:redn-aux-constr}}{Claim}}\label{app:proof-auxnode-redn}
    Consider the graph $G'=(V',E')$. Set $d=6mr+3m$. We add $d$ auxiliary nodes in $\Vaux$. Since the target degree of each node outside $S$ is at most $d$, and the nodes in $S$ need at most $6r+3m\leq d$ edges to reach target degree, the $d$ auxiliary nodes suffice to increase their degree. Thus, the total number of vertices in $G$ is $|V'|+d=\Theta(mr)$. 
    
    We now describe the construction. For each node $v\in V'$, define its deficiency $\defc(v)$ as the difference between its target degree mentioned above and its current degree in $G'$. Define the total deficiency $D:= \sum_{v\in V'} \defc(v)$.  Fix arbitrary orderings $v_1,v_2,\ldots v_{|V'|}$ of the nodes in $V'$ and $u_1,\ldots,u_d$ of the nodes in $\Vaux$. From each $v_i$, we join edges to $\defc(v_i)$ many $u_j$'s following the ordering. The first neighbor of $v_{i+1}$ is taken to be the node that is next to the last neighbor of $v_i$ in the ordering of the $u_j$'s. We wrap around to $u_1$ after $u_d$. Hence, the degree of two vertices in $\Vaux$ can differ by at most $1$. We check the min-degree of a node in $\Vaux$. If it's $\geq d+6r+3$, then we are done. Otherwise, suppose it's $d+6r+3-x$ for some $x>0$. Then we add $x$ disjoint matchings of size $d/2$ in $\Vaux$ to increase the degree of each node by $x$, i.e., to ensure that each node in $\Vaux$ has degree $\geq d+6r+3$.  

    It remains to prove the ``furthermore'' part. For this, we check the maximum and minimum deficiency of each node.

    Observe that in $G'$, each node in layer $0$ has degree at least $5$ (it's adjacent to the other $2$ nodes in its triple and to the $3$ nodes in $S$) and at most $2m+5$ (for each $i\in [m]$, it is adjacent to at most $2$ nodes that represent $i$ in layer $1$). Hence, the deficiency of any such node is in $[d-2m-5, d-5]$. 
    
    For any node in a layer $\ell\in [2r-1]$, it is adjacent to at least $8$ nodes (the other $2$ nodes in its triple, the $3$ nodes in $S$, and the $3$ nodes representing the same element in its ``paired layer'', i.e., layer $\ell+1$ if $\ell$ is odd, or layer $\ell-1$ if it's even) and at most $2m+8$ nodes (for each $i\in [m]$, it is adjacent to at most $2$ nodes that represent $i$ in layer $\ell-1$ if $\ell$ is odd, or layer $\ell+1$ if it's even). Thus, the deficiency of any such node is in $[d-2m-8, d-8]$.

    Any node in layer $2r$ representing $x_i$ with $b(x_i)=0$ has degree exactly $5$ (the $2$ nodes in its triple and the $3$ nodes representing the same element in layer $2r-1$). Finally, any other nodes in layer $2r$ has degree exactly $8$ (same as other nodes in layer $2r$ with the addition of the $3$ nodes in $S$). Hence, the deficiency of any such node is in $[d-8, d-5]$. 

    Therefore, the total deficiency $D$ is at least $(d-2m-8)\cdot |V_L|= (d-2m-8)\cdot d \geq 9d^2/10$ for large enough $m$. These edges are received almost equally by the $d$ nodes in $|\Vaux|$. Thus, the min-degree in $|\Vaux|$ after the addition of these edges is at least $9d/10$. Thus the value of $x$ is at most $d+6r+3-9d/10\leq d-3$ (again for large enough $m$). The $x$ matchings can induce degree at most $x\leq d-3$ in the subgraph $G[\Vaux]$, and hence we are done with the furthermore part.

\input{info}

%% file: directsum.tex
\begin{proof}
This proof is a verbatim reproduction of that in \cite{AssadiCK19}. For any $i \in [m]$, define $\bA^{<i} := \{A_{x_1},\ldots,A_{x_{i-1}}\}$, $\bB^{>i} := \{B_{x_{i+1}},\ldots,B_{x_{m}}\}$. Recall that 
the internal information cost of $\prot_\SI$ is $\IC{\prot_\SI}{\cD_\SI} := \mi{\rA}{\rProtSI \mid \rB} + \mi{\rB}{\rProtSI \mid \rA}$. In the following, we focus on bounding the first term. The second term can be bounded exactly
the same by symmetry.  

As $(\rI,\bA^{<\rI},\bB^{>\rI},\bC,\bD)$ is sampled via public randomness in $\prot_\SI$, by Proposition~\ref{prop:public-random},   
\begin{align*}
	\mi{\rA}{\rProtSI \mid \rB} &= \mi{\rA}{\rProtSI \mid \rB, \rI,\bA^{<\rI},\bB^{>\rI},\bC,\bD} \leq \mi{\rA}{\rProtSI \mid \rB, \rI,\bA^{<\rI},\bB^{>\rI}}. 
\end{align*}
The inequality is by Proposition~\ref{prop:info-decrease} as we now show $\rA \perp (\bC,\bD) \mid \rProtSI,\rB,\rI,\bA^{<\rI},\bB^{>\rI}$ (and hence conditioning on $(\bC,\bD)$ can only decrease the mutual information). 
This is because $\rA \perp (\bC,\bD) \mid \rB,\rI,\bA^{<\rI},\bB^{>\rI}$ by Observation~\ref{obs:distHPC} and $\rProtSI$ is transcript of a deterministic protocol plus $z_1,\ldots,z_j$ obtained deterministically and hence we can apply
Proposition~\ref{prop:hpc-rectangle}.  

Define a random variable $\rTheta \in \{0,1\}$ where $\rTheta = 1$ iff in Line~(\ref{line:AB}) of protocol $\prot_\SI$, we terminate the protocol. In other words $\rTheta = 1$ iff $x_i \in Z^{<j}$. Since $\rA \perp \rTheta \mid \rB, \rI,\bA^{<\rI},\bB^{>\rI}$, 
further conditioning on $\rTheta$ can only increase the mutual information term above by Proposition~\ref{prop:info-increase}, hence, 
\begin{align}
	\mi{\rA}{\rProtSI \mid \rB} &\leq \mi{\rA}{\rProtSI \mid \rB, \rI,\bA^{<\rI},\bB^{>\rI},\rTheta} \notag \\
	&= \frac{m-j}{m} \cdot \mi{\rA}{\rProtSI \mid \rB, \rI,\bA^{<\rI},\bB^{>\rI},\rTheta=0} + \frac{j}{m} \cdot \mi{\rA}{\rProtSI \mid \rB, \rI,\bA^{<\rI},\bB^{>\rI},\rTheta=1} \notag \\
	&\leq \frac{m-j}{m} \cdot \mi{\rA}{\rProtSI \mid \rB, \rI,\bA^{<\rI},\bB^{>\rI},\rTheta=0}, \label{eq:mi-ra}
\end{align}
since conditioned on $\rTheta=1$, the protocol $\rProtSI$ is simple some prefix of $\rZ^{<j}$ and is hence independent of the input $(\rA,\rB)$ and carries no information about $\rA$ (see \itfacts{info-zero}). 
We now further bound the RHS of Eq~(\ref{eq:mi-ra}). When $\rTheta = 0$, $\rProtSI = (\rZ^{<j} ,\rProt_1,\ldots,\rProt_j) = (\rE^{<j},\rProt_j)$. Hence, we can write, 
  \begin{align*}
	 \mi{\rA}{\rProtSI \mid \rB, \rI,\bA^{<\rI},\bB^{>\rI},\rTheta=0} &\leq \mi{\rA}{\rE_j,\rProt_j \mid \rB, \rI,\bA^{<\rI},\bB^{>\rI},\rTheta=0} \\
	&= \mi{\rA}{\rZ^{<j} \mid \rB, \rI,\bA^{<\rI},\bB^{>\rI},\rTheta=0} \\
	&\hspace{25pt} + \mi{\rA}{\rProt^{<j},\rProt_j \mid \rZ^{<j},\rB, \rI,\bA^{<\rI},\bB^{>\rI},\rTheta=0} \tag{by chain rule in \itfacts{chain-rule} and since $\rE_j = (\rProt^{<j},\rZ^{<j})$} \\
	&\leq \mi{\rA}{\rProt \mid \rZ^{<j},\rB, \rI,\bA^{<\rI},\bB^{>\rI},\rTheta=0},
\end{align*}
as $\rA \perp \rZ^{<j} \mid \rTheta=0$ (and other variables) and hence the first term is zero, and in the second term $\rProt$ contains $\rProt^{<j},\rProt_j$ (plus potentially other terms) and so having $\rProt$ in instead can only increase
the information. By further expanding the conditional information term above,
\begin{align*}
	 &\mi{\rA}{\rProtSI \mid \rB, \rI,\bA^{<\rI},\bB^{>\rI},\rTheta=0} \\
	 &\hspace{50pt}\leq \E_{(Z^{<j},i) \mid \rTheta=0} \bracket{\mi{\rA}{\rProt \mid \rB, \bA^{<i},\bB^{>i},\rI=i,\rZ^{<j}=Z^{<j},\rTheta=0}} \\
	&\hspace{50pt}= \E_{Z^{<j} \mid \rTheta=0} \bracket{\sum_{\substack{i=1 \\~i \notin Z^{<j}}}^{m} \frac{1}{m-j}  \mi{\rA_{x_i}}{\rProt \mid \rB_{x_i}, \bA^{<i},\bB^{>i},\rI=i,\rZ^{<j}=Z^{<j},\rTheta=0}} \tag{conditioned on $\rTheta=0$, $i$ is chosen
	uniformly at random from $Z^{<j}$; also $(\rA,\rB) = (\rA_{x_i},\rB_{x_i})$} \\
	&\hspace{50pt}= \E_{Z^{<j} \mid \rTheta=0} \bracket{\sum_{i \notin Z^{<j}} \frac{1}{m-j} \cdot \mi{\rA_{x_i}}{\rProt \mid \rB_{x_i}, \bA^{<i},\bB^{>i},\rZ^{<j}=Z^{<j},\rTheta=0}}
	\tag{we dropped the conditioning on $\rI=i$ as all remaining variables are independent of this event}  \\
	&\hspace{50pt}= \E_{Z^{<j} \mid \rTheta=0} \bracket{\sum_{i \notin Z^{<j}} \frac{1}{m-j} \cdot \mi{\rA_{x_i}}{\rProt \mid \bA^{<i},\bB,\rZ^{<j}=Z^{<j},\rTheta=0}}
	\tag{as $\rA_{x_i} \perp \bB^{<i} \mid \rB_{x_i},\bA^{<i}$ by Observation~\ref{obs:distHPC} and hence we can apply Proposition~\ref{prop:info-increase}} \\
	&\hspace{50pt}\leq \E_{Z^{<j} \mid \rTheta=0} \bracket{\sum_{i=1}^{m} \frac{1}{m-j} \cdot \mi{\rA_{x_i}}{\rProt \mid \bA^{<i},\bB,\rZ^{<j}=Z^{<j},\rTheta=0}} \\ 
	\tag{mutual information is non-negative by \itfacts{info-zero} and so we can add the terms in $Z^{<j}$  as well} \\
	&\hspace{50pt}= \E_{Z^{<j} \mid \rTheta=0} \bracket{\sum_{i=1}^{m} \frac{1}{m-j} \cdot \mi{\rA_{x_i}}{\rProt \mid \bA^{<i},\bB,\rZ^{<j}=Z^{<j},\rTheta=0}} \\
	&\hspace{50pt}= \frac{1}{m-j} \cdot \E_{Z^{<j} \mid \rTheta=0} \bracket{\mi{\bA}{\rProt \mid \bB,\rZ^{<j}=Z^{<j},\rTheta=0}} \tag{by chain rule in \itfacts{chain-rule}} \\
	&\hspace{50pt}= \frac{1}{m-j} \cdot \mi{\bA}{\rProt \mid \bB,\rZ^{<j},\rTheta=0} \\
	&\hspace{50pt}\leq \frac{1}{m-j} \cdot \paren{\mi{\bA}{\rProt \mid \bB,\rTheta=0} + \en{\rZ^{<j}}} \\
	&\hspace{50pt}= \frac{1}{m-j} \cdot \paren{\mi{\bA}{\rProt \mid \bB} + \en{\rZ^{<j}}} \tag{transcript of the protocol $\prot_\bhpc$ (namely $\rProt$) on input $(\bA,\bB)$ is independent of $\rTheta$} \\
	&\hspace{50pt}\leq \frac{1}{m-j} \cdot \paren{\en{\rProt} + \en{\rZ^{<j}}} \leq \frac{\CC{\prot_\bhpc}{} }{m-j}   + \frac{j \cdot \log{m}}{m-j}. \tag{by sub-additivity of entropy (\itfacts{sub-additivity}) and \itfacts{uniform}}
\end{align*}

By plugging in this bound in Eq~(\ref{eq:mi-ra}), we have that,
\begin{align*}
	\mi{\rA}{\rProtSI \mid \rB} &\leq \frac{m-j}{m} \cdot \paren{\frac{\CC{\prot_\bhpc}{} }{m-j}   + \frac{j \cdot \log{m}}{m-j}} = \frac{\CC{\prot_\bhpc}{} }{m}   + \frac{j \cdot \log{m}}{m}. 
\end{align*}
By symmetry, we can also prove the same bound on $\mi{\rB}{\rProtSI \mid \rA}$. As such, we have,
\begin{align*}
	\mi{\rA}{\rProtSI \mid \rB} + \mi{\rB}{\rProtSI \mid \rA} \leq 2 \cdot \paren{\frac{\CC{\prot_\bhpc}{} }{m}   + \frac{j \cdot \log{m}}{m}}.
\end{align*}

We shall note that strictly speaking the factor $2$ above is not needed (similar to the proof of Proposition~\ref{prop:cc-ic}) but as this factor is anyway suppressed through O-notation later in the proof, the above bound
suffices for our purpose. 

\end{proof}

%% file: info.tex
\section{Background on Information Theory}\label{prelim-sec:info}

We now briefly introduce some definitions and facts from information theory that are needed in this thesis. We refer the interested reader to the text by Cover and Thomas~\cite{CoverT06} for an excellent introduction to this field, 
and the proofs of the statements used in this Appendix. 

For a random variable $\rA$, we use $\supp{\rA}$ to denote the support of $\rA$ and $\distribution{\rA}$ to denote its distribution. 
When it is clear from the context, we may abuse the notation and use $\rA$ directly instead of $\distribution{\rA}$, for example, write 
$A \sim \rA$ to mean $A \sim \distribution{\rA}$, i.e., $A$ is sampled from the distribution of random variable $\rA$.

We denote the \emph{Shannon Entropy} of a random variable $\rA$ by
$\en{\rA}$, which is defined as: 
\begin{align}
	\en{\rA} := \sum_{A \in \supp{\rA}} \Pr\paren{\rA = A} \cdot \log{\paren{1/\Pr\paren{\rA = A}}} \label{eq:entropy}
\end{align} 

\noindent
The \emph{conditional entropy} of $\rA$ conditioned on $\rB$ is denoted by $\en{\rA \mid \rB}$ and defined as:
\begin{align}
\en{\rA \mid \rB} := \E_{B \sim \rB} \bracket{\en{\rA \mid \rB = B}}, \label{eq:cond-entropy}
\end{align}
where 
$\en{\rA \mid \rB = B}$ is defined in a standard way by using the distribution of $\rA$ conditioned on the event $\rB = B$ in Eq~(\ref{eq:entropy}). 

\noindent
Recall that the binary entropy $H(p)$ is the entropy of a two-valued random
variable that takes one value with probability $p$.
We need the following upper bound on the binary entropy.

\begin{fact}\label{fact:ent-ub}
  The binary entropy $H_2(p) \leq 2 \sqrt{ p \cdot (1 - p) }$.
\end{fact}

The \emph{mutual information} of two random variables $\rA$ and $\rB$ is denoted by
$\mi{\rA}{\rB}$ and is defined as:
\begin{align}
\mi{\rA}{\rB} := \en{A} - \en{A \mid  B} = \en{B} - \en{B \mid  A}. \label{eq:mi}
\end{align}
\noindent
The \emph{conditional mutual information} $\mi{\rA}{\rB \mid \rC}$ is $\en{\rA \mid \rC} - \en{\rA \mid \rB,\rC}$ and hence by linearity of expectation:
\begin{align}
	\mi{\rA}{\rB \mid \rC} = \E_{C \sim \rC} \bracket{\mi{\rA}{\rB \mid \rC = C}}. \label{eq:cond-mi}
\end{align}

\subsection{Useful Properties of Entropy and Mutual Information}\label{sec:prop-en-mi}

We shall use the following basic properties of entropy and mutual information throughout.

\begin{fact}\label{fact:it-facts}
  Let $\rA$, $\rB$, $\rC$, and $\rD$ be four (possibly correlated) random variables.
   \begin{enumerate}
  \item \label{part:uniform} $0 \leq \en{\rA} \leq \log|\supp{\rA}|$. The right equality holds
    iff $\distribution{\rA}$ is uniform.
  \item \label{part:info-zero} $\mi{\rA}{\rB \mid \rC} \geq 0$. The equality holds iff $\rA$ and
    $\rB$ are \emph{independent} conditioned on $\rC$.
    \item \label{part:info-atmost-rv} $\mi{\rA}{\rB \mid \rC} \leq \en{\rA}$.
  \item \label{part:cond-reduce} \emph{Conditioning on a random variable reduces entropy}:
    $\en{\rA \mid \rB,\rC} \leq \en{\rA \mid  \rB}$.  The equality holds iff $\rA \perp \rC \mid \rB$.
    \item \label{part:sub-additivity} \emph{Subadditivity of entropy}: $\en{\rA,\rB \mid \rC}
    \leq \en{\rA \mid C} + \en{\rB \mid  \rC}$.
   \item \label{part:ent-chain-rule} \emph{Chain rule for entropy}: $\en{\rA,\rB \mid \rC} = \en{\rA \mid \rC} + \en{\rB \mid \rC,\rA}$.
  \item \label{part:chain-rule} \emph{Chain rule for mutual information}: $\mi{\rA,\rB}{\rC \mid \rD} = \mi{\rA}{\rC \mid \rD} + \mi{\rB}{\rC \mid  \rA,\rD}$.
  \item \label{part:data-processing} \emph{Data processing inequality}: for a function $f(\rA)$ of $\rA$, $\mi{f(\rA)}{\rB \mid \rC} \leq \mi{\rA}{\rB \mid \rC}$. 
   \end{enumerate}
\end{fact}

\noindent
We also use the following two standard propositions, regarding the effect of conditioning on mutual information.

\begin{proposition}\label{prop:info-increase}
  For random variables $\rA, \rB, \rC, \rD$, if $\rA \perp \rD \mid \rC$, then, 
  \[\mi{\rA}{\rB \mid \rC} \leq \mi{\rA}{\rB \mid  \rC,  \rD}.\]
\end{proposition}
 \begin{proof}
  Since $\rA$ and $\rD$ are independent conditioned on $\rC$, by
  \itfacts{cond-reduce}, $\HH(\rA \mid  \rC) = \HH(\rA \mid \rC, \rD)$ and $\HH(\rA \mid  \rC, \rB) \ge \HH(\rA \mid  \rC, \rB, \rD)$.  We have,
	 \begin{align*}
	  \mi{\rA}{\rB \mid  \rC} &= \HH(\rA \mid \rC) - \HH(\rA \mid \rC, \rB) = \HH(\rA \mid  \rC, \rD) - \HH(\rA \mid \rC, \rB) \\
	  &\leq \HH(\rA \mid \rC, \rD) - \HH(\rA \mid \rC, \rB, \rD) = \mi{\rA}{\rB \mid \rC, \rD}. \qed
	\end{align*}
	
\end{proof}

\begin{proposition}\label{prop:info-decrease}
  For random variables $\rA, \rB, \rC,\rD$, if $ \rA \perp \rD \mid \rB,\rC$, then, 
  \[\mi{\rA}{\rB \mid \rC} \geq \mi{\rA}{\rB \mid \rC, \rD}.\]
\end{proposition}
 \begin{proof}
 Since $\rA \perp \rD \mid \rB,\rC$, by \itfacts{cond-reduce}, $\HH(\rA \mid \rB,\rC) = \HH(\rA \mid \rB,\rC,\rD)$. Moreover, since conditioning can only reduce the entropy (again by \itfacts{cond-reduce}), 
  \begin{align*}
 	\mi{\rA}{\rB \mid  \rC} &= \HH(\rA \mid \rC) - \HH(\rA \mid \rB,\rC) \geq \HH(\rA \mid \rD,\rC) - \HH(\rA \mid \rB,\rC) \\
	&= \HH(\rA \mid \rD,\rC) - \HH(\rA \mid \rB,\rC,\rD) = \mi{\rA}{\rB \mid \rC,\rD}. \qed
 \end{align*}

\end{proof}

\subsection{Background on Communication and Information Complexity}\label{sec:cc-ic}

\paragraph{Communication complexity.} We briefly review the standard definitions of the two-party communication model of Yao~\cite{Yao79}.  See the text by Kushilevitz and Nisan~\cite{KushilevitzN97} for 
an extensive overview of communication complexity. In Section~\ref{sec:lower-bound-hpc}, we also use a standard generalization of this model to allow for more than two players, but we defer
the necessary definitions to that section. 

Let $P: \mathcal{X} \times \mathcal{Y} \rightarrow \mathcal{Z}$ be a relation.  Alice receives an input $X
\in \mathcal{X}$ and Bob receives $Y \in \mathcal{Y}$, where $(X,Y)$ are chosen from a
joint distribution $\cD$ over $\mathcal{X} \times \mathcal{Y}$. We allow players to have access to both public and private randomness. 
They communicate with each other by exchanging messages such that each message
depends only on the private input and random bits of the player sending the message, and the already communicated messages plus the public randomness. 
At the end, one of the players need to output an answer  $Z$ such that $Z \in P(X,Y)$.  

We use $\prot$ to denote a protocol used by the players. We always assume that the protocol $\prot$ can be randomized (using both public and
private randomness), \emph{even against a prior distribution $\cD$ of inputs}. For any
$0 < \delta < 1$, we say $\prot$ is a $\delta$-error protocol for $P$ over a distribution
$\cD$, if the probability that for an input $(X,Y)$, $\prot$ outputs some $Z$ where $Z \notin P(X,Y)$ is at most
$\delta$ (the probability is taken over the randomness of \emph{both} the distribution and the protocol).

\begin{definition}[Communication cost]
  The \emph{communication cost} of a protocol $\prot$ on an input
  distribution $\cD$, denoted by $\CC{\prot}{\cD}$, is the \emph{worst-case} bit-length of the transcript
  communicated between Alice and Bob in the protocol $\prot$, when the inputs are chosen from $\cD$.

\end{definition}

Communication complexity of a problem $P$ is defined as the minimum communication cost of a protocol $\prot$ that solves $P$ on every distribution $\cD$ with probability at least $2/3$.

\paragraph{Information complexity.} There are several possible definitions of information
cost of a communication prtocol that have been considered depending on the application (see, e.g.,~\cite{ChakrabartiSWY01,Bar-YossefJKS02,BarakBCR10,BravermanR11,BravermanEOPV13}).  
We use the notion of \emph{internal information cost}~\cite{BarakBCR10} that measures the average amount of information each player learns about the input of the other player
by observing the transcript of the protocol. 
\begin{definition}[Information cost]
  Consider an input distribution $\cD$ and a protocol $\prot$. Let $(\rX,\rY) \sim \cD$ denote the random variables for the input of Alice and Bob and $\rProt$ be the 
  the random variable for the transcript of the protocol \emph{concatenated} with the public randomness $\rR$ used by $\prot$. 
  The \emph{(internal) information cost} of $\prot$ with respect to
  $\cD$ is $\ICost{\prot}{\cD}:=\mii{\rProt}{\rX \mid \rY}{\cD} + \mii{\rProt}{\rY \mid \rX}{\cD}$. 

\end{definition}

One can also define information complexity of a problem $P$ similar to communication complexity with respect to the information cost. However, we avoid presenting this definition formally due to some subtle technical issues
 that need to be addressed which lead to multiple different but similar-in-spirit definitions. As such, we state our results directly in terms of information cost. 

Note that any public coin protocol is a distribution over private coins protocols, run by
first using public randomness to sample a random string $\rR=R$ and then running the
corresponding private coin protocol $\prot^R$. We also use $\rProt^R$ to denote the transcript of the protocol $\prot^R$. 
We have the following standard proposition. 
\begin{proposition}\label{prop:public-random}
	 For any distribution $\cD$ and any protocol $\prot$ with public randomness $\bR$, $$\ICost{\prot}{\cD} = \mii{\rProt}{\rX \mid \rY, \rR}{\cD} + \mii{\rProt}{\rY \mid \rX, \rR}{\cD} = \E_{R \sim \rR}\bracket{\ICost{\prot^R}{\cD}}.$$ 
\end{proposition}
\begin{proof} By definition of internal information cost, 
\begin{align*}
	\ICost{\prot}{\cD} &= \mii{\rProt}{\rX \mid \rY}{\cD} + \mii{\rProt }{\rY \mid \rX}{\cD} = \mi{\rProt,\rR}{\rX \mid \rY} + \mi{\rProt,\rR}{\rY \mid \rX} \tag{$\Prot$ denotes the transcript and the public randomness} \\
	&= \mi{\rR}{\rX \mid \rY}  + \mi{\rProt}{\rX \mid \rY , \rR} + \mi{\rR}{\rY \mid \rX}  + \mi{\rProt}{\rY \mid \rX , \rR}  \tag{chain rule of mutual information, \itfacts{chain-rule}} \\
	&= \mi{\rProt}{\rX \mid \rY , \rR} + \mi{\rProt}{\rY \mid \rX , \rR} \tag{$\mi{\rR}{\rX \mid \rY}  = \mi{\rR}{\rY \mid \rX}  = 0$ since $\rR \perp \rX, \rY$ and \itfacts{info-zero}} \\
	&= \E_{R \sim \rR}\bracket{\mi{\rProt}{\rX \mid \rY , \rR=R} + \mi{\rProt}{\rY \mid \rX , \rR=R}} = \E_{R \sim \rR}\bracket{\ICost{\prot^R}{\cD}}, 
\end{align*}
concluding the proof.
\end{proof}

The following well-known proposition relates communication cost and information cost. 
\begin{proposition}[cf.~\cite{BravermanR11}]\label{prop:cc-ic}
  For any distribution $\cD$ and any protocol $\prot$: $\ICost{\prot}{\cD} \leq \CC{\prot}{\cD}$.
\end{proposition}
\begin{proof}
Let us assume first that $\prot$ only uses private randomness and thus $\rProt$ only contain the transcript.  
For any $b \in [\CC{\prot}{\cD}]$, we define $\Prot_b$ to be the $b$-th bit of the transcript. We have, 
\begin{align*}
	 \ICost{\prot}{\cD} &= \mi{\rProt}{\rX \mid \rY}+ \mi{\rProt }{\rY \mid \rX} \\
	 &= \sum_{b=1}^{\CC{\prot}{\cD}} \mi{\rProt_b}{\rX \mid \rProt^{<b},\rY} + \mi{\rProt_b}{\rY \mid \rProt^{<b},\rX} \tag{by chain rule of mutual information in \itfacts{chain-rule}}\\
	 &= \sum_{b=1}^{\CC{\prot}{\cD}} \E_{\Prot^{<b}}\bracket{\mi{\rProt_b}{\rX \mid \rProt^{<b} = \Prot^{<b},\rY} + \mi{\rProt_b}{\rY \mid \rProt^{<b}= \Prot^{<b},\rX}}.
\end{align*}
Consider each term in the RHS above. By conditioning on $\Prot^{<b}$, the player that transmit $\rProt_b$ would become fix. If this player is Alice, then $\mi{\rProt_b}{\rY \mid \rProt^{<b}= \Prot^{<b},\rX} = 0$, because
$\rProt_b$ is only a function of $(\rProt^{<b},\rX)$ in this case; similarly, if this player is Bob, then $\mi{\rProt_b}{\rX \mid \rProt^{<b}= \Prot^{<b},\rY} = 0$. Moreover, 
$\mi{\rProt_b}{\rX \mid \rProt^{<b} = \Prot^{<b},\rY} \leq \en{\rProt_b}  \leq 1$ and similarly $\mi{\rProt_b}{\rY \mid \rProt^{<b}= \Prot^{<b},\rX} \leq 1$. As such, the above term can be upper bounded 
by $\CC{\prot}{\cD}$. To finalize the proof, note that by Proposition~\ref{prop:public-random}, for any public-coin protocol $\prot$, $\ICost{\prot}{\cD} = \E_{R \sim \rR}\bracket{\ICost{\prot^R}{\cD}} \leq \E_{R \sim \rR}\bracket{\CC{\prot^R}{\cD}} \leq \CC{\prot}{\cD}$, where
the first inequality is by the first part of the argument.  
\end{proof}

Proposition~\ref{prop:cc-ic} provides a convinent way of proving communication complexity lower bounds by lower bounding information cost of any protocol.

\subsection*{Rectangle Property of Communication Protocols}\label{sec:statstics}
We conclude this section by mentioning some basic properties of communication protocols. For any protocol $\prot$ and inputs $x \in \mathcal{X}$ and $y \in \mathcal{Y}$, we define $\Prot_{x,y}$ as the transcript
of the protocol conditioned on the input $x$ to Alice and input $y$ to Bob. Note that for randomized protocols, $\Prot_{x,y}$ is a random variable which we denote by $\rProt_{x,y}$. 

The following is referred to as the rectangle property of deterministic protocols. 
\begin{fact}[Rectangle property]\label{fact:rectangle}
	For any deterministic protocol $\prot$ and inputs $x,x' \in \mathcal{X}$ to Alice and $y,y' \in \mathcal{Y}$ to Bob, if $\Prot_{x,y} = \Prot_{x',y'}$, then $\Prot_{x,y'} = \Prot_{x',y}$. 
\end{fact}
\noindent
Fact~\ref{fact:rectangle} implies that the set of inputs consistent with any transcript $\Prot_{x,y}$ of a deterministic protocol forms a combinatorial rectangle. 
One can also extend the rectangle property of deterministic protocols to randomized protocols using the following fact. 
\begin{fact}[Cut-and-paste property; cf.~\cite{Bar-YossefJKS02}]\label{fact:r-rectangle}
	For any randomized protocol $\prot$ and inputs $x,x' \in \mathcal{X}$ to Alice and $y,y' \in \mathcal{Y}$ to Bob, $\hd{\rProt_{x,y}}{\rProt_{x',y'}} = \hd{\rProt_{x,y'}}{\rProt_{x',y}}$. 
\end{fact}

%% file: main.bbl
\newcommand{\etalchar}[1]{$^{#1}$}
\begin{thebibliography}{BvdBE{\etalchar{+}}22}

\bibitem[A17]{Assadi17sc}
Sepehr~Assadi A.
\newblock Tight space-approximation tradeoff for the multi-pass streaming set cover problem.
\newblock In {\em Proceedings of the 36th {ACM} {SIGMOD-SIGACT-SIGAI} Symposium on Principles of Database Systems, {PODS} 2017, Chicago, IL, USA, May 14-19, 2017}, pages 321--335, 2017.

\bibitem[A22]{Assadi22}
Sepehr~Assadi A.
\newblock A two-pass (conditional) lower bound for semi-streaming maximum matching.
\newblock In Joseph~(Seffi) Naor and Niv Buchbinder, editors, {\em Proceedings of the 2022 {ACM-SIAM} Symposium on Discrete Algorithms, {SODA} 2022, Virtual Conference / Alexandria, VA, USA, January 9 - 12, 2022}, pages 708--742. {SIAM}, 2022.

\bibitem[AA20]{AlonA20}
Noga Alon and Sepehr Assadi.
\newblock Palette sparsification beyond ({\(\Delta\)}+1) vertex coloring.
\newblock In Jaroslaw Byrka and Raghu Meka, editors, {\em Approximation, Randomization, and Combinatorial Optimization. Algorithms and Techniques, {APPROX/RANDOM} 2020, August 17-19, 2020, Virtual Conference}, volume 176 of {\em LIPIcs}, pages 6:1--6:22. Schloss Dagstuhl - Leibniz-Zentrum f{\"{u}}r Informatik, 2020.

\bibitem[ACK19]{AssadiCK19}
Sepehr Assadi, Yu~Chen, and Sanjeev Khanna.
\newblock Polynomial pass lower bounds for graph streaming algorithms.
\newblock In {\em {STOC}}, pages 265--276. {ACM}, 2019.

\bibitem[ACKP21]{AbboudCKP21}
Amir Abboud, Keren Censor{-}Hillel, Seri Khoury, and Ami Paz.
\newblock Smaller cuts, higher lower bounds.
\newblock {\em {ACM} Trans. Algorithms}, 17(4):30:1--30:40, 2021.

\bibitem[AD21]{AssadiD21}
Sepehr Assadi and Aditi Dudeja.
\newblock A simple semi-streaming algorithm for global minimum cuts.
\newblock In {\em {SOSA}}, pages 172--180. {SIAM}, 2021.

\bibitem[ADBV05]{Alvarez-HamelinDBV05}
J.~Ignacio Alvarez{-}Hamelin, Luca Dall'Asta, Alain Barrat, and Alessandro Vespignani.
\newblock Large scale networks fingerprinting and visualization using the k-core decomposition.
\newblock In {\em Advances in Neural Information Processing Systems 18 [Neural Information Processing Systems, {NIPS} 2005, December 5-8, 2005, Vancouver, British Columbia, Canada]}, pages 41--50, 2005.

\bibitem[AJJ{\etalchar{+}}22]{AssadiJJST22}
Sepehr Assadi, Arun Jambulapati, Yujia Jin, Aaron Sidford, and Kevin Tian.
\newblock Semi-streaming bipartite matching in fewer passes and optimal space.
\newblock In {\em {SODA}}, pages 627--669. {SIAM}, 2022.

\bibitem[AKSY20]{AssadiKSY20}
Sepehr Assadi, Gillat Kol, Raghuvansh Saxena, and Huacheng Yu.
\newblock Multi-pass graph streaming lower bounds for cycle counting, max-cut, matching size, and other problems.
\newblock In {\em 61st Annual {IEEE} Symposium on Foundations of Computer Science, {FOCS} (to appear)}, 2020.

\bibitem[AKZ22]{AssadiKZ22}
Sepehr Assadi, Gillat Kol, and Zhijun Zhang.
\newblock Rounds vs communication tradeoffs for maximal independent sets.
\newblock In {\em 63rd {IEEE} Annual Symposium on Foundations of Computer Science, {FOCS} 2022, Denver, CO, USA, October 31 - November 3, 2022}, pages 1193--1204. {IEEE}, 2022.

\bibitem[AN21]{AssadiN21}
Sepehr Assadi and Vishvajeet N.
\newblock Graph streaming lower bounds for parameter estimation and property testing via a streaming {XOR} lemma.
\newblock In Samir Khuller and Virginia~Vassilevska Williams, editors, {\em {STOC} '21: 53rd Annual {ACM} {SIGACT} Symposium on Theory of Computing, Virtual Event, Italy, June 21-25, 2021}, pages 612--625. {ACM}, 2021.

\bibitem[AR20]{AssadiR20}
Sepehr Assadi and Ran Raz.
\newblock Near-quadratic lower bounds for two-pass graph streaming algorithms.
\newblock In {\em {FOCS}}, pages 342--353. {IEEE}, 2020.

\bibitem[AYZ97]{AlonYZ97}
Noga Alon, Raphael Yuster, and Uri Zwick.
\newblock Finding and counting given length cycles.
\newblock {\em Algorithmica}, 17(3):209--223, 1997.

\bibitem[BBCR10]{BarakBCR10}
Boaz Barak, Mark Braverman, Xi~Chen, and Anup Rao.
\newblock How to compress interactive communication.
\newblock In {\em Proceedings of the 42nd {ACM} Symposium on Theory of Computing, {STOC} 2010, 5-8 June 2010}, pages 67--76, 2010.

\bibitem[BCD{\etalchar{+}}19]{BachrachCDELP19}
Nir Bachrach, Keren Censor{-}Hillel, Michal Dory, Yuval Efron, Dean Leitersdorf, and Ami Paz.
\newblock Hardness of distributed optimization.
\newblock In Peter Robinson and Faith Ellen, editors, {\em Proceedings of the 2019 {ACM} Symposium on Principles of Distributed Computing, {PODC} 2019, Toronto, ON, Canada, July 29 - August 2, 2019}, pages 238--247. {ACM}, 2019.

\bibitem[BCG20]{BeraCG19}
Suman~K. Bera, Amit Chakrabarti, and Prantar Ghosh.
\newblock Graph coloring via degeneracy in streaming and other space-conscious models.
\newblock In {\em 47th International Colloquium on Automata, Languages, and Programming, {ICALP} 2020, July 8-11, 2020, Saarbr{\"{u}}cken, Germany (Virtual Conference)}, pages 11:1--11:21, 2020.

\bibitem[BEO{\etalchar{+}}13]{BravermanEOPV13}
Mark Braverman, Faith Ellen, Rotem Oshman, Toniann Pitassi, and Vinod Vaikuntanathan.
\newblock A tight bound for set disjointness in the message-passing model.
\newblock In {\em 54th Annual {IEEE} Symposium on Foundations of Computer Science, {FOCS} 2013, 26-29 October, 2013, Berkeley, CA, {USA}}, pages 668--677, 2013.

\bibitem[BGKV14]{BonchiGKV14}
Francesco Bonchi, Francesco Gullo, Andreas Kaltenbrunner, and Yana Volkovich.
\newblock Core decomposition of uncertain graphs.
\newblock In Sofus~A. Macskassy, Claudia Perlich, Jure Leskovec, Wei Wang, and Rayid Ghani, editors, {\em The 20th {ACM} {SIGKDD} International Conference on Knowledge Discovery and Data Mining, {KDD} '14, New York, NY, {USA} - August 24 - 27, 2014}, pages 1316--1325. {ACM}, 2014.

\bibitem[BHNT15]{BhattacharyaHNT15}
Sayan Bhattacharya, Monika Henzinger, Danupon Nanongkai, and Charalampos~E. Tsourakakis.
\newblock Space- and time-efficient algorithm for maintaining dense subgraphs on one-pass dynamic streams.
\newblock In {\em Proceedings of the Forty-Seventh Annual {ACM} on Symposium on Theory of Computing, {STOC} 2015, Portland, OR, USA, June 14-17, 2015}, pages 173--182, 2015.

\bibitem[BJKS02]{Bar-YossefJKS02}
Ziv Bar{-}Yossef, T.~S. Jayram, Ravi Kumar, and D.~Sivakumar.
\newblock An information statistics approach to data stream and communication complexity.
\newblock In {\em 43rd Symposium on Foundations of Computer Science {(FOCS} 2002), 16-19 November 2002, Proceedings}, pages 209--218, 2002.

\bibitem[BKS02]{Bar-YossefKS02}
Ziv Bar{-}Yossef, Ravi Kumar, and D.~Sivakumar.
\newblock Reductions in streaming algorithms, with an application to counting triangles in graphs.
\newblock In {\em Proceedings of the Thirteenth Annual {ACM-SIAM} Symposium on Discrete Algorithms, January 6-8, 2002, San Francisco, CA, {USA.}}, pages 623--632, 2002.

\bibitem[BKV12]{BahmaniKV12}
Bahman Bahmani, Ravi Kumar, and Sergei Vassilvitskii.
\newblock Densest subgraph in streaming and mapreduce.
\newblock {\em {PVLDB}}, 5(5):454--465, 2012.

\bibitem[BR11]{BravermanR11}
Mark Braverman and Anup Rao.
\newblock Information equals amortized communication.
\newblock In {\em {IEEE} 52nd Annual Symposium on Foundations of Computer Science, {FOCS} 2011, October 22-25, 2011}, pages 748--757, 2011.

\bibitem[BvdBE{\etalchar{+}}22]{blikstadBEMN22}
Joakim Blikstad, Jan van~den Brand, Yuval Efron, Sagnik Mukhopadhyay, and Danupon Nanongkai.
\newblock Nearly optimal communication and query complexity of bipartite matching.
\newblock In {\em {FOCS}}, pages 1174--1185. {IEEE}, 2022.

\bibitem[CFHT20]{ChangFHT20}
Yi{-}Jun Chang, Martin Farach{-}Colton, Tsan{-}sheng Hsu, and Meng{-}Tsung Tsai.
\newblock Streaming complexity of spanning tree computation.
\newblock In {\em 37th International Symposium on Theoretical Aspects of Computer Science, {STACS} 2020, March 10-13, 2020, Montpellier, France}, pages 34:1--34:19, 2020.

\bibitem[CGMV20]{ChakrabartiGMV20}
Amit Chakrabarti, Prantar Ghosh, Andrew McGregor, and Sofya Vorotnikova.
\newblock Vertex ordering problems in directed graph streams.
\newblock In {\em Proceedings of the 2020 {ACM-SIAM} Symposium on Discrete Algorithms, {SODA} 2020, Salt Lake City, UT, USA, January 5-8, 2020}, pages 1786--1802, 2020.

\bibitem[Cha00]{Charikar00}
Moses Charikar.
\newblock Greedy approximation algorithms for finding dense components in a graph.
\newblock In Klaus Jansen and Samir Khuller, editors, {\em Approximation Algorithms for Combinatorial Optimization, Third International Workshop, {APPROX} 2000, Saarbr{\"{u}}cken, Germany, September 5-8, 2000, Proceedings}, volume 1913 of {\em Lecture Notes in Computer Science}, pages 84--95. Springer, 2000.

\bibitem[CKP{\etalchar{+}}21a]{ChenKPSSY21}
Lijie Chen, Gillat Kol, Dmitry Paramonov, Raghuvansh~R. Saxena, Zhao Song, and Huacheng Yu.
\newblock Almost optimal super-constant-pass streaming lower bounds for reachability.
\newblock In Samir Khuller and Virginia~Vassilevska Williams, editors, {\em {STOC} '21: 53rd Annual {ACM} {SIGACT} Symposium on Theory of Computing, Virtual Event, Italy, June 21-25, 2021}, pages 570--583. {ACM}, 2021.

\bibitem[CKP{\etalchar{+}}21b]{ChenKPSSY21b}
Lijie Chen, Gillat Kol, Dmitry Paramonov, Raghuvansh~R. Saxena, Zhao Song, and Huacheng Yu.
\newblock Near-optimal two-pass streaming algorithm for sampling random walks over directed graphs.
\newblock In Nikhil Bansal, Emanuela Merelli, and James Worrell, editors, {\em 48th International Colloquium on Automata, Languages, and Programming, {ICALP} 2021, July 12-16, 2021, Glasgow, Scotland (Virtual Conference)}, volume 198 of {\em LIPIcs}, pages 52:1--52:19. Schloss Dagstuhl - Leibniz-Zentrum f{\"{u}}r Informatik, 2021.

\bibitem[CKP{\etalchar{+}}23]{ChenKPSSY23}
Lijie Chen, Gillat Kol, Dmitry Paramonov, Raghuvansh~R. Saxena, Zhao Song, and Huacheng Yu.
\newblock Towards multi-pass streaming lower bounds for optimal approximation of max-cut.
\newblock In Nikhil Bansal and Viswanath Nagarajan, editors, {\em Proceedings of the 2023 {ACM-SIAM} Symposium on Discrete Algorithms, {SODA} 2023, Florence, Italy, January 22-25, 2023}, pages 878--924. {SIAM}, 2023.

\bibitem[CSWY01]{ChakrabartiSWY01}
Amit Chakrabarti, Yaoyun Shi, Anthony Wirth, and Andrew~Chi{-}Chih Yao.
\newblock Informational complexity and the direct sum problem for simultaneous message complexity.
\newblock In {\em 42nd Annual Symposium on Foundations of Computer Science, {FOCS} 2001, 14-17 October 2001}, pages 270--278, 2001.

\bibitem[CT06]{CoverT06}
Thomas~M. Cover and Joy~A. Thomas.
\newblock {\em Elements of information theory {(2.} ed.)}.
\newblock Wiley, 2006.

\bibitem[CZL{\etalchar{+}}20]{ChuZ00ZXZ20}
Deming Chu, Fan Zhang, Xuemin Lin, Wenjie Zhang, Ying Zhang, Yinglong Xia, and Chenyi Zhang.
\newblock Finding the best k in core decomposition: {A} time and space optimal solution.
\newblock In {\em 36th {IEEE} International Conference on Data Engineering, {ICDE} 2020, Dallas, TX, USA, April 20-24, 2020}, pages 685--696. {IEEE}, 2020.

\bibitem[DBS17]{DhulipalaBS17}
Laxman Dhulipala, Guy~E. Blelloch, and Julian Shun.
\newblock Julienne: {A} framework for parallel graph algorithms using work-efficient bucketing.
\newblock In Christian Scheideler and Mohammad~Taghi Hajiaghayi, editors, {\em Proceedings of the 29th {ACM} Symposium on Parallelism in Algorithms and Architectures, {SPAA} 2017, Washington DC, USA, July 24-26, 2017}, pages 293--304. {ACM}, 2017.

\bibitem[DBS18]{DhulipalaBS18}
Laxman Dhulipala, Guy~E. Blelloch, and Julian Shun.
\newblock Theoretically efficient parallel graph algorithms can be fast and scalable.
\newblock In Christian Scheideler and Jeremy~T. Fineman, editors, {\em Proceedings of the 30th on Symposium on Parallelism in Algorithms and Architectures, {SPAA} 2018, Vienna, Austria, July 16-18, 2018}, pages 393--404. {ACM}, 2018.

\bibitem[DNO19]{DobzinskiNO19}
Shahar Dobzinski, Noam Nisan, and Sigal Oren.
\newblock Economic efficiency requires interaction.
\newblock {\em Games Econ. Behav.}, 118:589--608, 2019.

\bibitem[EH66]{erdHos1966chromatic}
Paul Erd{\H{o}}s and Andr{\'a}s Hajnal.
\newblock On chromatic number of graphs and set-systems.
\newblock {\em Acta Math. Acad. Sci. Hungar}, 17(61-99):1, 1966.

\bibitem[ELM18]{EsfandiariLM18}
Hossein Esfandiari, Silvio Lattanzi, and Vahab~S. Mirrokni.
\newblock Parallel and streaming algorithms for k-core decomposition.
\newblock In Jennifer~G. Dy and Andreas Krause, editors, {\em Proceedings of the 35th International Conference on Machine Learning, {ICML} 2018, Stockholmsm{\"{a}}ssan, Stockholm, Sweden, July 10-15, 2018}, volume~80 of {\em Proceedings of Machine Learning Research}, pages 1396--1405. {PMLR}, 2018.

\bibitem[FHW12]{FrischknechtHW12}
Silvio Frischknecht, Stephan Holzer, and Roger Wattenhofer.
\newblock Networks cannot compute their diameter in sublinear time.
\newblock In Yuval Rabani, editor, {\em Proceedings of the Twenty-Third Annual {ACM-SIAM} Symposium on Discrete Algorithms, {SODA} 2012, Kyoto, Japan, January 17-19, 2012}, pages 1150--1162. {SIAM}, 2012.

\bibitem[FKM{\etalchar{+}}05]{FeigenbaumKMSZ05}
Joan Feigenbaum, Sampath Kannan, Andrew McGregor, Siddharth Suri, and Jian Zhang.
\newblock On graph problems in a semi-streaming model.
\newblock {\em Theor. Comput. Sci.}, 348(2-3):207--216, 2005.

\bibitem[FT14]{Farach-ColtonT14}
Martin Farach{-}Colton and Meng{-}Tsung Tsai.
\newblock Computing the degeneracy of large graphs.
\newblock In Alberto Pardo and Alfredo Viola, editors, {\em {LATIN} 2014: Theoretical Informatics - 11th Latin American Symposium, Montevideo, Uruguay, March 31 - April 4, 2014. Proceedings}, volume 8392 of {\em Lecture Notes in Computer Science}, pages 250--260. Springer, 2014.

\bibitem[FT16]{Farach-ColtonT16}
Martin Farach{-}Colton and Meng{-}Tsung Tsai.
\newblock Tight approximations of degeneracy in large graphs.
\newblock In {\em {LATIN} 2016: Theoretical Informatics - 12th Latin American Symposium, Ensenada, Mexico, April 11-15, 2016, Proceedings}, pages 429--440, 2016.

\bibitem[GBGL20]{GalimbertiBGL20}
Edoardo Galimberti, Francesco Bonchi, Francesco Gullo, and Tommaso Lanciano.
\newblock Core decomposition in multilayer networks: Theory, algorithms, and applications.
\newblock {\em {ACM} Trans. Knowl. Discov. Data}, 14(1):11:1--11:40, 2020.

\bibitem[GLM19]{GhaffariLM19}
Mohsen Ghaffari, Silvio Lattanzi, and Slobodan Mitrovic.
\newblock Improved parallel algorithms for density-based network clustering.
\newblock In Kamalika Chaudhuri and Ruslan Salakhutdinov, editors, {\em Proceedings of the 36th International Conference on Machine Learning, {ICML} 2019, 9-15 June 2019, Long Beach, California, {USA}}, volume~97 of {\em Proceedings of Machine Learning Research}, pages 2201--2210. {PMLR}, 2019.

\bibitem[GO13]{GuruswamiO13}
Venkatesan Guruswami and Krzysztof Onak.
\newblock Superlinear lower bounds for multipass graph processing.
\newblock In {\em Proceedings of the 28th Conference on Computational Complexity, {CCC} 2013, K.lo Alto, California, USA, 5-7 June, 2013}, pages 287--298, 2013.

\bibitem[HSSW12]{HalldorssonSSW12}
Magn{\'{u}}s~M. Halld{\'{o}}rsson, Xiaoming Sun, Mario Szegedy, and Chengu Wang.
\newblock Streaming and communication complexity of clique approximation.
\newblock In {\em Automata, Languages, and Programming - 39th International Colloquium, {ICALP} 2012, Warwick, UK, July 9-13, 2012, Proceedings, Part {I}}, pages 449--460, 2012.

\bibitem[IKL{\etalchar{+}}12]{IvanyosKLSW12}
G{\'{a}}bor Ivanyos, Hartmut Klauck, Troy Lee, Miklos Santha, and Ronald de~Wolf.
\newblock New bounds on the classical and quantum communication complexity of some graph properties.
\newblock In {\em {FSTTCS}}, volume~18 of {\em LIPIcs}, pages 148--159. Schloss Dagstuhl - Leibniz-Zentrum f{\"{u}}r Informatik, 2012.

\bibitem[JKS03]{JayramKS03}
T.~S. Jayram, Ravi Kumar, and D.~Sivakumar.
\newblock Two applications of information complexity.
\newblock In {\em {STOC}}, pages 673--682. {ACM}, 2003.

\bibitem[KBVT15]{KhaouidBST15}
Wissam Khaouid, Marina Barsky, S.~Venkatesh, and Alex Thomo.
\newblock K-core decomposition of large networks on a single {PC}.
\newblock {\em Proc. {VLDB} Endow.}, 9(1):13--23, 2015.

\bibitem[Kla00]{klauck2000quantum}
Hartmut Klauck.
\newblock On quantum and probabilistic communication: Las vegas and one-way protocols.
\newblock In {\em 32nd Annual ACM Symposium on Theory of Computing ({STOC})}, pages 644--651, 2000.

\bibitem[KN97]{KushilevitzN97}
Eyal Kushilevitz and Noam Nisan.
\newblock {\em Communication complexity}.
\newblock Cambridge University Press, 1997.

\bibitem[KPP16]{KopelowitzPP16}
Tsvi Kopelowitz, Seth Pettie, and Ely Porat.
\newblock Higher lower bounds from the 3sum conjecture.
\newblock In Robert Krauthgamer, editor, {\em Proceedings of the Twenty-Seventh Annual {ACM-SIAM} Symposium on Discrete Algorithms, {SODA} 2016, Arlington, VA, USA, January 10-12, 2016}, pages 1272--1287. {SIAM}, 2016.

\bibitem[KPSY23]{KolPSY23}
Gillat Kol, Dmitry Paramonov, Raghuvansh~R. Saxena, and Huacheng Yu.
\newblock Characterizing the multi-pass streaming complexity for solving boolean csps exactly.
\newblock In Yael~Tauman Kalai, editor, {\em 14th Innovations in Theoretical Computer Science Conference, {ITCS} 2023, January 10-13, 2023, MIT, Cambridge, Massachusetts, {USA}}, volume 251 of {\em LIPIcs}, pages 80:1--80:15. Schloss Dagstuhl - Leibniz-Zentrum f{\"{u}}r Informatik, 2023.

\bibitem[LJS19]{LiuJS19}
Yang~P. Liu, Arun Jambulapati, and Aaron Sidford.
\newblock Parallel reachability in almost linear work and square root depth.
\newblock In David Zuckerman, editor, {\em 60th {IEEE} Annual Symposium on Foundations of Computer Science, {FOCS} 2019, Baltimore, Maryland, USA, November 9-12, 2019}, pages 1664--1686. {IEEE} Computer Society, 2019.

\bibitem[LSY{\etalchar{+}}22]{LiuSYDS22}
Quanquan~C. Liu, Jessica Shi, Shangdi Yu, Laxman Dhulipala, and Julian Shun.
\newblock Parallel batch-dynamic algorithms for k-core decomposition and related graph problems.
\newblock In Kunal Agrawal and I{-}Ting~Angelina Lee, editors, {\em {SPAA} '22: 34th {ACM} Symposium on Parallelism in Algorithms and Architectures, Philadelphia, PA, USA, July 11 - 14, 2022}, pages 191--204. {ACM}, 2022.

\bibitem[LWSX18]{LiWSX18}
Chao Li, Li~Wang, Shiwen Sun, and Chengyi Xia.
\newblock Identification of influential spreaders based on classified neighbors in real-world complex networks.
\newblock {\em Appl. Math. Comput.}, 320:512--523, 2018.

\bibitem[LZZ{\etalchar{+}}19]{LiZZQZL19}
Conggai Li, Fan Zhang, Ying Zhang, Lu~Qin, Wenjie Zhang, and Xuemin Lin.
\newblock Efficient progressive minimum k-core search.
\newblock {\em Proc. {VLDB} Endow.}, 13(3):362--375, 2019.

\bibitem[MB83]{MatulaB83}
David~W. Matula and Leland~L. Beck.
\newblock Smallest-last ordering and clustering and graph coloring algorithms.
\newblock {\em J. {ACM}}, 30(3):417--427, 1983.

\bibitem[McG14]{McGregor14}
Andrew McGregor.
\newblock Graph stream algorithms: a survey.
\newblock {\em {SIGMOD} Rec.}, 43(1):9--20, 2014.

\bibitem[MN20]{MukhopadhyayN20}
Sagnik Mukhopadhyay and Danupon Nanongkai.
\newblock Weighted min-cut: sequential, cut-query, and streaming algorithms.
\newblock In {\em {STOC}}, pages 496--509. {ACM}, 2020.

\bibitem[MTVV15]{McGregorTVV15}
Andrew McGregor, David Tench, Sofya Vorotnikova, and Hoa~T. Vu.
\newblock Densest subgraph in dynamic graph streams.
\newblock In {\em Mathematical Foundations of Computer Science 2015 - 40th International Symposium, {MFCS} 2015, Milan, Italy, August 24-28, 2015, Proceedings, Part {II}}, pages 472--482, 2015.

\bibitem[NW91]{NisanW91}
Noam Nisan and Avi Wigderson.
\newblock Rounds in communication complexity revisited.
\newblock In {\em Proceedings of the 23rd Annual {ACM} Symposium on Theory of Computing, May 5-8, 1991, New Orleans, Louisiana, {USA}}, pages 419--429, 1991.

\bibitem[NW93]{NisanW93}
Noam Nisan and Avi Wigderson.
\newblock Rounds in communication complexity revisited.
\newblock {\em {SIAM} J. Comput.}, 22(1):211--219, 1993.

\bibitem[PRV99]{PonzioRV99}
Stephen Ponzio, Jaikumar Radhakrishnan, and Srinivasan Venkatesh.
\newblock The communication complexity of pointer chasing: Applications of entropy and sampling.
\newblock In {\em Proceedings of the Thirty-First Annual {ACM} Symposium on Theory of Computing, May 1-4, 1999, Atlanta, Georgia, {USA}}, pages 602--611, 1999.

\bibitem[PS84]{PapadimitriouS84}
Christos~H. Papadimitriou and Michael Sipser.
\newblock Communication complexity.
\newblock {\em J. Comput. Syst. Sci.}, 28(2):260--269, 1984.

\bibitem[Raz92]{Razborov92}
Alexander~A. Razborov.
\newblock On the distributional complexity of disjointness.
\newblock {\em Theor. Comput. Sci.}, 106(2):385--390, 1992.

\bibitem[RSW18]{RubinsteinSW18}
Aviad Rubinstein, Tselil Schramm, and S.~Matthew Weinberg.
\newblock Computing exact minimum cuts without knowing the graph.
\newblock In {\em {ITCS}}, volume~94 of {\em LIPIcs}, pages 39:1--39:16. Schloss Dagstuhl - Leibniz-Zentrum f{\"{u}}r Informatik, 2018.

\bibitem[RW13]{RodittyW13}
Liam Roditty and Virginia~Vassilevska Williams.
\newblock Fast approximation algorithms for the diameter and radius of sparse graphs.
\newblock In Dan Boneh, Tim Roughgarden, and Joan Feigenbaum, editors, {\em Symposium on Theory of Computing Conference, STOC'13, Palo Alto, CA, USA, June 1-4, 2013}, pages 515--524. {ACM}, 2013.

\bibitem[SGJ{\etalchar{+}}13]{SariyuceGJWC13}
Ahmet~Erdem Sariy{\"{u}}ce, Bugra Gedik, Gabriela Jacques{-}Silva, Kun{-}Lung Wu, and {\"{U}}mit~V. {\c{C}}ataly{\"{u}}rek.
\newblock Streaming algorithms for k-core decomposition.
\newblock {\em Proc. {VLDB} Endow.}, 6(6):433--444, 2013.

\bibitem[Yao79]{Yao79}
Andrew~Chi{-}Chih Yao.
\newblock Some complexity questions related to distributive computing (preliminary report).
\newblock In {\em Proceedings of the 11h Annual {ACM} Symposium on Theory of Computing, April 30 - May 2, 1979, Atlanta, Georgia, {USA}}, pages 209--213, 1979.

\bibitem[Yeh20]{Yehudayoff20}
Amir Yehudayoff.
\newblock Pointer chasing via triangular discrimination.
\newblock {\em Comb. Probab. Comput.}, 29(4):485--494, 2020.

\end{thebibliography}
